\documentclass[format=acmsmall, review=false,nonacm]{acmart}
\usepackage{booktabs} 
\usepackage[ruled]{algorithm2e} 

\SetAlFnt{\small}
\SetAlCapFnt{\small}
\SetAlCapNameFnt{\small}
\SetAlCapHSkip{0pt}
\IncMargin{-\parindent}

\theoremstyle{plain}
\newtheorem{theorem}{Theorem}[section]

\newtheorem{lemma}[theorem]{Lemma}

\newtheorem{claim}[theorem]{Claim}

\newtheorem{corollary}[theorem]{Corollary}

\newtheorem{assumption}[theorem]{Assumption}

\newtheorem*{theorem*}{Theorem}
\newtheorem*{lemma*}{Lemma}
\newtheorem{claim*}{Claim}

\theoremstyle{definition}
\newtheorem*{definition*}{Definition}
\newtheorem{definition}[theorem]{Definition}

\newtheorem{example}[theorem]{Example}

\theoremstyle{remark}
\newtheorem{remark}[theorem]{Remark}

\newtheorem*{remark*}{Remark}
\newtheorem*{property*}{Property}

\newtheorem*{problem*}{Problem formulation}

\newcommand{\Emph}[1]{\textbf{#1}}

\newcommand{\parans}[1]{{\left(#1 \right)}}
\newcommand{\paranm}[1]{{\left\{#1 \right\}}}
\newcommand{\paranl}[1]{{\left[#1 \right]}}

\newcommand{\separator}{
  \begin{center}
    \rule{\columnwidth}{0.3mm}
  \end{center}
}

\newcommand{\beq}{\begin{eqnarray*}}
\newcommand{\eeq}{\end{eqnarray*}}
\newcommand{\beqn}{\begin{eqnarray}}
\newcommand{\eeqn}{\end{eqnarray}}
\newcommand{\bemn}{\begin{multiline}}
\newcommand{\eemn}{\end{multiline}}

\def\N{\mathbb{N}}
\def\R{\mathbb{R}}

\usepackage{mathtools}

\usepackage{bbm}
\usepackage{xcolor}
\usepackage{multirow}
\usepackage{makecell}
\usepackage{nicefrac}
\newtheorem{innercustomthm}{Informal Statement of Theorem}
\newenvironment{customthm}[1]
  {\renewcommand\theinnercustomthm{#1}\innercustomthm}
  {\endinnercustomthm}
\newcommand{\xhdr}[1]{\vspace{2mm} \noindent{\bf #1}}

\newcommand{\given}{\mid}

\newcommand{\ie}{{\em i.e.,~\xspace}}

\newcommand{\eg}{{\em e.g.,~\xspace}}

\usepackage{color-edits}
\addauthor{ss}{blue}      
\addauthor{kv}{red}   

\renewcommand{\paragraph}[1]{\smallskip \noindent{\bf #1.} }

\newcommand{\IND}{\mathbbm{1}}

\DeclareMathOperator*{\argmax}{argmax}
\DeclareMathOperator*{\argmin}{argmin}

\newcommand{\opt}{\textsc{Opt}}

\renewcommand{\Pr}[1]{{\mathbb{P} \left[ #1 \right]}}

\newcommand{\Ex}[1]{{\mathbb{E} \left[ #1 \right]}}
\newcommand{\Exu}[2]{\ensuremath{\mathbb{E}_{#1}\left[#2\right]}}

\newcommand{\eps}{\varepsilon}

\newcommand{\Om}{\Omega}
\newcommand{\om}{\omega}
\newcommand{\omb}{\bar{\omega}}
\newcommand{\var}{\mathrm{Var}}

\newcommand{\poa}{\textsc{PoA}}
\newcommand{\pos}{\textsc{PoS}}
\title[Delegating to Multiple Agents]{Delegating to Multiple Agents}
\titlenote{
We thank the anonymous reviewers of EC'23 for insightful feedback.
\newline
\emph{Version history.} May 2023 (accepted to EC'23). May 2023 (changed title).}

\author{MohammadTaghi Hajiaghayi}
\affiliation{
  \institution{University of Maryland, College Park}
  \city{College Park}
  \country{USA}}
\email{hajiagha@umd.edu }

\author{Keivan Rezaei}
\affiliation{
  \institution{University of Maryland, College Park}
  \city{College Park}
  \country{USA}}
\email{krezaei@umd.edu }

\author{Suho Shin}
\affiliation{
  \institution{University of Maryland, College Park}
  \city{College Park}
  \country{USA}}
\email{suhoshin@umd.edu}

\begin{abstract}
We consider a \emph{multi-agent delegation} mechanism without money.
In our model, given a set of agents, each agent has a fixed number of solutions which is exogenous to the mechanism, and privately sends a signal, \eg a subset of solutions, to the principal.
Then, the principal selects a final solution based on the agents' signals.
In stark contrast to single-agent setting by Kleinberg and Kleinberg (EC'18) with an approximate \emph{Bayesian mechanism}, we show that there exists efficient approximate \emph{prior-independent mechanisms} with both information and performance gain, thanks to the competitive tension between the agents.
Interestingly, however, the amount of such a compelling power significantly varies with respect to the information available to the agents, and the degree of correlation between the principal's and the agent's utility.
Technically, we conduct a comprehensive study on the multi-agent delegation problem and derive several results on the approximation factors of Bayesian/prior-independent mechanisms in complete/incomplete information settings.
As a special case of independent interest, we obtain comparative statics regarding the number of agents which implies the dominance of the multi-agent setting ($n \ge 2$) over the single-agent setting ($n=1$) in terms of the principal's utility.
We further extend our problem by considering an examination cost of the mechanism and derive some analogous results in the complete information setting.
\end{abstract}

\begin{document}

\maketitle

\section{Introduction}\label{sec:intro}
In many real-world situations, a decision-maker (called "principal") is encountered with a difficult task that cannot be solved by herself.\footnote{Throughout the paper we use feminine pronouns to denote the principal and masculine pronouns to denote the agents.}
To cope with it, she delegates the process of finding a solution to several agents who have the ability to solve the task, and then examines agents' proposed solutions to take the best one.
An interesting tension arises when the utility of the agents are not aligned with that of the principal.
Since the agents are self-interested to maximize their own payoff, they may not propose the solutions for sake of the principal. 

As a concrete example, consider a research proposal selection process by a grant agency. In this problem, a {\em principal} (\eg the grant agency) introduces a topic for proposals  and asks each {\em agent} (\eg a principal investigator) to send a private proposal, or a set of proposals, based on a set of his research directions (possibly coming from some distribution) related to the topic introduced by the principal.
Each proposal has some valuation for the agent and some valuation for the principal which can be different (\eg a proposal may be more aligned with a research direction by the agent but less aligned with the principal's interests). 
The goal of the principal is to {\em maximize} own valuation of the finally selected proposal while the goal of each agent is to maximize his utility which is his valuation for a proposal selected by the principal and zero otherwise.

In this example, the principal faces a problem of \emph{mechanism design without money} to incentivize the agents to behave more in favor of her.
In practice, the setting here can be an {\em incomplete information} game in the sense that though an agent knows the topic introduced by the principal and general research topics of the other agents (\ie distributions of possible proposals), he may not know the exact topics that other agents consider for their proposals.
In addition, the principal may have an {\em examination cost} (\eg travel cost for time-consuming site visits) for each proposal which should be accounted for in the proposal selection process as well. 
A fundamental question in this example is how should the principal design a mechanism to maximize her own utility function.

%

Initiated by \citet{holmstrom1980theory}, there exists a rich line of literature studying the theory of \emph{delegation without monetary transfer}.
The restriction that the principal cannot commit to a contingent monetary transfer in the mechanism imparts a restrictive structure on feasible mechanisms in which the mechanism simply commits to a set of acceptable solutions, and only decide whether or not to accept the proposed solutions.
\citet{holmstrom1980theory} study the optimality of these simple mechanisms and the characterizations of the regimes for these mechanisms to be optimal.
Subsequent works including \citet{alonso2008optimal} and \citet{armstrong2010model} develop the theory of delegation in this context and broaden its scope in fruitful direction.

Recently, with a lens of computer science, \citet{kleinberg2018delegated}  study a problem of designing approximate mechanisms in the \emph{delegated search} problem.\footnote{They study two types of model, one of which is similar to delegated choice without search cost by \cite{armstrong2010model}, and the other is that with search (sampling) cost. Our model can be framed as either of a delegated search problem with exogenous solutions, or delegated choice problem which delegates the choice between given alternatives.}
They consider a simple mechanism so-called a {\em single-proposal} mechanism, which consists of single-round interaction between the principal and the agent.
Remarkably, it achieves a $2$-approximation compared to an ideal scenario in which the principal can directly select the best solution among the entire set of agents' solutions.
Their analysis is based on a novel observation regarding a connection between the dynamics of the single-proposal mechanism and the dynamics of threshold-based algorithms in \emph{prophet inequalities} (see \citet{samuel1984comparison},
\citet{krengel1987prophet},
\citet{hajiaghayi2007automated},
\citet{alaei2012online}, \citet{abolhassani2017beating}, \citet{correa2017posted} for more reference to prophet inequalities).

As we discussed in the above example, there exists a variety of real-world situations, spanning economical applications to modern intelligent systems, such that the principal needs to delegate the task to multiple agents with misaligned payoff.
Surprisingly, however, the problem of delegating to \emph{multiple}  agents in the framework of mechanism design without money has received less attention in the community.
The recent work by \citet{gan2022optimal}, preceded by \citet{alonso2014resource}, study the optimal delegation mechanism with two agents in the context of mechanism design without money, but their model is significantly different from ours, and also is restricted to small number of agents with specific utility functions.\footnote{We provide the detailed discussion in Section~\ref{sec:related_work}.}

In this context, we introduce the \emph{multi-agent delegation} problem, and conduct a comprehensive study on the approximation factors of various mechanisms without monetary transfer.
Technically, we aim to design a multi-agent delegation mechanism that improves upon the guarantees of a single-agent mechanism by \citet{kleinberg2018delegated}.
Surprisingly, different from the single-agent mechanism which assumes the prior knowledge on the agent's distribution to derive a constant approximation factor, \ie \Emph{Bayesian mechanism} (BM), we show that we can construct an efficient mechanism even without those prior knowledge, \ie by \Emph{prior-independent mechanism} (PIM), which is mainly due to the competitive tension arise from the multi-agent extension.
Notably, however, we further observe that the amount of such a compelling power from competition differs with respect to the information available to the agents, and the degree of correlation between the principal's and the agent's utility.
Overall, we comprehensively analyze the approximation factors of BM and PIM in both the \Emph{complete} and \Emph{incomplete information} settings (in terms of the agents' information).
We note that our main technical results lie on the \emph{analysis of approximation factors of PIM}.
To highlight a few of the main results,
\vspace{-0.3mm}
\begin{enumerate}
    \item In the complete information setting, we show that there exists a PIM such that the principal's expected utility is bounded below by the \emph{first-best solution} (in the principal's view) of the \emph{second-best agent}. We further quantify this approximation factor by analyzing the upper and lower bound of the principal's expected utility in various regimes. This tends to asymptotically optimal in certain regimes, \ie the gains from multi-agent delegation is significant. En route to this analysis, we present several technical results regarding the order statistics of i.i.d. random variables, which might be of independent interest.
    \item In the incomplete information setting, we show that under symmetric agents with uniform distribution $U[0,1]$ and the assumption of independent utility, we can construct a PIM in which there exists a $(1-e^{-n^2/2(n-1)^2})$-approximate \emph{Bayes Nash equilibrium} such that the principal's expected utility becomes exactly optimal, where $1-e^{-n^2/2(n-1)^2}$ converges to $1-\nicefrac{1}{\sqrt{e}} \cong 0.3935$ as the number of agents $n$ increases.
    We numerically verify that this guarantee is almost tight. Using a similar technique, we analyze that there exists a \emph{Bayes correlated equilibrium} with a multiplicative approximation factor of $e^{n^2/2(n-1)^2} \cong \sqrt{e}$.
    \item We further show that the assumptions of symmetric agents and independent utility are necessary. More formally, without these assumptions, we can construct a problem instance such that there exists a \emph{Bayesian Nash equilibrium} in which the agents propose their worst solutions for the principal, which may induce arbitrarily bad principal's utility.
\end{enumerate}
\vspace{-0.3mm}
To select the best solution among the solutions from multiple agents, however, the principal suffers a cost of \emph{examining the candidates}.
In this context, we extend our problem setup and further study a trade-off between such an examination cost and an efficiency of mechanisms, and provide an analogous result in the complete information setting.


The remainder of the paper is organized as follows.
In Section~\ref{sec:overview}, we summarize overall results and the techniques therein.
Section~\ref{sec:related_work} provides detailed discussions on the related work.
Section \ref{sec:setup} introduces the formal problem setup.
In Section \ref{section:no_cost} and \ref{section:no_cost_incomp}, we provide the results in the complete and incomplete information setting, respectively.
Finally in Section \ref{sec:cost}, we introduce the examination cost and some results in this extension.




\section{Overview of results}\label{sec:overview}
Before presenting the summary of our results, we briefly discuss our problem setup and some preliminaries.
We refer to Section~\ref{sec:setup} for a more formal problem setup.
In \emph{multi-agent delegation} problem,
there exists a principal who wants to find a solution for a task by
designing a mechanism that delegates the task to $n$ agents.
Each agent $i$ observes $k_i$ solutions by sampling from his own distributions, where both the number of solutions and the distributions are \emph{exogenous} to the mechanism.
These exogenous parameters can be interpreted as the intrinsic level of effort and skill of the agent.
Afterward, each agent proposes a subset of observed solutions (signal) to the principal.
Eventually, the mechanism selects a single solution, \emph{winner}, from the proposed ones.
If the solution $\omega$ is selected as the winner,
the principal receives the utility of $x(\omega)$.
If $\omega$ is proposed by agent $i$, he gets the utility of $y_i(\omega)$ while other agents realize the utility of zero.
Given the mechanism, each agent $i$ determines a (possibly mixed) strategy $\sigma_i$ to maximize own utility, which maps the observed solutions and information regarding others' solutions to the signal, and plays an action based on the strategy.
Importantly, we impose no restriction on the utility function $x(\cdot)$ and $y_i(\cdot)$, \ie they might be misaligned, which brings an interesting trade-off in designing the mechanism.

We introduce technical notations to summarize our results.
Denote by $\omb_i = \{\om_{i,1}, \om_{i,2}, \ldots \om_{i,k_i}\}$  the agent $i$'s observed solutions for $i \in [n]$.\footnote{We use $[n]$ to denote $\{1,2,\ldots, n\}$ for $n \in \N$ throughout the paper.}
Define $X_{i,j} = x(\om_{i,j})$ for $i \in [n]$ and $j \in [k_i]$.
Sort $X_{i,j}$ for $j \in [k_i]$ in decreasing order, and denote by $X_{i,(1)} \ge X_{i,(2)} \ge \ldots \ge X_{i,(k_i)}$.
We denote by $X_{i,(1)}$ the first-best solution of agent $i$.
Again, sort $X_{i,(1)}$ for $i \in [n]$ in decreasing order, and obtain $X_{(1)} \ge X_{(2)} \ge \ldots \ge X_{(n)}$.
Let the owner of $X_{(i)}$ be the $i$-th best agent.
We often overwrite $X_{(1)}$ by $X_{\max}$.
Our benchmark of the mechanism is $\Ex{X_{(1)}}$, \ie the case when the principal can select the best one by directly accessing all the observed solutions.
Let $f_{M,\sigma}$ be an interim allocation function which specifies the outcome under the mechanism $M$ and some strategies $\sigma$, given a realization of the solutions.
Then, given some strategies $\sigma$ of the agents, we say that mechanism $M$ is $(\rho, \gamma)$-approximate if 
\begin{align*}
\rho \Ex{f_{M,\sigma}} + \gamma \geq \Ex{X_{(1)}}.
\end{align*}
We denote by $\rho$ and $\gamma$ the multiplicative and additive approximation factors, respectively.
We say that $M$ has multiplicative \emph{price of anarchy} ($\poa_m$) of $\rho$ if it is $(\rho,0)$-approximate under \emph{any Nash equilibrium} of agents' strategies, and it has multiplicative \emph{price of stability} ($\pos_m$) of $\rho$ if there exists \emph{a Nash equilibrium} such that $M$ is $(\rho,0)$-approximate. We similarly define $\poa_a$ and $\pos_a$ to denote the price of anarchy/stability for $(1,\gamma)$-approximation, \ie  additive approximation factor.\footnote{See Definition~\ref{def:apx-mech} and \ref{def:poa-pos} for more details.}


\subsection{Contributions and techniques}
We here summarize our key results and the major techniques therein.
Basically, we analyze two types of mechanisms: (1) Bayesian mechanism (BM) with knowledge of the prior distributions, and (2) prior-independent mechanism (PIM) which does not know the distributions in advance.
We further consider complete and incomplete information settings.
In the complete information setting, each agent knows the others' observed solutions while in the incomplete information setting, he only knows the distributions from which the others sample, but not the realized solutions.

We begin by showing that any mechanism can be reduced to a \emph{multi-agent single-proposal} mechanism (MSPM, Definition~\ref{def:mspm}) in Theorem \ref{thm:multiagent_revel}.
Essentially, this theorem implies that we can narrow down our scope of interest to the space of the MSPMs.
We note that the MSPM can be either BM or PIM depending on whether the mechanism uses the knowledge on the distributions or not.
All the positive results we present below will be built upon MSPMs.

\paragraph{General result in the complete information setting}
We start with the general result on MSPM.
\begin{customthm}{\ref{theorem:multi-agent-unlimited-budget}}
\label{thm:1.1_informal}
    In the complete information setting, for any $\tau \ge 0$, there exists a MSPM with $\poa_a$ of
	\begin{align*}
		\Ex{X_{(1)} - X_{(2)}\IND[X_{(2)}\ge \tau] - \tau \IND[X_{(1)} \ge \tau > X_{(2)}]}
        .
	\end{align*}
\end{customthm}
Intuitively, the first-best agent needs to guarantee that he will be selected by the principal while maximizing his own utility.
This motivates him to propose a solution that yields the principal's utility at least that of the first-best solution of the second-best agent.
The proof is built upon such simple intuition, however, it is still tricky to characterize arbitrary (possibly mixed) Nash equilibrium of the agents, and the formal proof can be found in Appendix~\ref{pf_theorem:multi-agent-unlimited-budget}.
This result implies that the additive approximation factor is characterized by the difference between $X_{(1)}$ and $X_{(2)}$, \ie the difference of first-best solution between the first-best and second-best agents.
On top of it, the principal can also smartly define the threshold $\tau$ using the prior information, and further improve the quality of the candidates, although it may decrease the probability of the acceptance.

\paragraph{BM in the complete information setting}
From the general result, we obtain the following result on BM which recovers the single-agent guarantee in \citet{kleinberg2018delegated}.
\begin{customthm}{\ref{thm:recover} and \ref{thm:lowerbound_comp_worst}}\label{thm:recover_informal}
        In the complete information setting, there exists a BM with $\poa_m$ of $2$.
    In addition, this multiplicative approximation factor is tight, \ie there exists a problem instance that no BM can achieve $\pos_m < 2-\eps$.
\end{customthm}

\begin{table}[]
\centering
\resizebox{11cm}{!}{
\begin{tabular}{|c|c|c|c|}
\hline
\textbf{Information}                                                                 & \textbf{Bound} & \textbf{General}   & \textbf{Symmetric $\alpha$-decaying}
\\ \hline
\multirow{2}{*}{\textbf{\begin{tabular}[c]{@{}c@{}}\\ Complete\end{tabular}}}   & Lower     &$\pos_m \ge 2$ \quad ( Thm \ref{thm:lowerbound_comp_worst}) & open 
\\ \cline{2-4} 
& Upper&$\poa_m \le 2$ \quad (Thm \ref{thm:recover})& $\poa_m \le \parans{\frac{\alpha nk}{\alpha nk+1}}^{\frac{-1}{\alpha nk}}$  (Thm \ref{thm:informed}) \\ \hline
\multirow{2}{*}{\textbf{\begin{tabular}[c]{@{}c@{}}\\ Incomplete\end{tabular}}} & Lower & $\pos_m \ge 2$ \quad (Cor \ref{thm:recover_bayes})& open 
\\ \cline{2-4} & Upper & $\poa_m \le 2$ \quad (Cor \ref{thm:recover_bayes})& $\poa_m \le \parans{\frac{\alpha nk}{\alpha nk+1}}^{\frac{-1}{\alpha nk}}$  (Cor \ref{thm:recover_bayes}) \\ \hline
\end{tabular}
}
\caption{Summary of results on Bayesian mechanism.  \emph{Symmetric} refers to the symmetric agents and \emph{$\alpha$-decaying} denotes the distributions such that its c.d.f. $F$ satisfies $F(x) \le x^{\alpha}$ for some $\alpha > 0$. Lower bounds for the symmetric setting remain as open problems.}
\vspace{-5mm}
\label{table:1}
\end{table}

The upper bound follows from the reduction from multi-agent setting to single-agent setting, where we in fact obtain stronger comparative statics such that the multi-agent delegation is always superior to the single-agent delegation.
We discuss a more detailed comparison between the single-agent and multi-agent settings in Appendix~\ref{sec:compare_mul_sin}.
The lower bound is obtained by introducing a concept of \emph{super-agent} equipped with the standard worst-case instance in prophet inequalities.
This worst-case instance implies that there is no merit in delegating the task to multiple agents with highly \emph{asymmetric} levels of skill.

We then naturally consider a \emph{symmetric} case (see Definition \ref{def:symmetric}) when the agents are equipped with the same parameters, \ie they have similar levels of skill and effort.
In this case, for a certain class of distributions, we show that the multiplicative approximation factor can be subconstant for BM, \ie converges to $1$ as $n$ (the num. of agents) or $k$ (the num. of solutions per agent) increases.
\begin{customthm}{\ref{thm:informed}}\label{thm:subconstant_informal}
    In the complete information setting, consider the symmetric agents with $X_{i, j} \sim D$ where $D$ is a distribution supported on $[0,1]$ and the cumulative distribution function (c.d.f.) of $F(x)$ satisfies $F(x) \le x^{\alpha}$ for some $\alpha > 0$.
    Then, there exist a BM with $\poa_m$ of $\parans{\frac{\alpha nk}{\alpha nk+1}}^{-\frac{1}{\alpha nk}}$. Hence, the principal's expected utility converges to $\Ex{X_{\max}}$ as $n$ or $k$ increases.
\end{customthm}
To obtain the subconstant bound above, we use threshold-based eligible sets that are tighter than the ones used in the prophet inequalities.
We present the summary of results for BM in Table~\ref{table:1}.

%

\paragraph{PIM in the complete information setting}
For PIM, we observe that Theorem~\ref{thm:1.1_informal} yields an additive approximation factor of $\Ex{X_{(1)}- X_{(2)}}$ by setting $\tau = 0$ (see Corollary~\ref{cor:infbud_notau}).
This indeed requires no prior knowledge as it accepts any solution.
On the other hand, we argue that this guarantee can be arbitrarily bad if we do not restrict our problem instances to symmetric agents.
\begin{customthm}{\ref{thm:comp_obl_neg}}\label{thm:arbbad_informal}
    In the complete information setting, suppose that principal's utility function lies in $[0,L]$ for some $L>0$.
    Then for any $\eps > 0$, there exists a problem instance such that no PIM can achieve $\pos_a \le L-\eps$.
\end{customthm}
The proof is based on a construction worst-case instance with a super-agent similar to the lower bound under the general setting in BM.

We thus consider the symmetric agents and derive the following result.
\begin{customthm}{\ref{thm:comp_sym_general}}
    In the complete information setting, consider symmetric agents with the principal's utility function lies in $[0,L]$ for some $L>0$.
    Then, there exists a PIM with $\poa_a$ of $(1-\frac{1}{n})^{n-1}$.
\end{customthm}
Note that $(1-\frac{1}{n})^{n-1} \le 1/2$ for all $n \ge 2$, and decreases to $\nicefrac{1}{e}$ as $n$ increases.
This constant approximation is a big improvement over the arbitrarily bad approximation in the general setting.

Next, we improve the approximation factors of PIM to be subconstant with respect to some canonical family of distributions.
At the heart of this analysis, we come up with several new results (see Lemma \ref{lm:order_stat_diff}, \ref{lm:mhr_preserving}, \ref{lm:order_stat_diff_restricted}, and \ref{lm:id_preserving}) in order statistics\footnote{Given $n$ i.i.d. samples, $r$-th order statistics refer to the $r$-th smallest value among the i.i.d. samples.}, which might be of independent interest.
Detailed discussions are presented in Section \ref{sec:comp_pim}.
\begin{customthm}{\ref{thm:mhr_upperbound} and \ref{thm:decaying_upperbound}}
    In the complete information setting, consider symmetric agents such that $X_{i, j} \sim D$ where $D$ is a distribution with a probability distribution function (p.d.f.) of $f(x)$.
    Let $D_{\max}$ be the distribution of $X_{i,(1)}$ for any $i \in [n]$. 
    Denote by $\var(\cdot)$ the variance of a distribution.
    Then, the following holds.
    \begin{enumerate}
	\item If $D$ has MHR (see Section~\ref{sec:comp_pim} for definition), then there exists a PIM with $\poa_a$ of $\sqrt{\frac{4\var(D_{\max})}{3}}$. In case of $D = U[0,1]$, we have $\sqrt{\frac{4\var(D_{\max})}{3}} = \sqrt{\frac{4k}{3(k+1)^2(k+2)}} \approx \frac{1.155}{k}$.
	\item If $f(x)$ is nondecreasing on $x$, then there exists a PIM with $\poa_a$ of $\frac{\sqrt{12\var(D_{\max})}}{n+1}$. In case of $D = U[0,1]$, we have $\frac{\sqrt{12\var(D_{\max})}}{n+1} = \frac{1}{n+1}\sqrt{\frac{12k}{(k+1)^2(k+2)}} \approx \frac{3.464}{k(n+1)}$.
    \end{enumerate}
\end{customthm}

\paragraph{Incomplete information setting}
In the incomplete information setting, we first discuss that the results on the BM, Theorem~\ref{thm:recover_informal} and Theorem~\ref{thm:subconstant_informal}, still carry over in the incomplete information setting.
In addition, without the assumption of symmetric agents, we still have an arbitrarily bad lower bound in the approximation ratio for any PIM, \ie Theorem~\ref{thm:arbbad_informal} still holds.
This justifies our assumption of symmetric agents in the incomplete information setting as well.
Interestingly, however, assuming the symmetric agents is not sufficient to derive a positive result.
\begin{customthm}{\ref{thm:pim_neg_incomp}}
    In the incomplete information setting, under symmetric agents with $X_{i,j} \sim U[0,1]$, there is a problem instance where no PIM has $\poa_a < \frac{k-1}{k+1} - \sqrt{\frac{\log n}{2k+3}}$.
\end{customthm}
Note that by making $k$ arbitrarily large give $n$, RHS can be arbitrarily close to $1$.
The proof is based on a construction of strongly negatively correlated utility between the principal and the agents, preceded by some algebraic manipulations.
Interestingly in this problem instance, each agent's utility decreases by playing for sake of the principal.
Hence, each agent settles with a solution that gives the agent largest ex-post utility (and thus smallest principal's utility), even though it yields a small probability to be selected.
We present the result for uniform distribution for simplicity, but we in fact provide a more general result.

\begin{table}
\centering
\resizebox{14cm}{!}{
\begin{tabular}{|c|c|c|c|c|c|c|}
\hline
\textbf{Information}  & \textbf{Bound} & \textbf{General}  & \textbf{Symmetric (Sym)}   & \textbf{Sym MHR}  & \textbf{Sym Inc}  & \textbf{Sym Ind}            
\\ \hline
\multirow{2}{*}{\textbf{\begin{tabular}[c]{@{}c@{}}\\ Complete\end{tabular}}}   & Lower      & \begin{tabular}[c]{@{}c@{}}$\pos_a \ge L$\\ Thm \ref{thm:comp_obl_neg}\end{tabular} &
Open
& Open & Open   & Remark \ref{remark:just_table2}  
\\ \cline{2-7} 
& Upper      & Cor \ref{cor:infbud_notau}   & \begin{tabular}[c]{@{}c@{}}$\poa_a \le L\parans{1-\frac{1}{n}}^{n-1}\cong \frac{L}{e}$\\ Thm  \ref{thm:comp_sym_general}\end{tabular}  & \begin{tabular}[c]{@{}c@{}}$\poa_a \le \frac{1.15}{k}$\\ Thm \ref{thm:mhr_upperbound}\end{tabular} & \begin{tabular}[c]{@{}c@{}}$\poa_a \le \frac{3.46}{k(n+1)}$\\ Thm \ref{thm:decaying_upperbound}\end{tabular} & Remark \ref{remark:just_table2}  
\\ \hline
\multirow{2}{*}{\textbf{\begin{tabular}[c]{@{}c@{}}\\ Incomplete\end{tabular}}} & Lower  & \begin{tabular}[c]{@{}c@{}}$\pos_a \ge L$\\ Cor \ref{thm:pim_lower}\end{tabular} & \begin{tabular}[c]{@{}c@{}}$\poa_a \ge \frac{k-1}{k+1} - \sqrt{\frac{\log n}{2k+3}}$\\ Thm \ref{thm:pim_neg_incomp}\end{tabular} & Remark \ref{remark:just_table} & Remark \ref{remark:just_table}& Open \\ \cline{2-7} & Upper      & Remark \ref{remark:just_table3} & Remark \ref{remark:just_table}   & Remark \ref{remark:just_table} & Remark \ref{remark:just_table} & \begin{tabular}[c]{@{}c@{}}$0.39$-$\pos$ is $\opt$\\ Thm \ref{thm:bayes_oblivious} \end{tabular} \\ \hline
\end{tabular}
}
\caption{Summary of results on prior-independent mechanism. \emph{ MHR} refers to MHR distribution, \emph{ Inc} refers to nondecreasing p.d.f., and \emph{Sym Ind} refers to independent utility setting with the uniform distribution $U[0,1]$.
\emph{Open} denotes that this regime remains an open problem.
For general case lower bound for both settings, and symmetric case upper bound in the complete information setting, we assume that the support of the principal's utility function lies in $[0,L]$. For Sym MHR, Sym Inc, and lower bound in the incomplete information setting under Sym, we present the results for $U[0,1]$ for brevity.}
 \vspace{-6mm}
 \label{table:2}
\end{table}

In this context, we additionally consider an \emph{independent utility} assumption in which the agents' and principal's utilities are independent of each other (see Definition~\ref{def:indep-utility} for more details).
Then we show that there exists a good approximate Bayes Nash equilibrium (BNE) (see Definition \ref{def:apxbne} for more details) such that the principal's expected utility becomes exactly optimal.
\begin{customthm}{\ref{thm:bayes_oblivious}}
    In the incomplete information setting, consider symmetric agents and independent utility assumption such that $X_{i,j},Y_{i,j} \sim U[0,1]$.
    Then, there is an optimal PIM with $\beta$-approximate BNE where $\beta = 1-\exp\parans{\frac{-n^2}{2(n-1)^2}}$.
    Note that $\beta$ becomes $0.3935...$ as $n$ grows.
\end{customthm}
Different from the complete information setting, deriving and analyzing BNE is far more challenging, especially for our problem setup.
One reason is that the feasible action of an agent depends on the realization of his solution.
In addition, computing the expected payoff of an agent's strategy requires us to compute the expectation over every possible realization of the others' solutions.
Despite these difficulties, we obtain the above result, of which the proof heavily relies on the FKG inequality (Lemma \ref{lm:FKG}) and some properties on the joint density function of two order statistics, preceded by some algebraic manipulations.
Using a similar but more involved technique, in Appendix~\ref{sec:correlated_equi}, we show that there exists a PIM such that under some \emph{Bayes correlated equilibrium}, it achieves a multiplicative approximation factor of $\exp\parans{\frac{n^2}{2(n-1)^2}} \cong 1.649$ (as $n$ increases).
We present the summary of results on PIM in Table~\ref{table:2}.
We remark on some points regarding the table as follows.
\vspace{-5mm}
\begin{remark}\label{remark:just_table}
    We do not study the upper bound for symmetric agents in the incomplete information setting because the lower bound in this setting can be arbitrarily bad for certain problem instances.
    In addition, we do not consider Sym MHR and Sym INC in the incomplete information setting since the lower bound in the symmetric setting already implies a pessimistic result.
\end{remark}
\vspace{-3mm}
\begin{remark}\label{remark:just_table2}
    Since the complete information setting inherently admits positive results only with the symmetry assumption, we do not consider the symmetric agents and independent utility scenario in the complete information setting.
\end{remark}
\vspace{-3.5mm}
\begin{remark}\label{remark:just_table3}
    The upper bound for the general incomplete information setting is not studied, as one can easily obtain a mechanism with an approximation ratio of $\infty$, \eg a mechanism that selects a solution to maximize the agent's utility.
\end{remark}

\paragraph{Examination cost}
Finally, we study an extension to consider the examination cost of the mechanism.
The examination cost of a mechanism is defined by the number of candidates that the mechanism has to access to choose the winner.
We refer to Definition~\ref{def:exam} for more details.
In this extension, we obtain the following result of a similar flavor to Theorem~\ref{theorem:multi-agent-unlimited-budget}.
\begin{customthm}{\ref{theorem:multi-agent-limited-budget}}
    Define $X_{(n + 1)} = -\infty$.
    For any $\tau \ge 0$, there exists a mechanism with examination cost $2 \le B < n$ such that $\poa_a$ is at most,
    \begin{align*}
    \Ex{X_{(1)}}
    -
    \Ex{X_{(2)}\IND\left[X_{(2)} \geq \tau > X_{(B+1)}\right] -  \tau\IND\left[X_{(1)} \geq \tau > X_{(2)}\right]
    -
    \sum_{e=B+1}^{n} X_{(e-B+2)}\IND\left[X_{(e)} \geq \tau > X_{(e+1)}\right]}.
    \end{align*}
\end{customthm}
This depicts the trade-off between the examination cost and the approximation ratio.
Indeed, $\poa$ increases as $B$ decreases, and plugging $B = n$ carries over to the bound in Theorem~\ref{theorem:multi-agent-unlimited-budget}.







\section{Related work}\label{sec:related_work}
Initiated by \citet{holmstrom1980theory}, a rich line of literature has studied the delegation mechanism without monetary transfer.
\citet{holmstrom1980theory} considers a setting in which a single agent is delegated to solve an optimization problem over a compact interval, and characterizes conditions under which an optimal solution of the principal exists. \citet{alonso2008optimal} fulfill the results by providing a general characterization of the solution.
\citet{armstrong2010model} study a discrete model in which an agent samples a set of solutions from a distribution and optimizes over the sampled set, which is the most similar problem setup to ours.

While most literature in the theory of delegation aims at characterizing regimes in which the optimal mechanism exists or the structure of it, \citet{kleinberg2018delegated}  study the loss of efficiency from delegating the task into an agent with misaligned payoff. They consider two models, one of which is similar to delegated project choice of \citet{armstrong2010model}.
In this model, the agent samples a fixed number of solutions from a distribution and proposes some of the solutions to the principal in a strategic manner.
Surprisingly, they show that a simple mechanism enjoys a $2$-approximation ratio compared to the ideal scenario in which the principal can directly choose the optimal solution, based on a connection to the prophet inequalities.

Several works have been devoted to the delegation mechanism without money in the presence of multiple agents.
One of the most relevant problem setting would be \emph{legislative game} with asymmetric information by~\citet{gilligan1989asymmetric}, \citet{krishna2001asymmetric}, \citet{martimort2008informational} and \citet{fuchs2022listen}.
They consider a model in which two self-interested agents with information advantage observe \emph{a state of nature} and independently send a private signal to the principal.
The principal aims to extract truthful information on the state of nature to maximize own utility that depends on it.
Recent works by \citet{alonso2014resource} and \citet{gan2022optimal} extend these settings by considering multiple actions for the principal.
Importantly however, all of their models are significantly different from our model.
In our model, the information asymmetry between the principal and the agents arises from "moral hazard" versus "adverse selection" of the previous works, as per \citet{ulbricht2016optimal}.\footnote{Notably, \citet{ulbricht2016optimal} study both effect of moral hazard and adverse selection in the delegation mechanism but with a monetary transfer. Likewise, there exists a fruitful line of literature dealing with contract-based delegation, \eg \citet{krishna2008contracting} and \citet{lewis2012theory}, but we will not discuss in detail on the delegation mechanism with money.}
In moral hazard, the principal cannot directly observe all the alternatives, but can only access a subset of them based on the signals, whereas in adverse selection, a state of nature is revealed only to the agent but affects the principal's utility  as well.

Finally, there is an increasing attention in the Computer Science community regarding an algorithm design with a delegated decision.
\citet{bechtel2020delegated} considers a stochastic probing problem when the principal delegates the role of probing to a self-interested agent. They quantify the loss of delegation similar to \citet{kleinberg2018delegated} by revealing a connection between delegated stochastic probing and generalized prophet inequalities.
\citet{bechtel2022delegated} generalizes the second model of \citet{kleinberg2018delegated}, namely Weitzman's box problem which deals with a costly sampling of the agent, by considering a combinatorial constraint.
\section{Problem setup}
\label{sec:setup}
\vspace{-2mm}\xhdr{Multi-agent delegation.}
Consider a \emph{principal} who designs a delegation mechanism and a~set of $n$ \emph{agents} who perform the task and decide solutions.
Let $\Omega$ be a set of potential solutions.
Each agent is equipped with a number of solutions $k_i \in \N$, and each solution $j$ for $j \in [k_i]$ is sampled from a probability distribution $D_{i,j}$ supported on $\Omega$.
We assume that $k_i \ge 2$ for $i \in [n]$ to avoid trivial regime.
Each distribution $D_{i,j}$ may refer to
the potential quality of the $j$th solution of agent $i$, and the number of solutions $k_i$ can be interpreted as the agent's willingness-to-pay in terms of his effort, both of which are exogenous to the mechanism.
We denote by $\om_{i,j} \in \Omega$ the agent $i$'s $j$-th solution for $i \in [n],j \in [k_i]$, and $\omb_i$ the agent $i$'s multi-set of all solutions.
We further use  $\omb = \cup_{i \in [n]}\omb_i \in \Omega^*$ to denote the multi-set of all the solutions sampled by the agents.
After agent $i$ observes his own solutions $\omb_i$, he sends a signal to the principal, and then the principal finally adopts one of the solutions based on the agents' signals.


For notational consistency, let \emph{solution} denote an outcome observed by an agent, \emph{candidates} to denote the set of solutions that the mechanism can access given the agents' signals, \emph{winner} to denote the final solution adopted by the principal.
We further say our problem to be \emph{single-agent} problem if $n=1$, and \emph{multi-agent} problem if $n \ge 2$.

\paragraph{Mechanism}
A \emph{mechanism} $M = (\Sigma, f, S)$ is specified by a set of signals $\Sigma = \Sigma_1 \times \ldots \times \Sigma_n$ of which the $i$-th coordinate $\Sigma_i$ corresponds to the signal from agent $i$, a set of possible random bits $S$ of the mechanism, and an allocation function $f: \Sigma \times S \mapsto \Omega$ which denotes the winner selected by the principal given the signals by agents and the random bits.
If the mechanism is deterministic, we overwrite the domain of the allocation as $f: \Sigma \mapsto \Omega$.
We mostly deal with a deterministic mechanism but introduce a general definition for completeness.
For some given random bits $s \in S$, we denote by $M_s$ the mechanism in which its random bits are fixed to be $s$. 
Note that we use a generalized notion of signal to denote any possible interaction from the agents to the principal.
For example, agent $i$ may simply propose a subset of his observed solutions as candidates, or he may exploit a mixed strategy of proposing solutions randomly.
Throughout the paper, we use $\Omega^*$ to denote a family of all finite multi-set over $\Omega$.

\paragraph{Agent's strategy} 
Each agent $i$ selects a certain signal based on a \emph{strategy} $\sigma_i:\Omega^* \mapsto \Sigma_i$.
More formally, given a multi-set of the $k_i$ solutions $\omb_i = \{\om_{i,1},\ldots, \omega_{i,k_i}\}$, agent $i$ sends a signal $\sigma(\omb_i)$ to the principal.
Importantly, agent $i$ determines the strategy based on the information regarding the others' solutions as well as his own solutions.
We elaborate more in the subsequent paragraph.
The principal has a utility function $x:\Omega \mapsto \R_{> 0}$ and
each agent $i$ for $i \in [n]$ has a utility function $y_i:\Omega\mapsto \R_{> 0 }$.
If solution $\omega \in \omb$ is selected as the winner and $\om$ belongs to agent $i$, then the principal and agent $i$ realize utilities of $x(\omega)$ and $y_i(\omega)$ respectively.
The other agents $j\neq i$ realize their utility to be zero.
Note that we restrict the winner to be within $\omb$ so that the principal cannot commit to a solution beyond the observed set of solutions.
We further assume that $\Omega$ includes a null outcome $\perp$.
If the principal selects $\perp$ as the winner, this represents that no solution is adopted by the principal and the realized utility is zero, \ie $x(\perp) = y_i(\perp) = 0$ for $i \in [n]$.



\paragraph{Information structure}
In terms of the principal's information for agents' distributions of solutions, we consider the following two types of mechanism.
\begin{definition}[Bayesian and prior-independent mechanism]
	We say that mechanism $M$ is \Emph{prior-independent mechanism} (PIM), if it does not have any information on the distributions from which the solutions are drawn.
	If it priorly knows the distribution from which the quantity $x(\cdot)$ over the solutions are drawn for each agent,
	we denote by \Emph{Bayesian mechanism} (BM).
\end{definition}
Importantly, as noticed in the definition, PIM has significant informational gain over BM.
It does not require any process of acquiring information on the distribution of solutions.
Indeed, in practice, it is difficult or even impossible to estimate the distribution of each agent in advance.
Even if it's possible, it usually entails a costly procedure to acquire such information, and thus PIM is more broadly applicable.

Regarding the principal's information on the agents' solutions, we consider an information asymmetry arises from \emph{moral hazard} such that the principal has no information on the realized set of solutions since she delegates the search (choice) to the self-interested agents.
Instead, she can only access a subset of solutions based on the signals provided by the agents.

To model the agent's information structure, we basically assume that the agents know the principal's utility function.
Furthermore, we consider two cases of complete information and incomplete information based on the agent's information with respect to the others.
In the \Emph{complete information} setting, all the observed solutions are revealed to each other, thereby each agent determines his own strategy after observing all the solutions.
In the \Emph{incomplete information} setting, only the distributions of the solutions are publicly available to the agents but not the exact observed solutions, and thus each agent determines his own strategy based on the given knowledge of the others' distributions.

\paragraph{Allocation function and overall mechanism}
We first define an interim allocation function which represents the winner adopted by the mechanism given the observed set of solutions.
\begin{definition}[Interim allocation function]
$f_{M,\sigma}: \Sigma \mapsto \Omega$ is an (interim) allocation function given mechanism $M$ and strategy $\sigma$.
That is, given $\omb = \cup_{i \in [n]}\omb_i$, the mechanism $M$ selects $f_{M,\sigma} (\omb)$ as the winner when each agent $i$ plays $\sigma_i$.
In case of randomized mechanism,  we use $f_{M_s, \sigma}$ to denote the allocation function given the random bits $s$, and $F_{M,\sigma}$ to denote the set of possible winners $\cup_s f_{M_s, \sigma}$.
\end{definition}
Overall, the mechanism proceeds as follows: (1) the principal commits to a mechanism $M$, (2) each agent $i$ observes solutions $\omb_i$, (3) each agent $i$ strategically sends a private signal $\sigma_i$, and (4) the principal selects the winner $f_{M,\sigma}$ (or $f_{M_s, \sigma}$).

\subsection{Further preliminaries}
We here present further notations that will be broadly used throughout the paper.
\begin{definition}[Nash equilibrium]
	Given a mechanism $M$, a set of strategies $\sigma = (\sigma_1,\sigma_2,\ldots, \sigma_n)$ is called a \emph{Nash equilibrium}, if $\sigma_i$ yields
	\begin{align}
		y_i(\sigma_i, \sigma_{-i})
		\ge
		y_i(\sigma_i', \sigma_{-i}),\label{ineq:01251707}
	\end{align}
	for agent $i$'s any other strategy $\sigma'_{i}$ given the others' strategy $\sigma_{-i}$.
	We call each $\sigma_i$ an equilibrium strategy of agent $i$. We often abuse equilibria or equilibrium to denote Nash equilibrium.
    Whenever the mechanism is randomized, we take expectation over the utility in~\eqref{ineq:01251707}.
\end{definition}


%
For sake of exposition, we define $X_{i,j} = x(\om_{i,j})$  and $Y_{i,j} = y_i(\om_{i,j})$ for each $j \in [k_i]$ and $i \in [n]$. 
Furthermore, we denote by $X_{i,(j)}$ the $j$-th highest solution in $\om_i$ in terms of $x(\cdot)$ so that $X_{i,(1)} \ge X_{i,(2)} \ge \ldots \ge X_{i,(k_i)}$ and 
define $Y_{i,(j)}$ the agent $i$'s utility for the solution that corresponds to $X_{i,(j)}$.
We often overwrite $X_{i,(1)}$ by the \emph{first-best solution} of agent $i$, and the agent who has $X_{(i)}$ by \emph{$i$-th best agent}.
Note that $X_{i,j}$ and $Y_{i,j}$ for each $j \in [k_i]$ and $i \in [n]$ are random variables such that $X_{i,j}$ and $Y_{i,j}$ are correlated with respect to $\om_{i,j}$.

We measure the principal's utility by comparing it to a strong benchmark $\opt$ given a realized set of solutions $\omb$, defined formally as follows.
\begin{align*}
\opt= \Ex{\max_{\om \in \omb} x(\omega)} = \Ex{\max_{i \in [n]} \max_{j \in [k_i]} X_{i,j}}.
\end{align*}
We often abuse $X_{\max} = X_{(1)} = \max_{i \in [n]} \max_{j \in [k_i]} X_{i,j}$.
Note that $\opt$ denotes the principal's utility when she can access all the solutions sampled by the agents, and select the best solution by herself.

In terms of the approximation ratio of the mechanism, we consider both the additive and multiplicative approximation factors, defined as follows.
\begin{definition}
    \label{def:apx-mech}
	Mechanism $M$ and corresponding strategies $\sigma$ are $(\rho,\gamma)$-approximate if
	\begin{align*}
		\rho\Ex{f_{M,\sigma(\omb)}} + \gamma
		\ge
		\opt.
	\end{align*}
\end{definition}
Note that the expectation in LHS is taken over the randomness in $M$ as well as that in solutions.

Since we deal with the approximation factor of the mechanism in its equilibrium, we further define the notion of the price of anarchy and the price of stability as follows.
\begin{definition}[Price of anarchy/stability]
    \label{def:poa-pos}
    Mechanism $M$ has multiplicative \emph{price of anarchy} ($\poa_m$) of $\rho$ if it is $(\rho,0)$-approximate under any equilibrium.
    It has multiplicative \emph{price of stability} ($\pos_m$) of $\rho$ if there exists an equilibrium such that $M$ is $(\rho,0)$-approximate.
    We similarly define $\poa_a$ and $\pos_a$ to denote the price of anarchy/stability for the additive approximation factor $(1,\gamma)$.
\end{definition}

For self-containedness, we present the mechanism introduced by \citet{kleinberg2018delegated}.
\begin{definition}[Single-proposal mechanism]
In a single-proposal mechanism (SPM) under the single-agent setting, an eligible set $R \subset \Omega$ is announced by the principal, the agent proposes a~solution $\om$, and the principal accepts it if $\om \in R$, and rejects it otherwise.
\end{definition}

In our multi-agent setup, we mainly investigate a variant of the SPM defined as follows.
\begin{definition}[Multi-agent single-proposal mechanism]\label{def:mspm}
In a \emph{multi-agent single-proposal mechanism (MSPM)}, a tie-breaking rule $\rho$ and an eligible set $R_i \subset \Omega$ are announced by the principal for each agent $i$.
Each agent proposes up to one solution, and the principal filters a candidate from agent $i$ if it is not contained in $R_i$.
After filtering, the principal selects the winner by choosing the candidate that maximizes her utility, given the tie-breaking rule $\rho$.
We say that MSPM has a homogeneous eligible set if $R_i = R$ for $i \in [n]$.\footnote{We provide further examples of reasonable mechanisms in Appendix~\ref{sec:fur_fea_mec}.}
\end{definition}

\begin{remark}\label{rm:bm_pim_comp}
Importantly, in terms of the principal's information, MSPM can be both BM and PIM depending on the eligible set $R_i$.
If we specify $R_i$ to admit any solution from $\Omega$, then it is PIM as we do not need any kind of prior knowledge on the solution.
However, if we want to specify $R_i$ to accept only a certain type of solution as in SPM by \citet{kleinberg2018delegated}, then we essentially need the prior information regarding the distribution of the solution.
This is indeed true because if we commit to an eligible set that possibly rejects some solutions without knowing the distribution of the solutions, then it is obvious that we can construct a worst-case problem instance in which the mechanism essentially cannot accept any solution, resulting in an arbitrarily bad approximation factor. We refer to Claim~\ref{cl:3} in Appendix~\ref{pf_thm:comp_obl_neg} for a  formal discussion.
\end{remark}

Finally, we pose a (minor) assumption that if an agent is indifferent between two solutions, he selects a solution that induces a higher utility for the principal.
\begin{assumption}[Pareto-optimal play]\label{as:pareto}
If an agent has two solutions $\om$ and $\om'$ such that $x(\om ) \ge x(\om')$, $y(\om) \ge y(\om')$ and at least one of these inequalities are strict, then the agent does not play strategy $\sigma_i = \om'$.
\end{assumption}
\begin{remark}
This assumption is reasonable since there exists no merit in playing such a dominated strategy $\om'$.
Moreover in practice, the agents may also want to behave in favor of the principal as much as possible unless they need to sacrifice their own utilities, for any future interactions.
\end{remark}
Before presenting our main results, we begin with the following observation that any mechanism can be reduced to MSPM with the proper choice of parameters.
\begin{theorem}\label{thm:multiagent_revel}
In the multi-agent setting, for any mechanism $M$ with an arbitrary equilibrium $\sigma$, there exists a MSPM $M'$ such that for any equilibrium $\sigma'$, for any $\om \in F_{M,\sigma}(\omb)$ and any $\om' \in F_{M',\sigma'}(\omb)$, we have $x(\om) \le x(\om')$ for any observed solutions $\omb \in \Omega^*$.
\end{theorem}
The proof is provided in Appendix \ref{pf_thm:multiagent_revel}.
Thanks to the theorem, it suffices to narrow down our scope of interest to the set of MSPMs.
Correspondingly, we focus on analyzing the approximation factor of MSPM in both the complete and incomplete information settings.

\section{Complete information}\label{section:no_cost}
In the complete information setting, we first present a generic result regarding MSPM with arbitrary threshold-based eligible sets.


\begin{theorem}\label{theorem:multi-agent-unlimited-budget}
Given any $\tau \ge 0$, let $R$ be an eligible set with threshold $\tau$ such that
$R:= \{\omega: x(\omega) \ge \tau, \omega \in \Omega\}$, and $\rho$ be
an arbitrary deterministic
tie-breaking rule.
Then, the MSPM with $\rho$ and a homogeneous eligible set $R$ has $\poa_a$ of at most
\begin{align*}
\Ex{X_{(1)}} - \Ex{X_{(2)}\IND\left[X_{(2)} \geq \tau\right]} - 
\Ex{\tau\IND\left[X_{(1)} \geq \tau > X_{(2)}\right]}.
\end{align*}
\end{theorem}
We present the proof in Appendix \ref{pf_theorem:multi-agent-unlimited-budget}.
To interpret the bound above, suppose that we set $\tau = 0$.
Then it implies that the principal can obtain the utility of first-best solution of the second-best agent.
In addition, by properly defining the eligible set, \ie $\tau$ in this case, she can further improve the potential quality of the candidates, although it might possibly decrease the probability that each agent samples eligible solutions.
Intuitively, the existence of the other agents make each agent more competitive in sending signals, thus motivating them to submit more qualified solution.

\begin{remark}
The role of deterministic tie-breaking rule is essential. 
Indeed if the principal exploits a randomized tie-breaking rule, it may reduce the power of competition in certain cases since the agents may settle with a bad solution in terms of the principal, as it still gives the agents higher utility with some small but still positive probability due to the randomized tie-breaking rule.
\end{remark}


As per Remark~\ref{rm:bm_pim_comp}, the mechanism presented in Theorem~\ref{theorem:multi-agent-unlimited-budget} can be either BM or PIM.
In what follows, we obtain guarantees for BM and PIM independently by quantifying the approximation factor in Theorem~\ref{theorem:multi-agent-unlimited-budget}.

\subsection{Bayesian mechanism}
For BM, we first show that we can exactly recover the guarantees under the single-agent setting provided in \citet{kleinberg2018delegated} by setting proper eligible sets.
\begin{theorem}\label{thm:recover}
    There exists a BM with $\poa_m \le 2$.
\end{theorem}
The proof can be found in Appendix \ref{pf_thm:recover}.
To obtain the result above,
we reduce the multi-agent delegation problem to a single-agent one and then
use the results by \citet{kleinberg2018delegated} regarding the upper bound on the multiplicative approximation factor.
\begin{remark}
    In the proof of Theorem \ref{thm:recover}, we in fact derive a stronger result such that there always is a multi-agent mechanism which yields higher (or equal) principal's utility compared to any mechanism in the corresponding single-agent setting.
    We complement this result in Appendix~\ref{sec:compare_mul_sin} by proving that for any single-agent problem instance, we can construct a multi-agent instance with higher (or equal) principal's utility.
\end{remark}
The tightness of the bound in Theorem~\ref{thm:recover} is also obtained as follows.
\begin{theorem}\label{thm:lowerbound_comp_worst}
        For any $\eps>0$, there exists a problem instance such that no BM achieves $\pos_m < 2-\eps$.
\end{theorem}
To construct a matching lower bound, we introduce a \emph{super-agent} of superior solutions.
The super-agent always induce solutions each of which is dominant over any solution by other agents.
The principal thus always select one of the super-agent's solution, which in turn reduces the multi-agent setting to the single-agent setting.
Then we borrow the standard worst-case problem instance of the prophet inequalities to construct a matching lower bound.
The formal proof is provided in Appendix \ref{pf_thm:lowerbound_comp_worst}.
In words, if the principal faces a superior agent who overwhelms all the other agents, there is no merit in delegating to multiple agents, other than the superior one.
%

This naturally bring us to consider a symmetric assumption defined as follows.
\begin{definition}[Symmetric agents]\label{def:symmetric} 
    \label{def:symmetric-agent}
	If $k_i = k$ for $i \in [n]$, and $D_{i, j} = D$ for $i \in [n], j \in [k]$ we say that the agents are symmetric.
\end{definition}
This assumption implies that the agents have a similar level of effort and skill, and importantly, all the agents can possibly become the owner of the winner.
\begin{example}
In the example of research proposal selection process we presented in Section~\ref{sec:intro}, the grant agency (principal) may announce a shortlist of the investigators (agents) with a similar level of research abilities, \eg by comparing their publication records.
Indeed, the grant agency will not include an investigator who is expected to be not qualified than the others in the shortlist.
\end{example}

The symmetry assumption effectively sparks a fire in the competitive tension between the agents due to their similar level of quality.
Formally, in the following theorem, we obtain an upper bound on the additive approximation factor of a Bayesian mechanism under the symmetric agents.
We refer to Appendix \ref{pf_thm:informed} for the proof.
\begin{theorem}\label{thm:informed}
	Consider symmetric agents such that $X_{i, j} \sim D$ where $D$ is supported on $[0,1]$ and its c.d.f. $F(x)$ satisfies $F(x) \le x^{\alpha}$ for any $x \in [0,1]$ for some $\alpha >0$. 
    Then, there exists a BM with $\poa_m \le \parans{\frac{\alpha nk}{\alpha nk+1}}^{\nicefrac{-1}{\alpha nk}}$, where $ \parans{\frac{\alpha nk}{\alpha nk+1}}^{\nicefrac{-1}{\alpha nk}}$ converges to $1$ as $\alpha n k$ increases.
\end{theorem}

In order to obtain the above result, we introduce a tighter threshold $\tau$ compared to the ones studied in the canonical prophet inequalities algorithm.
Combined with the guarantee of MSPM in Theorem~\ref{theorem:multi-agent-unlimited-budget}, followed by some elementary calculations, we derive the subconstant bound above, where the principal attains the optimal utility as $n$ and $k$ increase.
This is in stark contrast to the general setting under which the principal can obtain only a constant factor of the optimal utility.

\subsection{Prior-independent mechanism}\label{sec:comp_pim}
Now we turn our attention to the class of PIM.
\footnote{Standard approach to analyze prior-independent mechanism would be to use Bulow-Klemperer-style result by \citet{bulow1994auctions}, but this does not apply for our case. We discuss more in details in Appendix~\ref{sec:BKS}.}
By plugging $\tau = 0$ in Theorem~\ref{theorem:multi-agent-unlimited-budget}, we directly obtain the following corollary, where we omit the proof as it is trivial to derive.
\begin{corollary}\label{cor:infbud_notau}
    There exists a PIM with $\poa_{a}$ of at most $\Ex{X_{(1)}} - \Ex{X_{(2)}}$.
\end{corollary}
This corollary mainly implies that even without the eligible set, the multi-agent setting brings some advantages for the principal, which is mainly due to the competitive tension between the agents.
As the readers may have noticed, this additive approximation factor is equivalent to the expected difference between the $n$-th (highest) order statistics and $n-1$-th order statistics of the random variables $\left(X_{i,(1)}\right)_{i \in [n]}$, which is again the highest order statistics of each agent.

Although there exists a vast amount of papers studying order statistics of similar flavor in the literature, the expected gap $\Ex{X_{(1)} - X_{(2)}}$, \ie the expected difference of the highest order statistics between the first-best and second-best agents, has not been studied before to the extent of our knowledge.
In this context, we provide some analysis under which the additive factor is nicely bounded, using some characterizations on the order statistics.

Before presenting our positive results, we present the following negative result, which dictates that one should consider a reasonable assumption on the distribution to obtain positive results.
\begin{theorem}\label{thm:comp_obl_neg}
    Suppose that both the principal's utility functions are supported on $[0,L]$ for some $L >0$.
    For any $\eps > 0$, there exists a problem instance such that $\pos_a \ge L-\eps$.
    
\end{theorem}
We provide the proof in Appendix \ref{pf_thm:comp_obl_neg}.
Similar to the case of BM, we consider the symmetric agents.
In this case, we can derive the following upper bound on the additive approximation factor regardless of the underlying distributions.
We present the proof in Appendix \ref{pf_thm:comp_sym_general}.
\begin{theorem}\label{thm:comp_sym_general}
    Under the symmetric agents, suppose that the principal's utility functions are supported on $[0,L]$. Then there exists a PIM with $\pos_a$ of $L(1-\nicefrac{1}{n})^{n-1}$, where $(1-\nicefrac{1}{n})^{n-1} \le 1/2$ for all $n \ge 2$, and it decreases to $\nicefrac{1}{e}$ as $n$ increases.
\end{theorem}
We also note that this upper bound $L(1-\nicefrac{1}{n})^{n-1}$ is tight for $\Ex{X_{(1)} - X_{(2)}}$, \ie we can construct a problem instance with exactly having $\Ex{X_{(1)} - X_{(2)}} = L(1-\nicefrac{1}{n})^{n-1}$.
Note that this approximation ratio is a big improvement over the general setting under which the additive approximation factor can be arbitrarily bad.

Moreover, if we consider a specific scenario where the principal's utility is endowed with a certain class of distributions, we show that the approximation factor can become arbitrarily small as the number of solutions per agent ($k$) increases.
At the heart of this result, we prove several new properties on the order statistics, which might be of independent interest.
First, we show that the expected difference between any top $k$-th and $k+1$-th order statistics is a nonincreasing function on $n$ for any distributions with \emph{monotone hazard rates} (MHR).
Note that a distribution $D$ with p.d.f. $f(x)$ and c.d.f. $F(x)$ has a monotone hazard rate if $h(x) = f(x)/(1-F(x))$ is nondecreasing on $x$.
Using a similar technique, we also show that the expected difference between any $k$-th and $k+1$-th lowest order statistics is also a nonincreasing function on $n$, for any distribution with \emph{monotone reverse hazard rate}, where a distribution has MRHR if $h'(x) = f(x)/F(x)$ is nonincreasing on $x$ .

In case of the difference between the first and the second order statistics, or the $n-1$-th and the $n$-th order statistics under the distribution with MHR, \citet{li2005note} shows that they indeed are nonincreasing on $n$ with applications to the expected rent of the (reverse) auction.
To the best of our knowledge, we first provide a generalized and strengthened result of theirs.
Our proof builds upon the techniques by \citet{watt2021concavity} which proves the convexity of $k$-th order statistics with respect to $n$.
We present the precise statement below, and its proof is provided in Appendix \ref{pf_lm:order_stat_diff}.
\begin{lemma}\label{lm:order_stat_diff}
	   Let $X_1, \ldots, X_n$ be i.i.d. random variables equipped with a continuous distribution $D$.
		Let $X_{1:n} \le \ldots \le X_{n:n}$ be corresponding order statistics.
		Define $\Delta_{k,n} = \Ex{X_{n-k+1:n} - X_{n-k:n}}$, \ie the difference between the $k$-th largest and $k+1$-th largest order statistics, and $\delta_{k,n} = \Ex{X_{k+1:n} - X_{k:n}}$.
		Then, for any $k \le n-1$, the following holds.
		\begin{enumerate}
			\item If $D$ has MHR, then $\Delta_{k,n}$ is a nonincreasing function on $n$.
			\item If $D$ has MRHR, then $\delta_{k,n}$ is a nonincreasing function on $n$.
		\end{enumerate}
	\end{lemma}
 \begin{remark}
The fact that these quantities are nonincreasing on $n$ is fairly intuitive at first glance since as more samples are drawn their order statistics will lie more compactly within a fixed interval, however, this is not true in hindsight for certain cases. For example, in case of Bernoulli distribution with $p=1/2$, it is straightforward to check that $\Delta_{k,n} = \Pr{\sum_{i=1}^n X_i = n-(n-k+1)+1 = k}= \binom{n}{k}2^{-n}$. Hence, $\Delta_{k,n} \ge \Delta_{k,n+1}$ is equivalent to $\frac{n+1}{n-k+1} \le 2$, which holds only for $k \le \frac{n+1}{2}$.
\end{remark}
 
Next, we prove that any $k$-th order statistics of i.i.d. random variables with MHR distribution has MHR distribution, which also has not been studied before to the extent of our knowledge. We defer the proof in Appendix \ref{pf_lm:mhr_preserving}.
\begin{lemma}\label{lm:mhr_preserving}
		Let $X_1, \ldots, X_n$ be i.i.d. random variables equipped with a MHR distribution, then its $r$-th order statistics has MHR distribution for any $r \in [n]$.
	\end{lemma}

Using the two lemmas described above, we prove that under MHR distribution, PIM can achieve subconstant additive approximation factor as follows. We present the proof in Appendix \ref{pf_thm:mhr_upperbound}.
\begin{theorem}\label{thm:mhr_upperbound}
    Under symmetric agents with $X_{i,j} \sim D$ such that $D$ has MHR,  there exists a PIM with $\poa_a \le \sqrt{\frac{4}{3}\var(D_{\max})}$ where $D_{\max}$ is the distribution of $X_{i,(1)}$ and $\var(\cdot)$ denotes the variance of the distribution.
    In case of $D = U[0,1]$, the approximation factor becomes approximately $\frac{1.155}{k}$.
\end{theorem}

We note that for uniform distribution, the subconstant approximation factor above guarantees the negative dependency on the number of solutions $k$.
Thus, it improves largely upon the constant factor approximation of Theorem~\ref{thm:comp_sym_general}, and the constant factor approximation ratio of $2$ for the Bayesian mechanism in \citet{kleinberg2018delegated}.

Furthermore, if we consider a more restrictive class of problem instances, we prove that one can obtain both negative dependencies on $n$ and $k$, \ie the $\poa$ converges to $1$ whenever $n$ or $k$ increases.
Similar to the MHR distribution case, our proof heavily relies on the following lemmas.
\begin{lemma}\label{lm:order_stat_diff_restricted}
        Define variables the same as Lemma~\ref{lm:order_stat_diff}.
	Then, for any $k \le n-1$, the following holds.
		\begin{enumerate}
			\item If $f(x)$ is nondecreasing on $x$, then $(n+1)\Delta_{k,n} \le n\Delta_{k,n-1}$.
			\item If $f(x)$ is nonincreasing on $x$, then $(n+1)\delta_{k,n} \le n\delta_{k,n-1}$.
		\end{enumerate}
\end{lemma}

\begin{lemma}\label{lm:id_preserving}
        Let $X_1, \ldots, X_n$ be i.i.d. random variables  with p.d.f. $f$.
	Then, the following holds.
	\begin{enumerate}
		\item If $f(x)$ is nondecreasing, then the p.d.f. of the largest order statistics is  nondecreasing on $x$.
		\item If $f(x)$ is nonincreasing, then the p.d.f. of the smallest order statistics is  nonincreasing on $x$.
	\end{enumerate}
\end{lemma}

We highlight that Lemma \ref{lm:order_stat_diff_restricted} implies that the expected difference between two consecutive order statistics is strictly decreasing over $n$ with the factor of $\nicefrac{n}{n+1}$.
Using the above lemmas, we obtain the following theorem, where the proof is presented in Appendix \ref{pf_thm:decaying_upperbound}.
\begin{theorem}\label{thm:decaying_upperbound}
        Under symmetric agents with $X_{i,j} \sim D$ such that $D$ has nondecreasing p.d.f. $f(x)$, there exists a PIM with $\poa_a \le \frac{1}{n+1}\sqrt{12\var(D_{\max})}$.
	In case of $D = U[0,1]$, the approximation factor becomes approximately $\frac{3.464}{k(n+1)}$.
\end{theorem}
This strictly improves upon the approximation factor of Theorem \ref{thm:mhr_upperbound}, since the approximation factor has positive effects on both $k$ and $n$.
This highlights that delegating to multiple agents indeed brings significant advantages over the single-agent case, under which the best one may hope for is to construct constant factor approximation ratio even with Bayesian mechanism.
Notably, the approximation factor decreases to zero as $n$ or $k$ increases, \ie the principal's utility converges to $\Ex{X_{\max}}$.
It remains a major open problem to see whether the guarantees above are tight or not.

\section{Incomplete information}\label{section:no_cost_incomp}
Recall that in the incomplete information setting, each agent $i$ decides his own action at the interim stage of the Bayesian game, \ie when he observes his own solutions.
More formally, agent $i$'s strategy $\sigma_i: \Omega^* \mapsto \Sigma_i$ is a function from the observed solutions to a signal, which is typically a subset of solutions.
We denote by $\sigma_i(\omb_i)$ the agent's action given the observed solution $\omb_i$.
Given $\omb_i$ and others' strategies $\sigma_{-i}$, agent $i$'s strategy $\sigma_i(\omb_i)$ induces the following expected utility.
\begin{align*}
	u_i(\sigma_i(\omb_i), \sigma_{-i})
	=
	\Ex{y_i(f_{M,(\sigma_i(\omb_i), \sigma_{-i})})}
	=
	\int_{\omb_{-i}}\Ex{y_i(f_{M,(\sigma_i(\omb_i), \sigma_{-i}(\omb_{-i}))})}f_{-i}(\omb_{-i})d\omb_{-i},
\end{align*}
where $f_{-i}(\omb_{-i})$ denotes the p.d.f. of observing $\omb_{-i}$.
Namely, each agent accounts for the probability that certain solutions are observed by the other agents, confirms their corresponding strategies, and computes his own utility based on them.
This process is captured by the integral in the last equation.
Note that the expectation inside the integral accounts for the random bits of the mechanism, which does not appear in deterministic mechanisms.

Then, given the others' strategies $\sigma_{-i}$, agent $i$'s expected utility of playing $\sigma_i$ can be written as
\begin{align*}
	u_i(\sigma_i, \sigma_{-i})
	=
	\Exu{\omb_i}{u_i(\sigma_i(\omb_i), \sigma_{-i})}	
	= \int_{\omb_i} u_i(\sigma_i(\omb_i), \sigma_{-i})f_{i}(\omb_i)d\omb_i.
\end{align*}

Finally, a Bayesian Nash equilibrium can be defined as follows.
\begin{definition}[Bayesian Nash Equilibrium]
	Under the incomplete information setting, given a mechanism $M$, a set of strategies $(\sigma_1,\sigma_2,\ldots, \sigma_n)$ is defined to be a Bayesian Nash equilibrium (BNE), if for any $i \in [n]$, $\sigma_i$ satisfies the following for any other strategy $\sigma'_i$ of agent $i$.
	\begin{align*}
		u_i(\sigma_i, \sigma_{-i}) \ge u_i(\sigma'_i, \sigma_{-i}).
	\end{align*}
\end{definition}
Our main objective is to construct a mechanism that induces a large principal's expected utility under its BNE.
In the case of incomplete information, the analysis of the approximation ratio becomes more involved as all the computations of the expected payoff of some strategies entail probabilistic arguments regarding the realization of the others' solutions.
Moreover, since the \emph{feasible} strategy depends on the realization of the solutions for each agent, its analysis significantly differs from the standard Bayesian games such as the first-price auction or the all-pay auction.
Due to this fact, even under the symmetric agents, the standard technique of computing symmetric BNE in the first-price auction, \eg which involves the use of the Envelope theorem \cite{milgrom2002envelope}, does not apply to our setting.
Despite these analytical challenges, we succeed in obtaining several positive results in terms of the PIM.

Before presenting our main result on PIM, we first study the class of the Bayesian mechanism.
In the case of BM, it is straightforward to see that all the results in the complete information setting hold in the incomplete information setting as well.
\begin{corollary}\label{thm:recover_bayes}
	Theorem~\ref{thm:recover}, \ref{thm:lowerbound_comp_worst}, and \ref{thm:informed} still holds for the incomplete information setting.
\end{corollary}
We skip the proofs as exactly the same arguments of the corresponding proofs in the complete information setting directly carry over to the incomplete information setting.

In the case of PIM, similar to the complete information setting, we can directly argue that the additive approximation factor can be arbitrarily bad.
We omit the proof as it is exactly the same as the proof of Theorem \ref{thm:comp_obl_neg}.
\begin{corollary}\label{thm:pim_lower}
    Theorem \ref{thm:comp_obl_neg} still holds for the incomplete information setting.
\end{corollary}

This again naturally leads us to focus on the symmetric agents.
Interestingly however, in the following theorem, we show that even under the symmetric agents, for any PIM, there exists a problem instance such that proposing a solution that minimizes the principal's utility constructs a BNE, and thus we can construct a problem instance such that the approximation factor can be arbitrarily bad.
We defer the proof in Appendix \ref{pf_thm:pim_neg_incomp}.
\begin{theorem}\label{thm:pim_neg_incomp}
        Under the symmetric agents, there exists no PIM such that $\poa_a < \Ex{X_{(1)} - \max_i X_{i,(k)}}$.
        For the uniform distribution $x(D_{i,j}) \sim U[0,1]$, there exists no PIM such that $\poa_a < \frac{k-1}{k+1} - \sqrt{\frac{\log n}{2k+3}}$,
        \ie the additive approximation factor of any PIM can be arbitrarily bad as $k$ increases given $n$.
\end{theorem}

Compared to the subconstant approximation ratio of Theorem~\ref{thm:mhr_upperbound} and~\ref{thm:decaying_upperbound} in the complete information setting, Theorem~\ref{thm:pim_neg_incomp} implies rather pessimistic consequence such that its worst-case guarantees may become arbitrarily bad, alike to the efficiency of prior-independent mechanism in the single-agent setting.
Thus, the gain from multi-agent delegation becomes highly restrictive.

This negative result mainly arises from the fact that each agent has a utility function which is negatively correlated with the principal's utility, and they have limited information regarding each other's solutions.
This induces the agents to propose more in a selfish manner, compared to the complete information setting.

This phenomenon resembles the \emph{bid shading} effect in the BNEs of the first-price auction, in which the bidders tend to bid less than their valuations in equilibriums.
Due to the lack of information on the other's strategy, each agent may tend to submit a candidate which brings more ex-post utility even though it has a smaller chance of being selected.
Similarly from the first-price auction, this is mainly to prevent the case in which he submits more in favor of the principal  when the others' solutions were revealed to be less competitive in terms of being selected.
Combined with this effect and the fact that the agent's utility is strongly negatively correlated with the principal's utility, each agent tends to submit a solution that minimizes the principal's utility.

In this context, one natural question is whether we can recover a better approximation ratio in more restricted problem instances.
To this end, we consider the following assumption.
\begin{definition}[Independent utility]
    \label{def:indep-utility}
    A problem instance satisfies independence if $D_i$ is a product measure over $D_{i,x}$ and $D_{i,y_i}$, \ie $D_i = D_{i,x} \times D_{i,y_i}$, such that each of $x(\cdot)$ and $y_i(\cdot)$ are associated with the distributions $D_{i,x}$ and $D_{i,y_i}$ respectively, and $D_i$ is independent for $i \in [n]$.
\end{definition}
Note that with independent utility assumption, we can overcome the bad problem instance presented in the proof of Theorem \ref{thm:pim_neg_incomp}, as the principal and the agents do not longer have negatively correlated utility, but are rather not correlated at all.
Although we only obtain a positive result under the independent utility assumption due to the analytical tractability, we hope that one may extend our analysis into the case where the principal and the agents have positively correlated utility (or mild amount of negative correlation), and we leave it as a major open problem.

Due to the intrinsic difficulties in analyzing BNE, we consider a more general solution concept of approximate equilibrium defined as follows.
\begin{definition}[Approximate BNE]\label{def:apxbne}
	Given a mechanism $M$, a set of strategies $\sigma = (\sigma_1, \ldots, \sigma_n)$ is $\eps$-approximate BNE if for any $i \in [n]$, $\sigma_i$ satisfies the following for any other strategy $\sigma'_i$ of agent $i$.
	\begin{align*}
		u_i(\sigma_i, \sigma_{-i}) \ge u_i(\sigma'_i, \sigma_{-i}) - \eps,
	\end{align*}
\end{definition}
Our definition of $\pos$ and $\poa$ naturally carries over to the approximate equilibrium.
We eventually present our main result in the incomplete information setting.
Our main result in the incomplete information setting is as follows.
\begin{theorem}\label{thm:bayes_oblivious}
        Under the symmetric agents and independent utility with $X_{i,j} = X \sim U[0,1]$, there exists a PIM such that there exists $\eps$-approximate BNE that results in the optimal utility, \ie $\Ex{X_{\max}}$,  where $\eps = 1-e^{-\frac{n^2}{2(n-1)^2}}$.
        Note that $1 - e^{-\frac{n^2}{2(n-1)^2}}$ converges to $1-\nicefrac{1}{\sqrt{e}} \cong 0.3935$ as $n$ increases.
\end{theorem}
The proof can be found in Appendix \ref{pf_thm:bayes_oblivious}.
Remarkably, this implies that the principal can recover the optimal utility if the agents agree to sacrifice only a constant portion of their utility.
This is indeed possible due to the independence assumption and the fact that the agent does not have the exact information on the others' solutions.

Our proof is based on analyzing the probability of the event that the principal's and the agent's utility is somehow aligned, \ie an agent maximizes his expected utility by proposing his best solution in terms of principal. We show that this probability is lower bounded by $\exp \parans{\frac{-n^2}{2(n-1)^2}}$,\footnote{In Appendix \ref{app:c}, we numerically verify that our lower bound is almost tight.} which converges to $\nicefrac{1}{\sqrt{e}}$ as $n$ increases.
The proof heavily relies on the FKG inequality (Lemma \ref{lm:FKG}) and some properties on the joint density function of two order statistics.
\begin{remark}
    In appendix \ref{sec:correlated_equi}, we also prove that using a similar but more involved technique, we can derive that there exists a PIM such that under some \emph{Bayes correlated equilibrium}, it achieves $\parans{\exp\parans{\frac{n^2}{2(n-1)^2}},0}$-approximation, \ie constant approximation ratio for sufficiently large $n$.
\end{remark}

\begin{remark}
Note that in some cases it might be complicated or even computationally intractable for each agent to compute a best-response given the others' strategies since he needs to account for every possible realization of the solutions for the others.
In this context, it would be reasonable for the agents to follow the mechanism's recommendation to propose a solution that maximizes the principal's utility, even though it may slightly decrease their utility up to some constant factors.
\end{remark}

\section{Examination cost}\label{sec:cost}
Since we allow multiple agents to propose a solution, on one hand, it helps the principal find a better solution because she can compare more solutions as observed in Section~\ref{section:no_cost} and \ref{section:no_cost_incomp}.
On the other hand, since the principal needs to select which ones to choose among the multiple candidates,
it entails an additional burden to \emph{evaluate} the candidates, compared to
the single-agent setting.

\begin{example}
Suppose that a project manager delegates the task of finding a solution for a project
to multiple teams.
Each team proposes a candidate and then
the manager needs to find the best candidate by
sequentially examining the candidates.
This procedure can be costly for the manager and even
cannot be parallelized since the manager may want to examine the candidates by herself.
\end{example}

To effectively study the trade-off between such a burden in evaluating the candidates and the principal's utility,
we introduce the following notion of the examination cost of a mechanism.

\begin{definition}[Examination]\label{def:exam}
Given a set of solutions $\omb \in \Omega^*$ and corresponding strategies $\sigma$,
suppose that mechanism $M$ can eventually choose a winner among a subset $p \subset \omb$.
The mechanism cannot directly access the subset $p$, and does not have information on their values,
but only knows from whom each candidate in $p$ is proposed.
In order to select the winner,
the mechanism should examine the solution by accessing each candidate in $p$, and this incurs an \emph{examination cost} of $1$.
We call this procedure an examination process.
\end{definition}

It is important to note that the principal may randomize the examination process.
We further denote that a mechanism has \emph{(examination) budget} $B$, if its total cost of examination does not exceed $B$ for any realized set of solutions and any corresponding strategies.

We aim to study the trade-off between the examination budget $B$ and the approximation factor of a mechanism.
Unlike the previous results in which the MSPM can examine all the candidates proposed by the agents and commit to the best one, if the examination budget is finite, the principal should selectively examine them.
In this context, we present the following variant of MSPM.
\begin{definition}[Randomized single-proposal mechanism]
In a \emph{randomized single-proposal mechanism (RSPM)} with budget $B$, a tie-breaking rule $\rho$, an eligible set $R_i \subset \Omega$ are announced by the principal for each agent $i$.
Each agent proposes up to one solution, and the principal selects $B$ candidates uniformly at random and adopts the one of them to maximize her utility.
\end{definition}
Note that RSPM with budget $B$ always examine at most $B$ candidates.
As it can be observed in the definition of the RSPM, once a candidate $\om \notin R_i$ is submitted by agent $i$, there is no chance of getting accepted by the mechanism.
This implies that there is no incentive for the agent to submit a candidate beyond the eligible set.
In this context, we impose the following mild assumption.
\begin{assumption}[Rationality]\label{as:rationality}
Given a RSPM $M$ with arbitrary tie-breaking rule $\rho$ and eligible sets $R_1,\ldots, R_n$, each agent $i$ for $i \in [n]$ does not play strategy $\sigma_i = \om$ if $\om \notin R_i$.
\end{assumption}

In Section ~\ref{section:no_cost} and \ref{section:no_cost_incomp}, we assume that the principal can examine  all candidates (at most $n$) and takes the best one as the winner, \ie the principal has a sufficient budget $B \ge n$.
To study the nontrivial effect of examination cost, we assume that the budget is \emph{not sufficient} to examine all the candidates, \ie $B < n$.
We analyze RPSM in the complete information setting by generalizing Theorem~\ref{theorem:multi-agent-unlimited-budget} in Section~\ref{section:no_cost} to the limited budget case.

\begin{theorem} \label{theorem:multi-agent-limited-budget}
Given any $\tau \ge 0$, let $R$ be an eligible set with threshold $\tau$ such that $R:= \{\omega: x(\omega) \ge \tau, \omega \in \Omega\}$ and $\rho$ be an arbitrary deterministic tie-breaking rule.
Let $X_{(n + 1)} = -\infty$.
Given $2\leq B < n$, RSPM with $\rho$ and
a~homogeneous eligible set $R$ has $\poa_a$ of at most
\begin{align*}
\Ex{X_{(1)}
-
X_{(2)}\IND\left[X_{(2)} \geq \tau > X_{(B+1)}\right] -  \tau\IND\left[X_{(1)} \geq \tau > X_{(2)}\right]
-
\sum_{e=B+1}^{n} X_{(e-B+2)}\IND\left[X_{(e)} \geq \tau > X_{(e+1)}\right]}.
\end{align*}
\end{theorem}
Its proof structure follows that of Theorem~\ref{theorem:multi-agent-unlimited-budget}, but is more involved due to the randomness in the mechanism.
Formal proof is deferred to Appendix~\ref{pf_theorem:multi-agent-limited-budget}.
This result depicts the trade-off between the budget and the the efficiency of a mechanism.
Indeed, the upper bound on $\poa$ presented above decreases as $B$ increases, and by plugging $B=n$ we can exactly recover the guarantee provided in Theorem~\ref{theorem:multi-agent-unlimited-budget}.
This is fairly intuitive as if the principal can evaluate more candidates, then she can commit to more efficient solutions.
Similar to Remark~\ref{rm:bm_pim_comp}, RSPM can be either BM or PIM depending on the choice of the eligible sets.
For PIM, we can directly obtain the following result.
\begin{corollary}\label{cor:limbud_notau}
    There exists a PIM having budget $B$ with $\poa_{a}$ of $\Ex{X_{(1)}} - \Ex{X_{(n-B+2)}}$.
\end{corollary}

Importantly, quantifying the additive approximation factor in Corollary~\ref{cor:limbud_notau} is intrinsically more challenging than the previous result
without the examination cost,
since we need to deal with the expected difference of two order statistics which are not consecutive.
One straightforward approach would be to simply sum up the bound for two consecutive ones, but we found this bound to be cumbersome and couldn't find any take-away.
Analyzing or simplifying this gap for arbitrary budget $B$ and deriving analogous results as in Section~\ref{section:no_cost} remains a major open problem.

Although RSPM has the benefit of its simplicity, however, it is not a clever choice as it examines the candidates in uniformly at random. A better approach might be to exploit the prior information on the distributions regarding the quality of candidates, which might be an interesting topic.

\section{Conclusion}
In this paper, we present a thorough study on efficient delegation mechanisms against multiple agents, spanning the analysis of Bayesian/prior-independent mechanisms in complete/incomplete information settings 
%
%
We mainly reveal that the competitive tension arising from the multiple agents enables us to obtain efficient prior-independent mechanisms, in stark contrast to the Bayesian single-agent mechanism by \citet{kleinberg2018delegated}.
Furthermore, the gain from competition largely depends on the information available to the agents and the correlation between the agent's and the principal's utility, \eg it significantly degrades in the incomplete information setting with strongly negatively correlated utility.
Closing the gap between the approximation factors in various regimes and analyzing exact Bayes Nash equilibriums along with its efficiency in the incomplete information setting remains as major open problems.
In perspective of model, exploring the impact of examination costs more in details, endogenizing agents' costs to search/choose solutions, and studying the strategic tension arise from a repeated interaction would be interesting directions as well.

\begin{acks}
The work is partially supported by DARPA QuICC, NSF AF:Small \#2218678, and NSF AF:Small \# 2114269. 
\end{acks}

\bibliographystyle{ACM-Reference-Format}
\bibliography{ref}


\begin{thebibliography}{32}


\ifx \showCODEN    \undefined \def \showCODEN     #1{\unskip}     \fi
\ifx \showDOI      \undefined \def \showDOI       #1{#1}\fi
\ifx \showISBNx    \undefined \def \showISBNx     #1{\unskip}     \fi
\ifx \showISBNxiii \undefined \def \showISBNxiii  #1{\unskip}     \fi
\ifx \showISSN     \undefined \def \showISSN      #1{\unskip}     \fi
\ifx \showLCCN     \undefined \def \showLCCN      #1{\unskip}     \fi
\ifx \shownote     \undefined \def \shownote      #1{#1}          \fi
\ifx \showarticletitle \undefined \def \showarticletitle #1{#1}   \fi
\ifx \showURL      \undefined \def \showURL       {\relax}        \fi
\providecommand\bibfield[2]{#2}
\providecommand\bibinfo[2]{#2}
\providecommand\natexlab[1]{#1}
\providecommand\showeprint[2][]{arXiv:#2}

\bibitem[Abolhassani et~al\mbox{.}(2017)]%
        {abolhassani2017beating}
\bibfield{author}{\bibinfo{person}{Melika Abolhassani}, \bibinfo{person}{Soheil
  Ehsani}, \bibinfo{person}{Hossein Esfandiari}, \bibinfo{person}{MohammadTaghi
  Hajiaghayi}, \bibinfo{person}{Robert Kleinberg}, {and}
  \bibinfo{person}{Brendan Lucier}.} \bibinfo{year}{2017}\natexlab{}.
\newblock \showarticletitle{Beating 1-1/e for ordered prophets}. In
  \bibinfo{booktitle}{\emph{Proceedings of the 49th Annual ACM SIGACT Symposium
  on Theory of Computing}}. \bibinfo{pages}{61--71}.
\newblock


\bibitem[Alaei et~al\mbox{.}(2012)]%
        {alaei2012online}
\bibfield{author}{\bibinfo{person}{Saeed Alaei}, \bibinfo{person}{MohammadTaghi
  Hajiaghayi}, {and} \bibinfo{person}{Vahid Liaghat}.}
  \bibinfo{year}{2012}\natexlab{}.
\newblock \showarticletitle{Online prophet-inequality matching with
  applications to ad allocation}. In \bibinfo{booktitle}{\emph{Proceedings of
  the 13th ACM Conference on Electronic Commerce}}. \bibinfo{pages}{18--35}.
\newblock


\bibitem[Alon and Spencer(2016)]%
        {alon2016probabilistic}
\bibfield{author}{\bibinfo{person}{Noga Alon} {and} \bibinfo{person}{Joel~H
  Spencer}.} \bibinfo{year}{2016}\natexlab{}.
\newblock \bibinfo{booktitle}{\emph{The probabilistic method}}.
\newblock \bibinfo{publisher}{John Wiley \& Sons}.
\newblock


\bibitem[Alonso et~al\mbox{.}(2014)]%
        {alonso2014resource}
\bibfield{author}{\bibinfo{person}{Ricardo Alonso}, \bibinfo{person}{Isabelle
  Brocas}, {and} \bibinfo{person}{Juan~D Carrillo}.}
  \bibinfo{year}{2014}\natexlab{}.
\newblock \showarticletitle{Resource allocation in the brain}.
\newblock \bibinfo{journal}{\emph{Review of Economic Studies}}
  \bibinfo{volume}{81}, \bibinfo{number}{2} (\bibinfo{year}{2014}),
  \bibinfo{pages}{501--534}.
\newblock


\bibitem[Alonso and Matouschek(2008)]%
        {alonso2008optimal}
\bibfield{author}{\bibinfo{person}{Ricardo Alonso} {and} \bibinfo{person}{Niko
  Matouschek}.} \bibinfo{year}{2008}\natexlab{}.
\newblock \showarticletitle{Optimal delegation}.
\newblock \bibinfo{journal}{\emph{The Review of Economic Studies}}
  \bibinfo{volume}{75}, \bibinfo{number}{1} (\bibinfo{year}{2008}),
  \bibinfo{pages}{259--293}.
\newblock


\bibitem[Armstrong and Vickers(2010)]%
        {armstrong2010model}
\bibfield{author}{\bibinfo{person}{Mark Armstrong} {and} \bibinfo{person}{John
  Vickers}.} \bibinfo{year}{2010}\natexlab{}.
\newblock \showarticletitle{A model of delegated project choice}.
\newblock \bibinfo{journal}{\emph{Econometrica}} \bibinfo{volume}{78},
  \bibinfo{number}{1} (\bibinfo{year}{2010}), \bibinfo{pages}{213--244}.
\newblock


\bibitem[Barlow and Proschan(1975)]%
        {barlow1975statistical}
\bibfield{author}{\bibinfo{person}{Richard~E Barlow} {and}
  \bibinfo{person}{Frank Proschan}.} \bibinfo{year}{1975}\natexlab{}.
\newblock \bibinfo{booktitle}{\emph{Statistical theory of reliability and life
  testing: probability models}}.
\newblock \bibinfo{type}{{T}echnical {R}eport}. \bibinfo{institution}{Florida
  State Univ Tallahassee}.
\newblock


\bibitem[Bechtel and Dughmi(2020)]%
        {bechtel2020delegated}
\bibfield{author}{\bibinfo{person}{Curtis Bechtel} {and}
  \bibinfo{person}{Shaddin Dughmi}.} \bibinfo{year}{2020}\natexlab{}.
\newblock \showarticletitle{Delegated stochastic probing}.
\newblock \bibinfo{journal}{\emph{arXiv preprint arXiv:2010.14718}}
  (\bibinfo{year}{2020}).
\newblock


\bibitem[Bechtel et~al\mbox{.}(2022)]%
        {bechtel2022delegated}
\bibfield{author}{\bibinfo{person}{Curtis Bechtel}, \bibinfo{person}{Shaddin
  Dughmi}, {and} \bibinfo{person}{Neel Patel}.}
  \bibinfo{year}{2022}\natexlab{}.
\newblock \showarticletitle{Delegated Pandora's box}.
\newblock \bibinfo{journal}{\emph{arXiv preprint arXiv:2202.10382}}
  (\bibinfo{year}{2022}).
\newblock


\bibitem[Bergemann and Morris(2016)]%
        {bergemann2016bayes}
\bibfield{author}{\bibinfo{person}{Dirk Bergemann} {and}
  \bibinfo{person}{Stephen Morris}.} \bibinfo{year}{2016}\natexlab{}.
\newblock \showarticletitle{Bayes correlated equilibrium and the comparison of
  information structures in games}.
\newblock \bibinfo{journal}{\emph{Theoretical Economics}} \bibinfo{volume}{11},
  \bibinfo{number}{2} (\bibinfo{year}{2016}), \bibinfo{pages}{487--522}.
\newblock


\bibitem[Bulow and Klemperer(1994)]%
        {bulow1994auctions}
\bibfield{author}{\bibinfo{person}{Jeremy~I Bulow} {and}
  \bibinfo{person}{Paul~D Klemperer}.} \bibinfo{year}{1994}\natexlab{}.
\newblock \bibinfo{title}{Auctions vs. negotiations}.
\newblock
\newblock


\bibitem[Cerone and Dragomir(2005)]%
        {cerone2005bounds}
\bibfield{author}{\bibinfo{person}{Pietro Cerone} {and}
  \bibinfo{person}{Sever~S Dragomir}.} \bibinfo{year}{2005}\natexlab{}.
\newblock \showarticletitle{Bounds for the Gini mean difference via the Sonin
  identity}.
\newblock \bibinfo{journal}{\emph{Computers \& Mathematics with Applications}}
  \bibinfo{volume}{50}, \bibinfo{number}{3-4} (\bibinfo{year}{2005}),
  \bibinfo{pages}{599--609}.
\newblock


\bibitem[Correa et~al\mbox{.}(2017)]%
        {correa2017posted}
\bibfield{author}{\bibinfo{person}{Jos{\'e} Correa}, \bibinfo{person}{Patricio
  Foncea}, \bibinfo{person}{Ruben Hoeksma}, \bibinfo{person}{Tim Oosterwijk},
  {and} \bibinfo{person}{Tjark Vredeveld}.} \bibinfo{year}{2017}\natexlab{}.
\newblock \showarticletitle{Posted price mechanisms for a random stream of
  customers}. In \bibinfo{booktitle}{\emph{Proceedings of the 2017 ACM
  Conference on Economics and Computation}}. \bibinfo{pages}{169--186}.
\newblock


\bibitem[David(1997)]%
        {david1997augmented}
\bibfield{author}{\bibinfo{person}{HA David}.} \bibinfo{year}{1997}\natexlab{}.
\newblock \showarticletitle{Augmented order statistics and the biasing effect
  of outliers}.
\newblock \bibinfo{journal}{\emph{Statistics \& probability letters}}
  \bibinfo{volume}{36}, \bibinfo{number}{2} (\bibinfo{year}{1997}),
  \bibinfo{pages}{199--204}.
\newblock


\bibitem[Elder(2016)]%
        {elder2016bayesian}
\bibfield{author}{\bibinfo{person}{Sam Elder}.}
  \bibinfo{year}{2016}\natexlab{}.
\newblock \showarticletitle{Bayesian adaptive data analysis guarantees from
  subgaussianity}.
\newblock \bibinfo{journal}{\emph{arXiv preprint arXiv:1611.00065}}
  (\bibinfo{year}{2016}).
\newblock


\bibitem[Fuchs et~al\mbox{.}(2022)]%
        {fuchs2022listen}
\bibfield{author}{\bibinfo{person}{William Fuchs}, \bibinfo{person}{Satoshi
  Fukuda}, {and} \bibinfo{person}{Mahyar Sefidgaran}.}
  \bibinfo{year}{2022}\natexlab{}.
\newblock \showarticletitle{Who to Listen to?: A Model of Endogenous
  Delegation}.
\newblock  (\bibinfo{year}{2022}).
\newblock


\bibitem[Gan et~al\mbox{.}(2022)]%
        {gan2022optimal}
\bibfield{author}{\bibinfo{person}{Tan Gan}, \bibinfo{person}{Ju Hu}, {and}
  \bibinfo{person}{Xi Weng}.} \bibinfo{year}{2022}\natexlab{}.
\newblock \showarticletitle{Optimal contingent delegation}.
\newblock \bibinfo{journal}{\emph{Journal of Economic Theory}}
  (\bibinfo{year}{2022}), \bibinfo{pages}{105597}.
\newblock


\bibitem[Gilligan and Krehbiel(1989)]%
        {gilligan1989asymmetric}
\bibfield{author}{\bibinfo{person}{Thomas~W Gilligan} {and}
  \bibinfo{person}{Keith Krehbiel}.} \bibinfo{year}{1989}\natexlab{}.
\newblock \showarticletitle{Asymmetric information and legislative rules with a
  heterogeneous committee}.
\newblock \bibinfo{journal}{\emph{American journal of political science}}
  (\bibinfo{year}{1989}), \bibinfo{pages}{459--490}.
\newblock


\bibitem[Hajiaghayi et~al\mbox{.}(2007)]%
        {hajiaghayi2007automated}
\bibfield{author}{\bibinfo{person}{Mohammad~Taghi Hajiaghayi},
  \bibinfo{person}{Robert Kleinberg}, {and} \bibinfo{person}{Tuomas Sandholm}.}
  \bibinfo{year}{2007}\natexlab{}.
\newblock \showarticletitle{Automated online mechanism design and prophet
  inequalities}. In \bibinfo{booktitle}{\emph{AAAI}}, Vol.~\bibinfo{volume}{7}.
  \bibinfo{pages}{58--65}.
\newblock


\bibitem[Holmstrom(1980)]%
        {holmstrom1980theory}
\bibfield{author}{\bibinfo{person}{Bengt Holmstrom}.}
  \bibinfo{year}{1980}\natexlab{}.
\newblock \bibinfo{booktitle}{\emph{On the theory of delegation}}.
\newblock \bibinfo{type}{{T}echnical {R}eport}.
  \bibinfo{institution}{Discussion Paper}.
\newblock


\bibitem[Kleinberg and Kleinberg(2018)]%
        {kleinberg2018delegated}
\bibfield{author}{\bibinfo{person}{Jon Kleinberg} {and} \bibinfo{person}{Robert
  Kleinberg}.} \bibinfo{year}{2018}\natexlab{}.
\newblock \showarticletitle{Delegated search approximates efficient search}. In
  \bibinfo{booktitle}{\emph{Proceedings of the 2018 ACM Conference on Economics
  and Computation}}. \bibinfo{pages}{287--302}.
\newblock


\bibitem[Krengel and Sucheston(1987)]%
        {krengel1987prophet}
\bibfield{author}{\bibinfo{person}{Ulrich Krengel} {and} \bibinfo{person}{Louis
  Sucheston}.} \bibinfo{year}{1987}\natexlab{}.
\newblock \showarticletitle{Prophet compared to gambler: an inequality for
  transforms of processes}.
\newblock \bibinfo{journal}{\emph{The Annals of Probability}}
  \bibinfo{volume}{15}, \bibinfo{number}{4} (\bibinfo{year}{1987}),
  \bibinfo{pages}{1593--1599}.
\newblock


\bibitem[Krishna(2001)]%
        {krishna2001asymmetric}
\bibfield{author}{\bibinfo{person}{Vijay Krishna}.}
  \bibinfo{year}{2001}\natexlab{}.
\newblock \showarticletitle{Asymmetric information and legislative rules: Some
  amendments}.
\newblock \bibinfo{journal}{\emph{American Political science review}}
  \bibinfo{volume}{95}, \bibinfo{number}{2} (\bibinfo{year}{2001}),
  \bibinfo{pages}{435--452}.
\newblock


\bibitem[Krishna and Morgan(2008)]%
        {krishna2008contracting}
\bibfield{author}{\bibinfo{person}{Vijay Krishna} {and} \bibinfo{person}{John
  Morgan}.} \bibinfo{year}{2008}\natexlab{}.
\newblock \showarticletitle{Contracting for information under imperfect
  commitment}.
\newblock \bibinfo{journal}{\emph{The RAND Journal of Economics}}
  \bibinfo{volume}{39}, \bibinfo{number}{4} (\bibinfo{year}{2008}),
  \bibinfo{pages}{905--925}.
\newblock


\bibitem[Lewis(2012)]%
        {lewis2012theory}
\bibfield{author}{\bibinfo{person}{Tracy~R Lewis}.}
  \bibinfo{year}{2012}\natexlab{}.
\newblock \showarticletitle{A theory of delegated search for the best
  alternative}.
\newblock \bibinfo{journal}{\emph{The RAND Journal of Economics}}
  \bibinfo{volume}{43}, \bibinfo{number}{3} (\bibinfo{year}{2012}),
  \bibinfo{pages}{391--416}.
\newblock


\bibitem[Li(2005)]%
        {li2005note}
\bibfield{author}{\bibinfo{person}{Xiaohu Li}.}
  \bibinfo{year}{2005}\natexlab{}.
\newblock \showarticletitle{A note on expected rent in auction theory}.
\newblock \bibinfo{journal}{\emph{Operations Research Letters}}
  \bibinfo{volume}{33}, \bibinfo{number}{5} (\bibinfo{year}{2005}),
  \bibinfo{pages}{531--534}.
\newblock


\bibitem[Lopez and Marengo(2011)]%
        {lopez2011upper}
\bibfield{author}{\bibinfo{person}{Manuel Lopez} {and} \bibinfo{person}{James
  Marengo}.} \bibinfo{year}{2011}\natexlab{}.
\newblock \showarticletitle{An upper bound for the expected difference between
  order statistics}.
\newblock \bibinfo{journal}{\emph{Mathematics Magazine}} \bibinfo{volume}{84},
  \bibinfo{number}{5} (\bibinfo{year}{2011}), \bibinfo{pages}{365--369}.
\newblock


\bibitem[Martimort and Semenov(2008)]%
        {martimort2008informational}
\bibfield{author}{\bibinfo{person}{David Martimort} {and}
  \bibinfo{person}{Aggey Semenov}.} \bibinfo{year}{2008}\natexlab{}.
\newblock \showarticletitle{The informational effects of competition and
  collusion in legislative politics}.
\newblock \bibinfo{journal}{\emph{Journal of Public Economics}}
  \bibinfo{volume}{92}, \bibinfo{number}{7} (\bibinfo{year}{2008}),
  \bibinfo{pages}{1541--1563}.
\newblock


\bibitem[Milgrom and Segal(2002)]%
        {milgrom2002envelope}
\bibfield{author}{\bibinfo{person}{Paul Milgrom} {and} \bibinfo{person}{Ilya
  Segal}.} \bibinfo{year}{2002}\natexlab{}.
\newblock \showarticletitle{Envelope theorems for arbitrary choice sets}.
\newblock \bibinfo{journal}{\emph{Econometrica}} \bibinfo{volume}{70},
  \bibinfo{number}{2} (\bibinfo{year}{2002}), \bibinfo{pages}{583--601}.
\newblock


\bibitem[Samuel-Cahn(1984)]%
        {samuel1984comparison}
\bibfield{author}{\bibinfo{person}{Ester Samuel-Cahn}.}
  \bibinfo{year}{1984}\natexlab{}.
\newblock \showarticletitle{Comparison of threshold stop rules and maximum for
  independent nonnegative random variables}.
\newblock \bibinfo{journal}{\emph{the Annals of Probability}}
  (\bibinfo{year}{1984}), \bibinfo{pages}{1213--1216}.
\newblock


\bibitem[Ulbricht(2016)]%
        {ulbricht2016optimal}
\bibfield{author}{\bibinfo{person}{Robert Ulbricht}.}
  \bibinfo{year}{2016}\natexlab{}.
\newblock \showarticletitle{Optimal delegated search with adverse selection and
  moral hazard}.
\newblock \bibinfo{journal}{\emph{Theoretical Economics}} \bibinfo{volume}{11},
  \bibinfo{number}{1} (\bibinfo{year}{2016}), \bibinfo{pages}{253--278}.
\newblock


\bibitem[Watt(2021)]%
        {watt2021concavity}
\bibfield{author}{\bibinfo{person}{Mitchell Watt}.}
  \bibinfo{year}{2021}\natexlab{}.
\newblock \showarticletitle{Concavity and Convexity of Order Statistics in
  Sample Size}.
\newblock \bibinfo{journal}{\emph{arXiv preprint arXiv:2111.04702}}
  (\bibinfo{year}{2021}).
\newblock


\end{thebibliography}
\clearpage

\appendix

\section{Numerical Experiments on the Tightness of Theorem \ref{thm:bayes_oblivious}}\label{app:c}
We here numerically verify that the lower bound on $\Pr{E}$ obtained in the proof of Theorem~\ref{thm:bayes_oblivious}, \ie $\exp\parans{-\frac{n^2}{2(n-1)^2}}$, is almost tight.
Figure~\ref{fig:num-exp-pim} represents $\Pr{E}$ with respect to the value of $n$ for some choices of $k$.
To obtain the true value of probability of the event $E$, we simulate the principal and agents and compute the frequency of the event that each agent's utility is aligned, \ie proposing a solution that maximizes the principal's utility also maximizes the agent's utility, for every agent.
To present some details in the experiment, we consider $n = \{5, 10, 20, 50, 100, 200, 300, 400\}$ and $k = \{10, 20, 100, 150, 1000\}$.
For each pair of $n$ and $k$, we run $30000$ times to compute the probability that the event $E$ happens.
We repeat this procedure $60$ times to obtain the variance for each setting, which is plotted as the shade in the figure.
The main take-away here is that the lower bound and the average value of the quantity $\Pr{E}$ is indeed becoming tight as $n$ increases.
As a side note, one can observe that as $n$ increases the precision error appears in each experiment (shade) becomes wider.
This is because large $n$ introduces more randomness in each experiment which requires more samples to stabilize its average, and thus we believe it not to be a meaningful phenomenon.
\begin{figure}[tbp]
\centering
    \hspace{-6mm}
    \includegraphics[width=12cm]{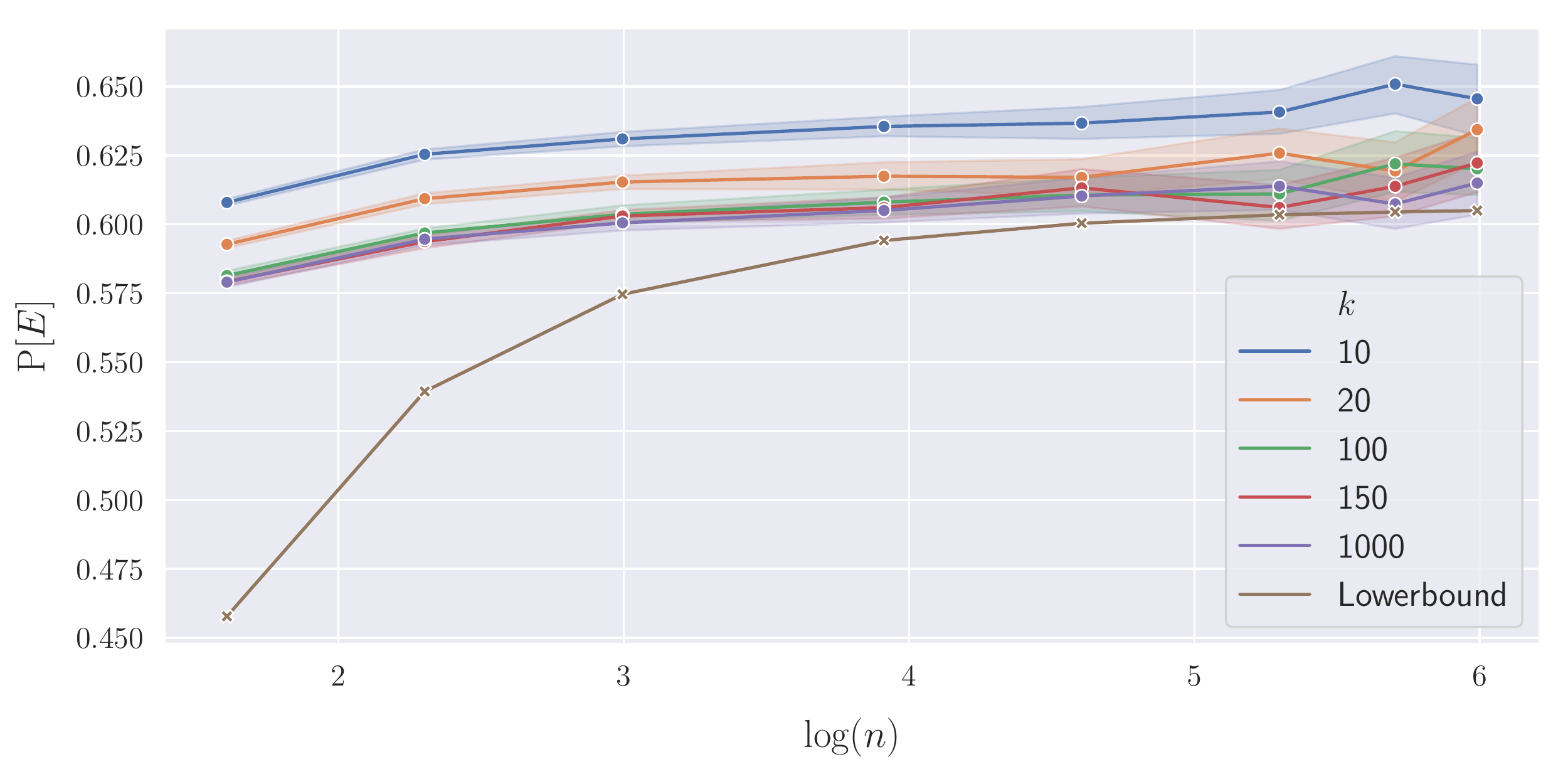} 
    \caption{
    Probability of the event $E$ defined in~\eqref{ineq:01231307}.
    Lowerbound refers to the lower bound of the event $E$ that we obtained in \eqref{ineq:01062049}, and $n$ and $k$ denotes the number of agents and the number of solutions per agent. We note that $x$-axis is presented in the base of the natural number $e$.
    }
    \label{fig:num-exp-pim}
\end{figure}
  
\section{Further discussion on feasible mechanisms}\label{sec:fur_fea_mec}
Another reasonable choice of mechanism will be to require each agent submitting certain number, say two, of solutions, \ie multi-proposal mechanism.
    Under this mechanism, each agent might select a pair of aggressive solution with large ex-post utility but smaller chance to be selected, and conservative solution with small ex-post utility but larger probability of winning, in order to balance these features of two solutions.
    More abstractly, one may consider a mechanism that requires each agent to signal some statistics of the observed solutions, \eg sample average or the number of eligible solutions, upon the solutions themselves.
    As astute readers may have noticed, in the complete information setting, however, the revelation principle style of result allows us to reduce any mechanism to single-proposal mechanism, so such generalization does not help. This is formalized in Theorem~\ref{thm:multiagent_revel}.
    Note that this reduction does not preclude any randomized mechanism.

\section{Bayes Correlated equilibrium}\label{sec:correlated_equi}
In this section, we discuss that the mechanism presented in Theorem~\ref{thm:bayes_oblivious} also has an efficient Bayes correlated equilibrium.
To be more specific, we are interested in the notion of Bayes correlated equilibrium \cite{bergemann2016bayes}.
Suppose that there exists a mediator who is able to access the realization of all the solutions observed by the agents.
Given the realization of all the solutions $\omb = (\omb_i)_{i \in [n]}$ observed by the agents, the mediator recommends a strategy $\psi_i: \Om^* \mapsto \Sigma_i$ for each agent $i \in [n]$.
It is important to note that the mediator recommends the strategy based on observing all the agents' solutions.
Given the recommended strategy $\psi_i$ and observed solution $\omb_i$, agent $i$ plays a strategy $\sigma_i(\omb_i, \psi_i): \Omega^* \times \Sigma_i \mapsto \Omega_i$.
The procedure of recommending strategies is specified as a part of the mechanism $M$, and we denote by $f_{M, \sigma, \psi}(\omb)$ the interim allocation function, \ie the selected winner, of the mechanism $M$ given $\psi(\omb)$ and $\sigma(\omb, \psi).$

Then, given $\psi_i(\omb)$ and $\omb_i$, agent $i$'s expected utility of playing $\sigma_i$ is
\begin{align*}
	u_i(\sigma_i(\psi_i, \omb_i), \sigma_{-i})
	=
	\Ex{\int_{\omb_{-i}} y_i(f_{M,(\sigma_i(\psi_i(\omb), \omb_i), \sigma_{-i})})f_{-i}(\omb_{-i})d\omb_{-i}},
\end{align*}
and agent $i$'s total expected utility of playing $\sigma_i$ is 
\begin{align*}
	u_i (\sigma_i, \sigma_{-i}) = \Exu{\omb_i}{u_i(\sigma_i(\psi_i, \omb_i), \sigma_{-i})}.
\end{align*}
We now define a Bayes correlated equilibrium as follows.
\begin{definition}[Bayes correlated equilibrium]
	We say that $\psi_i$ induces a \emph{Bayes correlated equilibrium} (BCE) if for any $i \in [n]$,
	\begin{align*}
		u_i(\psi_i, \sigma_{-i}) \ge u_i(\sigma'_i, \sigma_{-i}),	
	\end{align*}
	for any other strategy $\sigma'_i$, \ie there exists no incentive to deviate from the mediator's recommendation.
\end{definition}
Compared to BNE, due to the existence of the additional signals provided by the mediator who can observe all the solutions, both the principal and agents may enjoy the increased utility.
In the following theorem, we prove that there exists a PIM such that it achieves asymptotically constant multiplicative approximation factor for some BCEs.
We provide the proof in Appendix~\ref{pf_thm:bce}
\begin{theorem}\label{thm:bce}
	Consider a symmetric and independent setting such that the distributions of the both principal's and agent's utility follow $U[0,1]$.
	Then, there exists an $\parans{\exp\parans{\frac{n^2}{2(n-1)^2}},0}$-approximate PIM and corresponding BCE.
\end{theorem}

\section{Bulow-Klemperer-style result for prior-independent mechanism}\label{sec:BKS}
In constructing efficient PIM, one standard way would be to construct PIM from BM by applying the Bulow-Klemperer-style result by \citet{bulow1994auctions}, which considers a fictitious BM with $n+1$ agents to find a good upper bound on the efficiency of PIM with $n$ agents.
More precisely, they first show the existence of PIM (auction in their case) which always allocates the item to some bidders and then prove that such PIM is optimal among the mechanism that always allocates the item.
Then, for any Bayesian mechanism $M$ with an arbitrary reserve price, one can consider a fictitious auction with $n+1$ agents such that it runs $M$ with the first $n$ agents, and if it does not allocate the item, then it allocates the item to $n+1$-th agent.
Given the observations above, one can conclude that the PIM with $n$ agents defined above outperforms arbitrary BM with $n+1$ agents, and thus we can derive the approximation factor of PIM based on BM.
This technique, however, does not work in our case as we cannot easily argue that MSPM that admits all the solutions induces the optimal utility for any mechanism that always adopts the winner.
This is because we can consider a PIM that always admits all the solutions for some agents, but enforce some nontrivial (which does not always accept any solution) eligible sets for other agents. 
Indeed, this mechanism may outperform MSPM that admits all solutions, since it may motivate the other agents with nontrivial eligible sets to propose more in favor of the principal, while maintaining the mechanism to always delegate to some agents.
\section{Comparing multi-agent and single-agent delegation}\label{sec:compare_mul_sin}
To give a brief intuition in how one can reasonably compares the multi-agent and single-agent setting in terms of the principal, we start with an example.
\begin{example}\label{ex:cs}
    Consider a task requester in the online labor market, \eg crowdsourcing platform.
    The task requester wants to solve a task by finding a single best solution for the task, and she is capable of recruiting five online workers.
    Due to the intrinsic volatility of the online platform, each worker is equipped with a similar level of skill, and is expected to exert a similar level of effort per some unit of payment.
    As such, each worker has the same \emph{certainty equivalent} in coming up with a single solution, saying $\$10$.
    Because of the anonymity of the online labor market, each worker does not exactly obey the instruction of the task therein to answer the task, but instead propose a solution of his own interest, \ie misaligned payoffs.
    The task requester has a budget of $\$100$, and faces two options of (i) delegating to a single agent and (ii) delegating to five agents.
    The first option is to solely depend on a single worker, and the worker will sample $10$ solutions and commit to one of them.
    By exploiting the second option, five workers will observe two solutions each and commit to one of them.
    Given the assumptions above, it is likely that the expected quality of a single solution will be similar, i.e. workers are symmetric.
    Which option should the principal take?
    As observed in our previous results, delegating to multiple workers is expected to be generally better, but to any extent?
\end{example}
In what follows, we answer this question in  quantitative manner.
More formally, we provide comparative statics given the reasonable comparison between the single-agent and multi-agent setting, wherein our derivation of reasonable comparison follows from the intuition described in the example above.
Note that we obtain the result of similar flavor in Theorem~\ref{thm:recover}, but it only accommodates a reduction from a multi-agent instance to single-agent instance entirely to provide an approximation ratio of the multi-agent setting.
In this section however, we cement this result by showing that given any single-agent instance with arbitrary mechanism, we can construct a reasonable multi-agent instance such that there exists a multi-agent mechanism with principal's utility at least that under the original single-agent instance.
This result implies given any single-agent delegation problem instance, the principal can instead delegate to a reasonable set of multiple agents and obtain a superior utility. 
To this end, we first introduce how we construct a multi-agent instance from a single-agent instance.
\begin{definition}[Multi-agent correspondence]
Consider a single-agent problem with a distribution $D$, a number of solutions $k$, the agent's utility function $y$, and the principal's utility function $x$.
Consider a multi-agent problem with $n \ge 2$ agents in which the principal has the same utility function $x$ and each agent $i$ has a distribution $D_{i,j} = D$ for $j \in [k_i]$, a number of solutions $k_i$ such that $\sum_{i=1}^nk_i = k$, and the utility function $y_i = y$. We denote this multi-agent problem by multi-agent correspondence of the single-agent problem.
\end{definition}
The above construction of multi-agent instance is reasonable.
This is because to see whether delegating to the multi-agent indeed brings advantages, it is fair to compare under the assumption that the agents exert the same total amount of effort in the task in both the single-agent and the multi-agent problem.
Since the number of solutions can be translated as the amount of effort exerted in the task, our definition of the multi-agent correspondence would play a reasonable role in comparing the single-agent and the multi-agent problem.
We present the proof in Appendix~\ref{pf_theorem:multi-agent-correspondence}

\begin{theorem}\label{theorem:multi-agent-correspondence}
Given a single-agent problem $P$, denote its multi-agent correspondence $P'$.
For any mechanism $M$ under $P$, there exists an MSPM under $P'$ such that the principal's utility of the MSPM under $P'$ is at least that of  $M$ under $P$.
\end{theorem}

The main take-away is that the principal can actually increase (at least not decrease) her utility by recruiting more agents with similar utilities to delegate the task, although this does not quantify whether the multi-agent delegation largely improves upon the single-agent setting.
\begin{remark}
    As we discussed in Section~\ref{sec:cost}, delegating to multiple agents induce a trade-off in examining the candidates in terms of the principal's resource, \eg time.
    Meanwhile, delegating to multiple agents already has its advantage in time consumption in the sense that it would take much smaller time for the entire process to be finished in the multi-agent delegation as each agent's task solving process is independent from each other and thus can be parallelized.
    We discuss more in the following example.
\end{remark}
\begin{example}
    Back to Example~\ref{ex:cs}, it is indeed the case the task requester can simultaneously request five workers to solve the task and save her time.
    Assuming that obtaining a single solution takes around five minutes.
    Roughly speaking, single-agent delegation with $10$ solutions requires $50$ minutes, whereas delegating to five agents with two solutions would take $10$ minutes.
    Afterward, the principal will examine the task.
    Suppose that the principal requires $t$ minutes to examine a single solution.
    In the single-agent scenario, overall time consumption is $50+t$.
    Note that $t$ is added back as the principal indeed has to examine the solution even though it is only one as the use of Bayesian mechanism (and thus decide whether or not to accept it) is essential in the single-agent setting.
    Meanwhile, the five agents scenario require $10+5t$.
    For any $t \le 10 = 5\cdot 2$, \ie when the principal can examine the solutions faster than twice of the time that each agent needs to find one solution, delegating to five agents still carries its own benefit even under the total examination.
    This will be even more true when the production function of the worker is concave, \ie, when they need more time to find a solution as they've already found some solutions.
    In short, multi-agent delegation may have both time and efficiency gains to some extents. 
\end{example}



\section{Proofs}\label{sec:rem_proofs}
\subsection{Proof of Theorem \ref{thm:multiagent_revel}}\label{pf_thm:multiagent_revel}
\begin{proof}
Let $R_i$ be the set of all the solutions belonging to agent $i$ that can be selected by the mechanism $M$, over any feasible $\omb \in \Omega^*$ and corresponding equilibrium strategies, and any random bits of the mechanism.
Given a set of observed solutions by the agents $\omb = \cup_{i \in [n]} \omb_i \in \Omega^*$ and corresponding equilibrium strategy $\sigma$, consider any  $\omega_M \in F_{M,\sigma}(\omb)$.

Suppose that $\om_M$ belongs to agent $i$.
Obviously, $\om_M \in R_i$ holds due to our construction.
Furthermore, in terms of agent $j\neq i$, there exists no solution $\om' \in \omb_j$ such that $x(\om) > x(\om_M)$ since if so, agent $j$ will be selected in $M$ by playing $\om$, \ie for any $\om' \in \omb \setminus \omb_i$,
\begin{align}
x(\om_M) \ge x(\om').	\label{ineq:12051649}
\end{align}
Furthermore, if there exists another $\om \in \omb_i$ such that $x(\om) \ge x(\om')$ for all $\om' \in \omb \setminus \omb_i$ and $y_i(\om) > y_i(\om_M)$, then agent $i$'s utility becomes larger by playing $\om$, which contradicts the assumption that $\sigma_i = \om_M$ is an equilibrium strategy.
Hence, we conclude that whenever such $\om$ exists,
\begin{align}
y_i(\om_M) \ge y_i(\om). \label{ineq:12051650}
\end{align}

Consider a MSPM with an arbitrary tie-breaking rule and eligible sets $R_i$ for each agent $i$.
Let this mechanism $M'$.
Let $\sigma'$ be arbitrary equilibrium strategies under $M'$.
Let $\om_{M'}$ be the winner of $M'$.
We now consider two cases (i) when $\om_{M'} \in \omb_i$ and (ii) $\om_{M'} \in \omb_j$ for some $j \neq i$.

Firstly, assume that $\om_{M'}$ belongs to agent $i$.
This means that $\sigma'_i = \om_{M'}$.
Since $\sigma'_i$ is an equilibrium strategy, there exists no solution $\om' \in \omb \setminus \omb_i$ such that $x(\om')>x(\sigma'_i)$ since if so, the agent who owns $\om$ will be selected in $M'$ by playing $\om'$ which contradicts that $\sigma'_i$ is an equilibrium strategy.
Hence, we have $x(\sigma'_i) \ge x(\om')$ for any $\om' \in \omb \setminus \omb_i$.
By~\eqref{ineq:12051650}, this implies that $y_i(\om_M) \ge y_i(\sigma'_i) = y_i(\om_{M'})$.
By our assumption of Pareto optimal play, we further have
\begin{align*}
	x(\sigma'_i) = x(\om_{M'})  \ge x(\om_M).
\end{align*}

Otherwise, suppose that $\om_{M'}$ belongs to agent $j \neq i$.
In this case, suppose that $x(\om_{M'}) < x(\om_{M})$.
Then, agent $i$ is strictly better off by playing $\om_M$.
Hence, this is not equilibrium, and it completes the proof.
\end{proof}

\subsection{Proof of Theorem \ref{theorem:multi-agent-unlimited-budget}}\label{pf_theorem:multi-agent-unlimited-budget}

\begin{proof}
We mainly prove that
the expected payoff of randomized single-proposal mechanism $M$
is at least
$$\Ex{X_{(2)}\IND[X_{(2)} \ge \tau]} + \Ex{\tau\IND\left[X_{(1)} \geq \tau > X_{(2)}\right]}.$$

We find a lower bound on the principal's utility in an arbitrary Nash equilibrium. 
Let $\omb = \cup_{i \in [n]}\omb_i$ be the entire set of solutions observed by the agents and $\sigma = (\sigma_1,\ldots, \sigma_n)$ be arbitrary Nash equilibrium.
We separately consider two cases (i) for any $\om \in \omb$, $\om \notin R$ and (ii) there exists $\om \in \omb$ such that $\om \in R$.

{\em Case (i).} 
In this case, the utility of all agents is $0$,
and there does not exist a winner. 
Recall that for any Nash equilibrium,
there exists no incentive for any agent to deviate from equilibrium strategies.
We claim that $X_{(1)} < \tau$, where it concludes that the principal's utility is simply zero in this case.
Suppose not, \ie $X_{(1)} \geq \tau$. 
Denote by $\om^*$ one of the candidates which
corresponds to the utility of $X_{(1)}$, and observed by agent $t$.
If agent $t$ proposes $\om^*$, his utility strictly increases from $0$ to $y(\om^*)$,
which contradicts that the agents have played equilibrium strategies.
Hence we conclude that $X_{(1)} < \tau$, and the principal's utility is $0$.

{\em Case (ii).} 
We assume that agents play mixed strategies, \ie
they decide on a probability distribution over the observed solutions.
Formally, let $p_{i, j}$ be the probability that agent $i$ proposes $\omega_{i, j}$.
These probability distributions form a mixed Nash equilibrium.
As agents play mixed strategies, for each $i \in \left[n\right]$ we have $\sum_{j=1}^{k_i} p_{i, j} = 1$.
Define $S_i = \{\omega_{i, j}: p_{i, j} > 0\}$.
For each agent $i$, let $b_i$ be his worst candidate in terms of principal utility, \ie
$b_i = \argmin_{\omega \in S_i}x(\omega)$.
If there are multiple candidates with the worst utility, arbitrarily choose one.
We now proceed with the following claim.
\begin{claim}
    Denote the set of agents who have non-zero utility by $A$, then $|A| = 1$.
\end{claim}
\begin{proof}[Proof of the claim]
    Suppose not.
Let $t$ be an agent $t \in A$ such that for each agent $i \in A-\{t\}$, either
\begin{itemize}
	\item $x(b_t) < x(b_i)$
	\item $x(b_t) = x(b_i)$, and principal prefers agent $i$ over agent $t$ using tie-breaking rule $\rho$,
\end{itemize}
\ie $t$ is the agent with the worst $b_t$ who is least preferred according to the tie-breaking rule.
Since agent $t \in A$ has non-zero utility and $b_t \in S_t$, proposing $b_t$ also gives a non-zero utility to agent $t$.
Let $\omega_{i, j} \in S_i$ be the proposed candidate by arbitrary agent $i \in A-\{t\}$.
This is because if there exists a solution with zero utility in $S_t$, then just removing $b_t$ from $S_t$ and distributing its assigned probability arbitrarily to other solutions will give a strictly higher utility for agent $t$.
For any $i \in A \setminus \{t\}$ and $\om_{i,j} \in S_i$, we observe that $x(\omega_{i, j}) \geq x(b_i)$.
In addition, we have either $x(b_t) < x(\omega_{i, j})$,
or $x(b_t) = x(b_i) = x(\omega_{i, j})$.
Since the principal prefers $i$ over $t$ according to the tie-breaking rule $\rho$, the candidate proposed by agent $t$, \ie $b_t$, is dominated by the candidate proposed by agent $i$, \ie $\omega_{i, j}$.
This implies that the principal never selects $b_t$ as the winner, and it contradicts the fact that agent $t$ has non-zero utility when proposing $b_t$.
Therefore, we can conclude that $|A| = 1$.
\end{proof}

Suppose that agent $i^* \in A$ is the only agent with non-zero utility.
In this case, all agents except agent $i^*$ have a utility of zero. 
Due to the definition of Nash equilibrium, for any $\sigma_i = \om$ for some $\om \in \omb_i$, $y_i(f_{M,(\sigma_i,\sigma_{-i})}) = 0$.
This means that the principal never selects $\om \in \omb_i$ for agent $i \neq i^*$.
Hence, for all $i \neq i^*, \omega_{i^*, j} \in S_{i^*}$,  we have
\begin{align*}
    x(\omega_{i^*, j}) \geq \max_{\om \in \omb_i} x(\om) = X_{i, (1)}
\end{align*}
%
Furthermore, for all $\omega_{i^*, j} \in S_{i^*}$ we have $x(\omega_{i^*, j}) \leq X_{i^*, (1)}$.
Since for $n-1$ agents $i$ ($i \neq i^*$) we have $X_{i, (1)} \leq x(\omega_{i^*, j})$, and $x(\omega_{i^*, j}) \leq X_{i^*, (1)}$
using the definition of $X_{(r)}$, for all $\omega_{i^*, j} \in S_{i^*}$ we get
$$
X_{(n)} \leq X_{(n-1)} \leq ... \leq X_{(2)} \leq x(\omega_{i^*, j}) \leq X_{(1)}.
$$
Also, in this case, all candidates in $S_{i^*}$ should be eligible as they give non-zero utility to agent $i^*$.
This further implies that for all $\omega_{i^*, j} \in S_{i^*}$, $\tau \leq x(\omega_{i^*, j}) \leq X_{(1)}$, and we have
$$\max\left(X_{(2)}, \tau\right) \leq x(\omega_{i^*, j}) \leq X_{(1)}.$$
This inequality shows that
\begin{itemize}
    \item if $X_{(2)} \geq \tau$ then $x(\omega_{i^*, j}) \geq X_{(2)}$.
	Hence, in this case, the principal's expected utility is at least $X_{(2)}$ as
	$\Ex{f_{M, \sigma}(\omb)} = \sum_{j\in [k_{i^*}], \om_{i^*,j} \in S_{i^*}} p_{i^*, j} x(\omega_{i^*, j}) \geq \sum_{j\in [k_{i^*}], \om_{i^*,j} \in S_{i^*}} p_{i^*, j} X_{(2)} = X_{(2)}$.
    \item if $X_{(1)} \geq \tau \geq X_{(2)}$ then $x(\omega_{i^*, j}) \geq \tau$.
	Hence, in this case, the principal's expected utility is at least $\tau$ as
	$\Ex{f_{M, \sigma}(\omb)} = \sum_{j\in [k_{i^*}], \om_{i^*,j} \in S_{i^*}} p_{i^*, j} x(\omega_{i^*, j}) \geq \sum_{j\in [k_{i^*}], \om_{i^*,j} \in S_{i^*}} p_{i^*, j} \tau = \tau$.
\end{itemize}

Therefore, by combining the results of case (i) and case (ii) the expected utility of the principal is at least
\begin{align}
	\Ex{X_{(2)}\IND\left[X_{(2)} \geq \tau\right]} + 
	\Ex{\tau\IND\left[X_{(1)} \geq \tau > X_{(2)}\right]}.\label{ineq:infbud_tau}
\end{align}

\end{proof}

\subsection{Proof of Theorem \ref{thm:recover}}\label{pf_thm:recover}
\begin{proof}
	Given the multi-agent delegation problem, we construct a single-agent instance and
    then use guarantees on the single-agent setting to prove our theorem. 
    Consider a problem instance of multi-agent delegation where $n$ agents propose solutions.
	In order to find solutions, agent $i$ samples a solution from each of his $k_i$ distributions $\{D_{i, j}\}_{j=1}^{k_i}$.
    Note that all solutions lie in $\Omega$.
	We construct a single-agent problem instance
	where we denote the only agent by $s$.
    In the single-agent instance, for all $i \in [n]$, agent $s$ samples $k_i$ solutions,
    one from each of distributions $\{D_{i, j}\}_{j=1}^{k_i}$,
    hence $\sum_{i=1}^{n} k_i$ samples on aggregate.
    
    Let $\om$ be a solution that agent $s$ observes from the distributions that belong to agent $i$.
    We represent this solution in single-agent instance with a pair $\left(\om, i\right)$.
    Importantly, we define the solution space $\Omega'$ in single-agent instance by $\Omega' := \Omega \times \left[n\right]$.
    In the constructed instance, the principal's utility of the solutions, \ie the function $x(.)$, does not change.
    Precisely, $x(\left(\omega, i)\right) := x(\omega)$.
    However, we define $y(\left(\omega, i)\right)$, \ie
    agent $s$ utility for $\om$ when observed by distributions of agent $i$, to be $y(\left(\omega, i)\right) := y_i(\om)$.
 
	According to \citet{kleinberg2018delegated}, there exists an SPM with proper eligible set $R$ under the constructed instance such that
	principal's expected utility is at least $\frac{1}{2}\Ex{X_{(1)}}$.
	We use eligible sets $R_i = R$ for all agents in the multi-agent instance.
	We claim that for any realization of solutions $\omb = \cup_{i=1}^{n} \omb_i$,
	the principal's utility in the multi-agent setting is at least that of the single-agent setting.
    This further implies that
	using eligible sets $R_i = R$, principal's expected utility in multi-agent settings is at least $\frac{1}{2}\Ex{X_{(1)}}$.
	
	Consider a realization of solutions $\omb$. Let $(\omega, i)$ be the winner in the single-agent setting.
	If $(\omega, i) = \perp$, then all observed solutions are not eligible.
	Formally, for all $i$ in $\left[n\right]$, and $j$ in $\left[k_{i}\right]$:
	$\omega_{i, j} \notin R$. In this case, there is no eligible solution in the multi-agent setting which
	indicates that the principal's utility in both settings is zero.
	
    On the other hand, if $(\omega, i) \neq \perp$, let $\omega^*$ be the winner in the multi-agent instance, we claim that $x(\omega^*) \geq x(\omega)$.
	Note that $(\omega, i)$ is the winner of the single-agent instance so $\om$ is the eligible solution that
    maximizes agent $s$ utility by giving her the utility of $y\left((\om, i)\right) = y_i(\om)$.
 
    Denote by $\sigma_i$ the solution proposed by agent $i$ in the multi-agent instance. 
    As both $\sigma_i, \omega \in R$, and $\omega$ is the eligible solution that maximizes $y_i$, then 
	$y_i(\omega) \geq y_i(\sigma_i)$.
    As a result, according to Assumption~\ref{as:pareto}, $x(\sigma_i) \geq x(\omega)$.
	There are two cases

	\begin{itemize}
		\item
		$\omega^*$ is sampled by $i$, \ie $\omega^* = \sigma_i$. As shown above $x(\omega^*) = x(\sigma_i) \geq x(\omega)$.
		\item
		$\omega^*$ is sampled by another agent $j \neq i$.
		As the principal selects the best candidates among all agents, $x(\omega^*) \geq x(\sigma_i)$.
		Using the inequality $x(\sigma_i) \geq x(\omega)$ we further conclude that $x(\omega^*) \geq x(\omega)$.
	\end{itemize}

	This finishes the proof.

        Finally, consider a correspondence between the multi-agent and single-agent setting such that the entire solutions in the multi-agent setting is sampled by the agent in the single-agent setting (as we did in the proof).
        Then, our proof in fact shows that the multi-agent problem can always bring higher (or equal) principal's utility than the single-agent problem equipped with any possible mechanism without money.
        This is because any mechanism in the single-agent setting can be reduced to the single-proposal mechanism by Lemma 1 in \citet{kleinberg2018delegated}, and we prove that our construction of the mechanism above yields a higher (or equal) utility for the principal than the single-proposal mechanism in the single-agent setting.

\end{proof}

\subsection{Proof of Theorem \ref{thm:lowerbound_comp_worst}}\label{pf_thm:lowerbound_comp_worst}
\begin{proof}
	Consider a multi-agent instance where $n$ agents propose solutions.
	If there exists an agent whose sampled solutions are always better than others agents' solutions,
	we call that agent a~\emph{super}~agent. Formally, we say that agent $i$ is super~agent
	if for each $j \in \left[k_i\right]$, each $i' \in \left[n\right] \setminus \{i\}$,
	and each $j' \in \left[k_{i'}\right]$ when $\omega_{i, j} \sim D_{i, j}, \omega_{i', j'} \sim D_{i', j'}$
	we have $x(\omega_{i, j}) \geq x(\omega_{i', j'})$.

	When a multi-agent delegation problem instance has a super agent,
	then whatever the super agent proposes is taken as the winner as all other agents' proposed solutions
    are dominated by any solution of the super agent.
	As a result, the multi-agent instance reduces to a single-agent instance where
	only agent $i$ exists as the single agent.

	We now show that $(2,0)$-approximation is tight for Bayesian mechanisms.
	We use the notion of a super agent and construct a multi-agent delegation problem where
	there exists a super agent so that only his sampled solutions are important.
	We note that $X_{(1)}$ here corresponds to the solution observed by the super agent as other agents will never sample the best solution.

	Assume this agent samples two solutions from distributions $D_1$, and $D_2$.
    \begin{itemize}
        \item 
        $D_1$ always gives the same solution $\omega_1$ to the agent such that $x(\omega_1) = 1, y(\omega_1) = 10$.
        \item 
        $D_2$ gives the agent solution $\omega_{21}$ such that $x(\omega_{21}) = \frac{1}{\alpha}, y(\omega_{21}) = 5$
    	with probability of $\alpha$ while with probability of $1-\alpha$,
    	it gives solution $\omega_{22}$ to the agent such that $x(\omega_{22}) = \frac{\alpha}{1-\alpha}, y(\omega_{22}) = 5$.
        Note that $\alpha$ is a very small positive number.
    \end{itemize}
	There are two possibilities for sampled solutions by this agent
	either it is $\{\omega_{1}, \omega_{21}\}$, or it is $\{\omega_{1}, \omega_{22}\}$.
	We first consider $\Ex{X_{(1)}}$
	\begin{align*}
		\Ex{X_{(1)}} = \alpha \parans{\frac{1}{\alpha}} + \parans{1 - \alpha}\parans{1} = 2-\alpha.
	\end{align*}
	
	Now we consider the principal's policy. There are two cases
	\begin{itemize}
		\item
		principal accepts $\omega_{1}$, \ie $\omega_{1} \in R$ then as $y(\omega_{1}) > y(\omega_{21}) = y(\omega_{22})$,
		agent proposes $\omega_{1}$ no matter what he samples from $D_2$. In this case, the principal's utility is $1$.
		\item
		principal rejects $\omega_{1}$, \ie $\omega_{1} \notin R$. In this case, the principal accepts the other solution as it gives a non-zero utility to her.
		With probability of $\alpha$, $\omega_{21}$ is sampled by the agent which gives the utility of $\frac{1}{\alpha}$ to the principal,
		however, $\omega_{22}$ is sampled with probability of $1-\alpha$ which gives the utility of $\frac{\alpha}{1-\alpha}$ to the principal.
		Therefore, prinicpal's expected utility using this policy is $\alpha \parans{\frac{1}{\alpha}} + \parans{1-\alpha}\parans{\frac{\alpha}{1-\alpha}}= 1 + \alpha$.
	\end{itemize}
	Denote by $ALG$ the principal's utility using her policy. In both cases, the principal's expected utility $\Ex{ALG} \le 1 + \alpha$. By putting $\alpha := \frac{\eps}{3-\eps}$, we have
	\begin{align*}
		(2-\eps)\Ex{ALG} \le (2-\eps)(1 + \alpha) =
        (2-\eps)\left(1 + \frac{\epsilon}{3-\epsilon}\right) = 2 - \frac{\eps}{3-\eps} = 2 - \alpha = \Ex{X_{(1)}}.
	\end{align*}
    Therefore, $(2-\eps)\Ex{ALG} \le \Ex{X_{(1)}}$. This indicates that the multiplicative approximation factor is at least $2-\eps$, which finishes the proof.
\end{proof}

\subsection{Proof of Theorem \ref{thm:informed}}\label{pf_thm:informed}
\begin{proof}
	Define $R_i = R = \{\om : x(\om) \ge r(n,k)\}$.
	Suppose that we run MSPM with the eligible set $R$.
	Due to Theorem~\ref{theorem:multi-agent-unlimited-budget}, the principal's expected utility is at least
	\begin{align}
		\Ex{X_{(2)}\IND[X_{(2)} \ge \tau]} + \Ex{\tau \IND[X_{(1)} \ge \tau > X_{(2)}]}
		&\ge
		\Ex{\tau \IND[X_{(1)} \ge \tau]}
		\\
		&\ge
		\Pr{\exists i : X_{i,(1)} \ge r(n,k)}r(n,k)
		\nonumber
		\\
		&=
		\parans{1 - F(r(n,k))^{nk}}r(n,k)
		\nonumber
		\\
		&\ge
		\parans{1 - \parans{r(n,k)^{\alpha}}^{nk}}r(n,k)
		\nonumber
		\\
		&=
		r(n,k) - r(n,k)^{\alpha nk+1},
		\label{ineq:01041212}
	\end{align}
	where the second inequality comes from our assumption on $F(x) \le x^{\alpha}$.
	Plugging $r(n,k) = \parans{\frac{1}{\alpha nk + 1}}^{\frac{1}{\alpha nk}}$,\footnote{Note that this choice of $r(n,k)$ comes from derivating \eqref{ineq:01041212} w.r.t. $r(n,k)$ by considering it as a function of $\alpha nk$.} we can further obtain
	\begin{align*}
		\parans{\frac{\alpha nk}{\alpha nk +1}}^{\frac{1}{\alpha nk}},
	\end{align*}
	which yields the multiplicative approximation ratio of $\parans{\frac{\alpha nk}{\alpha nk +1}}^{\frac{-1}{\alpha nk}}$.
\end{proof}

\subsection{Proof of Theorem \ref{thm:comp_obl_neg}}\label{pf_thm:comp_obl_neg}
\begin{proof}
        In order to prove this theorem,
	we reuse the idea of super agent defined in the proof of   Theorem~\ref{thm:lowerbound_comp_worst}.
        Before that, we claim that the principal cannot announce an eligible set other than $\Omega$.
        The proof is based on straightforward arguments but we provide it to make it self-contained.
        \begin{claim}\label{cl:3}
            If $R_i \neq \Omega$ for some $i \in [n]$, then the additive price of anarchy can be arbitrarily close to $L$, \ie for any $\eps>0$ there exists a problem instance such that $\pos_a \ge L - \eps$.
        \end{claim}
        \begin{proof}[Proof of the claim]
            Since $R_i \neq \Omega$, there must be an element $\om \in \Omega$ such that the principal does not admit.
            Assume that $x(\om) = L-\eps/2$, and agent $i$'s distribution is a point mass on $\om$.
            Further, assume that other agents $j \neq i$ have distributions such that the principal's utility is at most $\eps/2$ for any point in their support.
            Then, our MSPM with the given eligible sets only accepts some solutions from $j \neq i$, which results in the principal's utility of at most $\eps/2$.
            Meanwhile, the optimal utility will be to select agent $i$'s solution $\om$ under which the principal can obtain $L-\eps$.
            Hence, the additive approximation factor is $L-\eps/2 - \eps/2 = L-\eps$, and it completes the proof.
        \end{proof}
        Due to the claim above, the principal should announce no eligible set in order to derive a nontrivial approximation ratio.
	Consider a multi-agent instance such that there exists a super agent. 
	Because the principal does not have any information regarding the distribution of agents,
	she cannot announce any eligible sets
	which implies that she accepts all solutions
	and takes the best one among the proposed solutions. 
	As a result, whatever the super agent proposes is selected by the principal.
	Assume the super agent samples two solutions from distributions $D_1$ and $D_2$.
	$D_1$ always gives the same solution $\omega_1$ to the agent such that $x(\omega_1) = L - \eps/2, y(\omega_1) = 5$.
	$D_2$ always gives the agent solution $\omega_{2}$ such that $x(\omega_{2}) = \eps/2, y(\omega_{2}) = 10$.
	We obtain
	\begin{align*}
		\Ex{X_{(1)}} = \Ex{\max\{\eps/2, L - \eps/2\}} = L -\eps/2.
	\end{align*}
	The super agent proposes $\omega_2$ as it will not be rejected by the principal
	and gives a higher utility to him, which means that the principal's utility is $x(\omega_2) = \eps/2$.
	Therefore, the additive approximation factor of this mechanism is $\Ex{X_{(1)}} - x(\omega_2) = L-\eps$. This finishes the proof.
\end{proof}

\subsection{Proof of Theorem \ref{thm:comp_sym_general}}\label{pf_thm:comp_sym_general}
\begin{proof}
        For ease of presentation, assume that $L=1$.
	By our assumption on symmetry and since each agent independently samples solutions, $X_{i,(1)}$ for $i \in [n]$ are identical and independent.
	Then, the result directly follows from the following theorem.
	\begin{lemma}[\citet{lopez2011upper}]\label{lm:lopez}
		Let $X_1,\ldots, X_n$ be i.i.d. random variables supported on $[0,1]$. Given an integer $s$ such that $1 \le s \le n-1$, let $Y_s$ and $Y_{s+1}$ be $s$-th and $s+1$-th order statistics of $X_1,\ldots, X_n$. Then, the following holds.
		\begin{align*}
			\Ex{Y_{s+1} - Y_s} \le \binom{n}{s}\left(\frac{s}{n}\right)^s \left(1-\frac{s}{n}\right)^{n-s}.
		\end{align*}
	\end{lemma}
	Plugging $s = n-1$ yields the desired result for $L=1$.
    By scaling the random variables, we can easily obtain the general of upper bound $L(1-\nicefrac{1}{n})^{n-1}$.

    To prove that this upper bound is tight upper bound for $\Ex{X_{(1)} - X_{(2)}}$, consider a symmetric problem instance equipped with a Bernoulli distribution with parameter $p$, \ie $X_{i,j} \sim \text{Bernoulli}(p)$.
	In this case, $X_{i,(1)}$ is also a Bernoulli distribution with parameter $q = 1-(1-p)^k$.
	Furthermore, we have
	\begin{align*}
		\Ex{X_{(1)} - X_{(2)}}
		&=
            \binom{n}{n-1}(1-q)^{n-1}q.
	\end{align*}
	If we set $q = 1-(1-p)^k = \frac{1}{n}$, which is equivalent to $p = 1 - (1-1/n)^{1/k}$, then we have the tight bound of $(1-1/n)^{n-1}$.
    By scaling this problem instance, we finish the proof.

\end{proof}

\subsection{Proof of Lemma \ref{lm:order_stat_diff}}\label{pf_lm:order_stat_diff}
\begin{proof}
		We separately prove two statements using similar techniques.
		Let $F$ and $f$ be the c.d.f. and p.d.f. of $D$, respectively.

		\paragraph{Part (1)}
		Rearranging $\Delta_{k,n} \le \Delta_{k,n-1}$, we need to show that
		\begin{align*}
			\Ex{X_{n-k+1:n} - X_{n-k:n-1}} \le \Ex{X_{n-k:n} - X_{n-k-1:n-1}}.
		\end{align*}
		Due to \citet{david1997augmented}, LHS and RHS can be represented as follows.
		\begin{align*}
			\Ex{X_{n-k+1:n} - X_{n-k:n-1}}
			&= 
			\binom{n-1}{n-k}\int_{-\infty}^{\infty}F^{n-k}(x)(1-F(x))^{k}dx.
			\\
			\Ex{X_{n-k:n} - X_{n-k-1:n-1}}
			&= 
			\binom{n-1}{n-k-1}\int_{-\infty}^{\infty}F^{n-k-1}(x)(1-F(x))^{k+1}dx.
		\end{align*}
		Plugging into the inequality and rearranging it, we need to show that
  
		\begin{align}
			\binom{n-1}{n-k-1}
			\parans{
			\int_{-\infty}^{\infty} F^{n-k-1}(x)(1-F(x))^k \left(\frac{n}{n-k}F(x) - 1\right)dx} \le 0.\label{ineq:12232248_new}
		\end{align}
		Define $G(t) = \int_{-\infty}^{t} F^{n-k-1}(x)(1-F(x))^{k-1} \left( \frac{n-1}{n-k}F(x) - 1\right)f(x)dx$.
		Obviously, $\frac{n-1}{n-k}F(x) - 1$ inside the integrand becomes positive for $x > x_0$ and negative otherwise for some $x_0 \in \R$.
		This implies that $G(t) \le G(\infty)$ for any $t \in \R$.
		Let $u = F(x)$ and by change of variables, we have
		\begin{align*}
			G(\infty)
			&=
			\int_{-\infty}^{\infty} u^{n-k-1}(1-u)^{k-1} \left( \frac{n}{n-k}u - 1\right)du
			\\
			&= 
			\frac{n}{n-k}B(n-k+1,k) - B(n-k,k)
			\\
			&=
			\frac{n}{n-k}\frac{(n-k)!(k-1)!}{n!} - \frac{(n-k-1)!(k-1)!}{(n-1)!}
			=0
		\end{align*}
		where $B(\cdot,\cdot)$ denotes the beta function.
		Hence, we have $G(t) \le 0$ for any $t \in \R$.
		Now we use the following lemma.
		\begin{lemma}[\citet{barlow1975statistical}]\label{lm:barlow}
			Given two functions $g:\R \mapsto [a,b]$ and $h:\R \mapsto [a,b]$ and $h$ is nonnegative, then the following holds.
			\begin{enumerate}
				\item If $\int_{-\infty}^t g(x)dx \le 0$ for any $t \in \R$, then for any nonincreasing function $h:\R \mapsto \R$, we have $\int_{-\infty}^t g(x)h(x)dx \le 0$ for any $t \in \R$.
				\item If $\int_{-\infty}^t g(x)dx \ge 0$ for any $t \in \R$, then for any nondecreasing function $h:\R \mapsto \R$, we have $\int_{-\infty}^t g(x)h(x)dx \ge 0$ for any $t \in \R$.
			\end{enumerate}
		\end{lemma}
		Plugging $g(x)$ to be the integrand of $G(x)$ and $h(x) = (\frac{f(x)}{1-F(x)})^{-1}$ to the first part of the above lemma, we conclude that ~\eqref{ineq:12232248_new} holds since $h(t)$ is nonincreasing due to the assumption on MHR.
		This completes the first part of the lemma.
		
		\paragraph{Part (2)}
		Now consider the MRHR distribution case.
		Rearranging $\delta_{k,n} \le \delta_{k,n-1}$, we need to show that
		\begin{align*}
			\Ex{X_{k+1:n} - X_{k+1:n-1}} \le \Ex{X_{k:n} - X_{k:n-1}}.	
		\end{align*}
		Again due to \cite{david1997augmented}, we have
		\begin{align*}
			\Ex{X_{k+1:n} - X_{k+1:n-1}}
			&=
			-\binom{n-1}{k}\int_{-\infty}^{\infty}F^{k+1}(x)(1-F(x))^{n-k-1}dx.
			\\
			\Ex{X_{k:n} - X_{k:n-1}}
			&=
			-\binom{n-1}{k-1}\int_{-\infty}^{\infty}F^{k}(x)(1-F(x))^{n-k}dx.	
		\end{align*}
		Plugging into the inequality and rearranging it, it suffices to show that
		\begin{align*}
			\frac{(n-1)!}{(k-1)!(n-k)!}\int_{-\infty}^{\infty}F^k(x)(1-F(x))^{n-k-1}\parans{\frac{n}{k}F(x) - 1}dx \ge 0.
		\end{align*}
		Define $G'(t) = \int_{t}^{\infty} F^{k-1}(x)(1-F(x))^{n-k-1}\parans{\frac{n}{k}F(x) - 1} f(x) dx$.
		Observe that $\parans{\frac{n}{k}}F(x) - 1$ becomes positive for $x > x_0$ and negative otherwise for some $x_0 \in \R$.
		Hence, $G'(t) \ge G'(-\infty)$.
		Now we show that $G'(-\infty) \ge 0$.
		Using change of variable with $F(x) = u$, we have
		\begin{align*}
			G'(-\infty)
			&=
			\int_{-\infty}^{\infty}u^{k-1}(1-u)^{n-k-1}\parans{\frac{n}{k}u - 1}du
			\\
			&=
			\frac{n}{k}B(k+1, n-k) - B(k, n-k)
			\\
			&=
			\frac{n}{k}\frac{k!(n-k-1)!}{n!} - \frac{(k-1)!(n-k-1)!}{(n-1)!}= 0,
		\end{align*}
		and it concludes that $G'(t) \ge 0$ for any $t \in \R$.
		Now, plugging $g(x)$ to be the integrand of $G'(x)$ and $h(x) = (\frac{f(x)}{F(x)})^{-1}$ to the second part of Lemma \ref{lm:barlow}, we conclude that ~\eqref{ineq:12232248_new} holds since $h(t)$ is nondecreasing due to the assumption on MRHR.
		This completes the first part of the lemma.

\end{proof}

\subsection{Proof of Lemma \ref{lm:mhr_preserving}}\label{pf_lm:mhr_preserving}
\begin{proof}
		Let $f(x)$ and $F(x)$ be p.d.f. and c.d.f. of the distribution $X_i \sim D$, and let $g(x)$ and $G(x)$ be that of $r$-th order statistics, respectively.
		Then, we have
		\begin{align*}
			1-G(x)
			&=
			\sum_{j=0}^{r-1}\binom{n}{j}F(x)^j(1-F(x))^{n-j}
			\\
			g(x)
			&=
			\frac{n!}{(r-1)!(n-r)!}f(x)F(x)^{r-1}(1-F(x))^{n-r}.
		\end{align*}
		Hence, its hazard rate can be obtained as
		\begin{align*}
			h(x)
			&=
			\frac{\frac{n!}{(r-1)!(n-r)!}f(x)F(x)^{r-1}(1-F(x))^{n-r}}{\sum_{j=0}^{r-1}\binom{n}{j}F(x)^j(1-F(x))^{n-j}}
			\\
			&= \frac{n!}{(r-1)!(n-r)!} \frac{f(x)}{1-F(x)}
			\frac{F(x)^{r-1}(1-F(x))^{n-r+1}}{\sum_{j=0}^{r-1}\binom{n}{j}F(x)^j(1-F(x))^{n-j}}
			\\
			&= \frac{n!}{(r-1)!(n-r)!} \frac{f(x)}{1-F(x)}
			\frac{1}{\sum_{j=0}^{r-1}\binom{n}{j}\left(\frac{1-F(x)}{F(x)}\right)^{r-j-1}}.
		\end{align*}
		Since $f(x)/(1-F(x))$ is nondecreasing by MHR property of $D$ and $1-F(x)/F(x)$ is nonincreasing function of $x$, we conclude that $h(x)$ is nondecreasing function of $x$.
\end{proof}

\subsection{Proof of Theorem \ref{thm:mhr_upperbound}}\label{pf_thm:mhr_upperbound}
\begin{proof}
	Recall that the additive approximation factor is upper bounded by $\Ex{X_{(1)} - X_{(2)}}$ by Corollary \ref{cor:infbud_notau}.
	As we discussed, the proof mainly follows from Lemma \ref{lm:order_stat_diff} on the expected difference between two consecutive order statistics.	
	First, by Lemma \ref{lm:mhr_preserving},  $X_{i,(1)}$ has MHR distribution for $i \in [n]$.
	Now, denote by $\Ex{X_{(1):n} - X_{(2):n}}$ the additive approximation factor in the problem setting when there exists $n$ agents.
	By repeating Lemma \ref{lm:order_stat_diff} with $k=1$, we can obtain the following inequalities.
	\begin{align*}
		\Ex{X_{(1):n} - X_{(2):n}} 
		\le
		\Ex{X_{(1):2} - X_{(2):2}}
		=
		\Ex{|A-B|},
	\end{align*}
	where $A$ and $B$ are two i.i.d. random variables of which the endowed distribution is equivalent to that of $X_{i, (1)}$.
	We now use the following inequality.\footnote{Weaker bound of $\Ex{|A-B|} \le \sqrt{\var(a)}$ can be obtained by simply rewriting it as $\Ex{\sqrt{(A-B)^2}}$ and applying Jensen's inequality.}
	\begin{lemma}[Theorem 5 in \cite{cerone2005bounds}]
		Let $A,B$ be i.i.d. random variables with mean $\mu$ and variance $\sigma^2$.
		For any $p,q$ such that $p>1$ and $\frac{1}{p} + \frac{1}{q} = 1$, suppose that $\Ex{|A - \mu|^{p}} < \infty$.
		Then, the following inequality holds.
		\begin{align*}
			\Ex{|A-B|} \le \frac{2}{(q+1)^{1/q}}\Ex{|A - \mu|^{p}}^{1/p}.
		\end{align*}
	\end{lemma}
	Plugging $p=q=2$ into the lemma, we obtain
	\begin{align*}
		\Ex{|A-B|} \le \sqrt{\frac{4\var(A)}{3}}.
	\end{align*}
	In case of uniform distribution $U[0,1]$, it is widely known that $X_{i,(1)} \sim Beta(k,1)$.
	Hence, we conclude that the additive approximation factor is bounded above by
	\begin{align*}
		\sqrt{\frac{4k}{3(k+1)^2(k+2)}} \approx \frac{1.155}{k}.
	\end{align*}
\end{proof}

\subsection{Proof of Lemma \ref{lm:order_stat_diff_restricted}}
\begin{proof}
		We separately prove two cases.
		
		\paragraph{Part (1)}
		Rearrange $(n+1)\Delta_{k,n} = \Delta_{k,n}  + n\Delta_{k,n} \le n\Delta_{k,n-1}$ and divide it by $n-1$,, we need to show that
		\begin{align*}
			\frac{\Delta_{k,n}}{n}  + \Ex{X_{n-k+1:n} - X_{n-k:n-1}} \le \Ex{X_{n-k:n} - X_{n-k-1:n-1}}.
		\end{align*}
		Again to \citet{david1997augmented}, each term can be represented as follows.
		\begin{align*}
			\frac{\Delta_{k,n}}{n}
			&=
			\frac{1}{n}\binom{n}{k}\int_{-\infty}^{\infty}F^{n-k}(x)(1-F(x))^{k} dx
			\\
			\Ex{X_{n-k+1:n} - X_{n-k:n-1}}
			&= 
			\binom{n-1}{n-k}\int_{-\infty}^{\infty}F^{n-k}(x)(1-F(x))^{k}dx.
			\\
			\Ex{X_{n-k:n} - X_{n-k-1:n-1}}
			&= 
			\binom{n-1}{n-k-1}\int_{-\infty}^{\infty}F^{n-k-1}(x)(1-F(x))^{k+1}dx.
		\end{align*}
		Plugging into the inequality and rearranging it, we need to show that
  
		\begin{align*}
			\binom{n-1}{n-k-1}
			\parans{
			\frac{1}{n-k}\int_{-\infty}^{\infty}F^{n-k}(x)(1-F(x))^{k}dx  + 
			\int_{-\infty}^{\infty} F^{n-k-1}(x)(1-F(x))^k \left( \frac{n}{n-k}F(x) - 1\right)dx} \le 0.
		\end{align*}
		We can further expand it to be
		\begin{align}
			\binom{n-1}{n-k-1}
			\parans{
			\int_{-\infty}^{\infty} F^{n-k-1}(x)(1-F(x))^k \left( \frac{n+1}{n-k}F(x) - 1\right)dx} \le 0.
			\label{ineq:01011741}
		\end{align}
		Define 
		\begin{align*}
			G(t)
			= 
			\int_{-\infty}^{t} F^{n-k-1}(x)(1-F(x))^k \left( \frac{n+1}{n-k}F(x) - 1\right)f(x) dx
		\end{align*}
		Obviously, $\frac{n+1}{n-k}F(x) - 1$ inside the integrand becomes positive for $x > x_0$ and negative otherwise for some $x_0 \in \R$.
		This implies that $G(t) \le G(\infty)$ for any $t \in \R$.
		Let $u = F(x)$ and by change of variables, we have
		\begin{align*}
			G(\infty)
			&=
			\int_{-\infty}^{\infty} u^{n-k-1}(1-u)^{k} \left( \frac{n+1}{n-k}u - 1\right)du
			\\
			&= 
			\frac{n+1}{n-k}B(n-k+1,k+1) - B(n-k,k+1)
			\\
			&=
			\frac{n+1}{n-k}\frac{(n-k)!k!}{(n+1)!} - \frac{(n-k-1)!k!}{n!} = 0
		\end{align*}
		where $B(\cdot,\cdot)$ denotes the beta function.
		Hence, we have $G(t) \le 0$ for any $t \in \R$.
		Plugging $g(x)$ to be the integrand of $G(x)$ and $h(x) = 1/f(x)$ to the first part of Lemma~\ref{lm:barlow}, we conclude that ~\eqref{ineq:01011741} holds since $h(t)$ is nonincreasing due to the assumption on $f(x)$.		
		This completes the first part of the lemma.
		
		\paragraph{Part (2)}
		Rearranging $(n+1)\delta_{k,n} \le n\delta_{k,n-1}$, we need to show that
		\begin{align*}
			\frac{1}{n}\Ex{X_{k+1:n} - X_{k:n}} + \Ex{X_{k+1:n} - X_{k+1:n-1}} \le \Ex{X_{k:n} - X_{k:n-1}}.	
		\end{align*}
		Due to \citet{david1997augmented}, we obtain the following equations.
		\begin{align*}
			\frac{1}{n}\Ex{X_{k+1:n} - X_{k:n}}
			&=
			\frac{1}{n}\binom{n}{n-k}\int_{-\infty}^{\infty}F^{k}(x)(1-F(x))^{n-k}dx
			\\
			\Ex{X_{k+1:n} - X_{k+1:n-1}}
			&=
			-\binom{n-1}{k}\int_{-\infty}^{\infty}F^{k+1}(x)(1-F(x))^{n-k-1}dx
			\\
			\Ex{X_{k:n} - X_{k:n-1}}
			&=
			-\binom{n-1}{k-1}\int_{-\infty}^{\infty}F^{k}(x)(1-F(x))^{n-k}dx
		\end{align*}
		Plugging into the inequality and rearranging it, we need to show that
		\begin{align*}
			\frac{(n-1)!}{(k-1)!(n-k)!}\parans{-\frac{1}{k}\int_{-\infty}^{\infty}F^k(x)(1-F(x))^{n-k}dx + \int_{-\infty}^{\infty}F^k(x)(1-F(x))^{n-k-1}\parans{\frac{n}{k}F(x) - 1}} \ge 0,
		\end{align*}
		which is equivalent to
		\begin{align*}
			\frac{(n-1)!}{(k-1)!(n-k)!}\parans{\int_{-\infty}^{\infty}F^k(x)(1-F(x))^{n-k-1}\parans{\frac{n+1}{k}F(x) - \frac{k+1}{k}}} \ge 0.
		\end{align*}
		Define $G'(t) = \int_{t}^{\infty}F^{k-1}(x)(1-F(x))^{n-k-1}\parans{\frac{n+1}{k}F(x) - \frac{k+1}{k}}f(x)dx$.
		Observe that $\frac{n+1}{k}F(x) - \frac{k+1}{k}$ becomes positive for $x > x_0$ and nonpositive otherwise for some $x_0 \in \R$.
		Hence, $G'(t) \ge G'(-\infty)$.
		Now we show that $G'(-\infty) \ge 0$.
		Using change of variable with $F(x) = u$, we have
		\begin{align*}
			G'(-\infty)
			&=
			\int_{-\infty}^{\infty}u^{k-1}(1-u)^{n-k-1}\parans{\frac{n+1}{k}u  - \frac{k+1}{k}}	du
			\\
			&=
			\frac{n+1}{k}B(k+1,n-k) - \frac{k+1}{k}B(
			k,n-k)
			\\
			&=
			\frac{n+1}{k}\frac{k!(n-k-1)!}{n!} - \frac{k+1}{k}\frac{(k-1)!(n-k-1)!}{(n-1)!}
			\\
			&=
			\frac{(k-1)!(n-k-1)!}{(n-1)!}\parans{\frac{n+1}{n} - \frac{k+1}{k}} \le 0,
		\end{align*}
		where the last inequality follows from $k \le n-1$.
		Hence we conclude that $G'(t) \ge 0$ for any $t \in \R$.
		Now, plugging in $g(x)$ to be the integrand of $G'(x)$ and $h(x) = f(x)^{-1}$ to the second part of Lemma \ref{lm:barlow}, we finish the proof.
\end{proof}

\subsection{Proof of Lemma \ref{lm:id_preserving}}\label{pf_lm:id_preserving}
\begin{proof}
		We separately prove the following two cases.
  
		\paragraph{Part (1)}
		Let $f(x)$ and $F(x)$ be p.d.f. and c.d.f. of the distribution $X_i \sim D$, and let $g(x)$ be p.d.f. of $n$-th order statistics, respectively.
		Then, we have
		\begin{align*}
			g(x)
			&=
			\frac{n!}{(n-1)!(n-n)!}f(x)F^{n-1}(x)(1-F(x))^{n-n} = nf(x)F(x)^{n-1}.
		\end{align*}
		Since both $f(x)$ and $F(x)$ are nondecreasing on $x$, we finish the proof.
		
		\paragraph{Part (2)}
		Let $g(x)$ be p.d.f. of first order statistics.
		Then, we have
		\begin{align*}
			g(x) = \frac{n!}{(1-1)!(n-1)!}f(x)F^{1-1}(x)(1-F(x))^{n-1} = nf(x)(1-F(x))^{n-1}.
		\end{align*}
		Since both $1-F(x)$ and $f(x)$ are nonincreasing on $x$, we finish the proof.
\end{proof}

\subsection{Proof of Theorem \ref{thm:decaying_upperbound}}\label{pf_thm:decaying_upperbound}
\begin{proof}
Recall that the additive approximation factor is upper bounded by $\Ex{X_{(1)} - X_{(2)}}$ by Corollary \ref{cor:infbud_notau}.
	As we discussed above, the proof mainly follows from Lemma \ref{lm:order_stat_diff_restricted} on the expected difference between two consecutive order statistics.
	By Lemma \ref{lm:id_preserving}, we know that $X_{i,(1)}$ has p.d.f. which is nondecreasing on $x$.
	Now, denote by $\Ex{X_{(1):n} - X_{(2):n}}$ the additive approximation factor in the problem setting when there exists $n$ agents.
	By repeating Lemma \ref{lm:order_stat_diff_restricted} with $k=1$, we have the following series of inequalities.
	\begin{align*}
		(n+1)\Ex{X_{(1):n} - X_{(2):n}} 
		&\le
		n \Ex{X_{(1):n-1} - X_{(2):n-1}}
		\\
		n\Ex{X_{(1):n-1} - X_{(2):n-1}} 
		&\le
		(n-1)\Ex{X_{(1):n-2} - X_{(2):n-2}}
		\\
		&\vdots
		\\
		4\Ex{X_{(1):3} - X_{(2):3}} 
		&\le
		3\Ex{X_{(1):2} - X_{(2):2}}.
	\end{align*}
	By multiplying the inequalities, we obtain
	\begin{align*}
		\Ex{X_{(1):n} - X_{(2):n}} \le \frac{3}{n+1}\Ex{|A-B|},
	\end{align*}
	where $A$ and $B$ are two i.i.d. random variables of which the endowed distribution is equivalent to that of $X_{i, (1)}$.
	Similarly from the proof of Theorem~\ref{thm:mhr_upperbound}, we can further obtain
	\begin{align*}
		\Ex{|A-B|} 
		\le
		\sqrt{\frac{4}{3}\var(A)},
	\end{align*}
	and it completes the proof.
	In case of uniform distribution $U[0,1]$, it is widely known that $X_{i,(1)} \sim Beta(k, 1)$.
	Hence, we conclude that the additive approximation factor is upper bounded by
	\begin{align*}
		\frac{3}{n+1}\sqrt{\frac{4k}{3(k+1)^2(k+2)}} \approx \frac{3.464}{k(n+1)}.
	\end{align*}

\end{proof}

\subsection{Proof of Theorem \ref{thm:pim_neg_incomp}}\label{pf_thm:pim_neg_incomp}
\begin{proof}
The proof consists of two parts.
First, we prove that there exists a problem instance such that proposing a solution that minimizes the principal's utility constructs an equilibrium, which leads to the principal's utility of $\Ex{\max_{i \in [n]}X_{i,(k)}}$.
Next, we analyze the bound on this quantity under the uniform distribution.

\paragraph{Worst-case instance}
Consider $n$ symmetric agents that find solutions $\omega$ such that
	$x(\omega) = x_0$ is sampled from a continuous distribution with c.d.f. of $F(x_0)$.
    Each agent samples $k$ solutions.
    For example, agent $i$ samples $\omega_{i, 1}, \omega_{i, 2}, ..., \omega_{i, k}$.
For an observed solution $\om_{i, j}$, define agents' utility function $y$ as a function of $x(\om_{i,j})$ as follows
\begin{align*}
		y(\omega_{i, j}) = \parans{1 - \parans{\parans{1 - F(x(\omega_{i, j}))}^k}^{n-1}}^{-2}.
\end{align*}
Under this setting, we claim that there exists a BNE such that each agent,
among all of his sampled solutions, proposes the solution which minimizes the principal's utility.
In order to show this claim, we prove that assuming other agents propose solutions with the lowest principal's utility,
best-response for agent $i$ is also to submit the solution which minimizes the principal's utility. Denote others' strategy by $\sigma_{-i}$.
\begin{align*}
    \Ex{u_i(\omega, \sigma_{-i}) | x(\omega) = x_0} = \Pr{\omega \text{ is winner} | \sigma_{-i}, x(\omega)=x_0} y(\omega)
\end{align*}
As all other agents propose solutions with the lowest principal's utility, the probability that $\omega$ is chosen as the winner is the probability that all agents observe at least one solution $\omega'$ such that $x(\omega') \le x(\omega) = x_0$.
Note that tie does not happen as $F$ is a continuous function.
Formally, 
\begin{align*}
	\Pr{\omega \text{ is winner} | \sigma_{-i}, x(\omega)=x_0} = \parans{1 - \parans{1 - F(x_0)}^k}^{n-1}.
\end{align*}
Plugging in the above equation, we have
\begin{align*}
	\Ex{u_i(\omega, \sigma_{-i}) | x(\omega) = x_0} = \parans{1 - \parans{1 - F(x_0)}^k}^{n-1} y(\omega) = \frac{1}{\parans{1 - \parans{1 - F(x_0)}^k}^{n-1}}
\end{align*}
This indicates that agent utility is a decreasing function of $x_0$ which further implies that agent $i$ best response is to propose solution $\omega$ such that $\omega = \argmin_{\omega \in \omb_i} x(\omega)$.
We prove that agent $i$ best-response is to propose the solution with the lowest principal's utility, which concludes that this set of strategies forms a BNE.
Hence in this case, the expected utility of the principal is $\Ex{\max_{i \in [n]} X_{i,(k)}}$.

\paragraph{Computing lower bound of $\poa_a$ for uniform distribution}
Now we derive an upper bound for this quantity on the uniform distribution $U[0,1]$ setting, \ie
$X_{i, j} \sim U[0, 1]$.
Note that $X_{i,(k)}$ follows the beta distribution with parameter $(1,k)$.
By Jensen's inequality,
\begin{align}
    e^{t\parans{\Ex{\max_{i \in [n]} X_{i,(k)}} - 1/(k+1)}}
    \le
    \Ex{e^{t \cdot \parans{\max_{i \in [n]} X_{i,(k)} -1/(k+1)}}}
    \le
    \sum_{i=1}^n \Ex{e^{t(X_{i,(k)} - 1/(k+1))}}.\label{ineq:01152042}
\end{align}
Now we use the fact that beta distribution is subgaussian with certain choices of parameters.
\begin{theorem}[Theorem 4 in \cite{elder2016bayesian}]
    The beta distribution with parameter $\alpha, \beta$ is $\frac{1}{4(\alpha+\beta) + 2}$-subgaussian.
\end{theorem}
Plugging $\alpha =1$ and $\beta=k$, we obtain that $X_{i,(k)}$ is $\frac{1}{4k+6}$-subgaussian.
Due to the definition of subgaussianity, we have
\begin{align*}
    \Ex{e^{t(X_{i,(k)})}} 
    \le e^{t^2/(8k+12)}.
\end{align*}
Plugging back to \eqref{ineq:01152042} and rearranging, we obtain
\begin{align*}
    \Ex{\max_{i \in [n]} X_{i,(k)}}
    \le
    \frac{1}{k+1} + \inf_{t > 0 } \frac{1}{t} \log \parans{n \cdot e^{t^2/(8k+12)}}
    =
    \frac{1}{k+1}  + \sqrt{\frac{\log n}{2k+3}}.
\end{align*}
Since $\Ex{X_{\max}} = \frac{nk}{nk+1}$, we conclude that the additive approximation factor is at least
\begin{align*}
    \poa_a \ge \frac{nk}{nk+1} - \frac{1}{k+1} - \sqrt{\frac{\log n}{2k+3}} \ge \frac{k-1}{k+1} - \sqrt{\frac{\log n}{2k+3}}.
\end{align*}
This concludes that there exists no mechanism such that $\poa_a < \frac{k-1}{k+1} - \sqrt{\frac{\log n}{2k+3}}$.
We note that by fixing $n$ and increasing $k$ sufficiently, we can make this bound arbitrarily close to $1$, and we finish the proof.
\end{proof}

\subsection{Proof of Theorem \ref{thm:bayes_oblivious}}\label{pf_thm:bayes_oblivious}
\begin{proof}
The proof consists of two main parts.
First, we construct an event $E$ such that there exists an $\Pr{E^c}$-approximate BNE which yields the optimal principal's utility.
Then, we analyze the probability of the event $E$ happening.

\paragraph{Existence of approximate BNE}
	Consider an MSPM with no eligible sets and a deterministic tie-breaking rule $\rho$.
	Define an event $E_i(\cdot)$ as follows.
	\begin{align*}
		E_i(\{x(\om), y(\om)\}_{\om \in \omb_i})
		=
		\paranm{\omb_i : \argmax_{\om \in \omb_i} x(\om)^{k(n-1)}y(\om) = \argmax_{\om \in \omb_i} x(\om)}.
	\end{align*}
	For notational simplicity, let $x(A) = \{x(\om)\}_{\om \in A}$ and similarly for $y(A)$.
        We often overwrite $E_i(\{x(\om), y(\om)\}_{\om \in \omb_i})$ by $E_i$ or $E_i(\omb_i)$. 
	Define 
	\begin{align}
		E(\omb) = E(x(\omb), y(\omb)) = \cap_{i \in [n]} E_i(x(\omb_i)).\label{ineq:01231307}
	\end{align}
        We also write $E(\omb)$ to denote $E$.
        Now, we will prove that given $E(x(\omb), y(\omb))$ and the other's strategies, proposing a solution that maximizes $x$ will be approximately best-response.
	Let's restrict our attention to the case in which $E$ happens given $\omb$.
	Suppose that agent $j \neq i$ proposes a candidate to maximize $x(\cdot)$, \ie $\sigma_{j} = \argmax_{\om \in \omb_j}x(\om)$.
	Then, conditioned on $E$, agent $i$'s expected payoff of proposing $\om \in \omb_i$ can be obtained as
	\begin{align*}
		u_i(\om, \sigma_{-i}|E(\omb))
		&=
		\Pr{\om \text{ is winner} \given E(\omb)}\cdot y(\om)
		\\
		&=
		\Pr{x(\om) \ge x(\om'), \forall\om' \in \omb_{-i}\given E(\omb)}\cdot y(\om)
		\\
		&= x(\om)^{k(n-1)}y(\om).
	\end{align*}
	Since we conditioned on $E(\omb)$, for agent $i$, the solution that maximizes $x(\om)^{k(n-1)}y(\om)$ is exactly the same to the solution that maximizes $x(\om)$.
	This verifies that conditioned on $E(\omb)$ and assuming that the other agent $j \neq i$ proposes a candidate that maximizes $x(\cdot)$, also maximizing $x(\cdot)$ is agent $i$'s best response.
	Let $\sigma_i^x$ be a strategy to propose a solution that maximizes $x(\cdot)$ for agent $i$.
 
	Given that all the other agents $j$ playing $\sigma_j^x$ for $j \neq i$, we can further obtain the following regarding agent $i$'s utility
	\begin{align*}
		u_i(\sigma_i^x, \sigma_{-i})
		&=
		\Ex{u_i(\sigma_i^x, \sigma_{-i})\given E}\Pr{E} + \Ex{u_i(\sigma_i^x, \sigma_{-i})\given E^c}\Pr{E^c}
		\\
		&\ge
		\Ex{u_i(\sigma_i^x, \sigma_{-i})\given E}\Pr{E}
		\\
		&\ge
		\Ex{u_i(\sigma'_i, \sigma_{-i})\given E}\Pr{E}
		\\
		&=
		\Ex{u_i(\sigma'_i, \sigma_{-i})} - \Ex{u_i(\sigma_i',\sigma_{-i}|E^c)}\Pr{E^c}
		\\
		&\ge
		\Ex{u_i(\sigma'_i, \sigma_{-i})} - \Pr{E^c},
	\end{align*}
	where the second inequality follows from the fact that given $E$, playing $\sigma_i^x$ is weakly dominant over any other strategy $\sigma'_i$ for agent $i$.
	This implies that if we characterize a good lower bound $\alpha$ such that $\Pr{E} \ge \alpha$, we have
	\begin{align*}
		u_i(\sigma_i^x, \sigma_{-i})
		\ge
		u_i(\sigma_i', \sigma_{-i}) - 1 + \alpha,
	\end{align*}
	which implies that $\sigma_i^x$ is $(1-\alpha)$-approximate BNE.

\paragraph{Analyzing the lower bound on $\Pr{E}$}
	Now, we compute the lower bound on $\Pr{E}$.
	Due to the independence of the utility distributions, we have
	\begin{align}
		\Pr{E(\omb)}
		= 
		\prod_{i \in [n]}
		\Pr{E_i(\omb_i)}.\label{ineq:01222241}
	\end{align}
We start with the following claim.
\begin{claim}\label{cl:4}
    $\Pr{E_i(\omb_i)} = \Ex{\prod_{j=2}^k \parans{1 - \frac{A_{i,j}}{2}}}$, where $A_{i,j} = \parans{\frac{X_{i,(j)}}{X_{i,(1)}}}^{k(n-1)}$ for $j \in [k]$.
\end{claim}
\begin{proof}[Proof of the claim]

 Note that each probability can be computed as the following.
	\begin{align*}
		\Pr{E_i(\omb_i)} 
		= 
		\int_{x(\omb_i)} \Pr{E_i(\omb_i) \given x(\omb_i)}f_x(x(\omb_i))dx(\omb_i).
	\end{align*}
	Given $x(\omb_i)$, we can further compute
	\begin{align*}
		\Pr{E_i(\omb_i) \given x(\omb_i)}
		&=
		\int_{y(\omb_i)}
		\Pr{E_i(x(\omb_i)) \given x(\omb_i), y(\omb_i)}f_y(y(\omb_i))dy(\omb_i)
		\\
		&=
		\Pr{\frac{y(\om^*_i)}{ y(\om)} \ge \parans{\frac{x(\om)}{x(\om_i^*)}}^{k(n-1)},\forall \om \in \omb_i \given x(\omb_i)}
		\\
		&=
		\Pr{\frac{Y_{i,(1)}}{Y_{i,(j)}} \ge \parans{\frac{X_{i,(j)}}{X_{i,(1)}}}^{k(n-1)}, \forall j \in [k] \given X_{i,j}, \forall j \in [k]}.
	\end{align*}
	Note that this probability is $1$ when $j=1$.
	Conditioned on the event $Y_{i,(1)} = t$, this probability can be further expanded as
	\begin{align*}
		\Pr{t\parans{\frac{X_{i,(1)}}{X_{i,(j)}}}^{k(n-1)}\ge Y_{i,(j)} \given X_{i,j},\forall j \in [k],Y_{i,(1)} = t}
		&=
		\min\parans{1, t\parans{\frac{1}{A_{i,j}}}}.
	\end{align*}
	Since each solution is i.i.d. sample, we have 
	\begin{align}
		\Pr{\frac{Y_{i,(1)}}{Y_{i,(j)}} \ge \frac{X_{i,(j)}}{X_{i,(1)}}, \forall j \in [k] \given X_{i,j}, \forall j \in [k]}
		&=
		\prod_{j=2}^k\parans{\int_{t=0}^{A_{i,j}}t\parans{\frac{1}{A_{i,j}}}dt + \int_{t=A_{i,j}}^1 1 dt}
		\nonumber
		\\
		&=
		\prod_{j=2}^k\parans{1-\frac{A_{i,j}}{2}}
		\label{ineq:01061949}
	\end{align}
	Hence we have
	\begin{align*}
		\Pr{E_i(\omb_i)} =
		\Pr{\frac{Y_{i,(1)}}{Y_{i,(j)}} \ge \frac{X_{i,(j)}}{X_{i,(1)}}, \forall j \in [k]}
		&=
		\Pr{\frac{Y_{i,(1)}}{Y_{i,(j)}} \ge \frac{X_{i,(j)}}{X_{i,(1)}}, \forall j \in [k] \given X_{i,j}, \forall j \in [k]}\Pr{X_{i,j},\forall j \in [k]}
		\\
		&=
		\int_{X_i} \prod_{j=2}^k\parans{1-\frac{A_{i,j}}{2}} dX_i
		\\
		&=
		\Ex{\prod_{j=2}^k\parans{1-\frac{A_{i,j}}{2}}}
	\end{align*}
\end{proof}

Now we prove that we can decompose the product inside the expectation as follows.
\begin{claim}\label{cl:5}
    $\Ex{\prod_{j=2}^k\parans{1-\frac{A_{i,j}}{2}}} \ge \int_{t=0}^1\prod_{j=2}^k\parans{\Ex{1-\frac{1}{2}\parans{\frac{X_{i,(j)}}{t}}^{k(n-1)} \given X_{i,(1)}= t}}\Pr{X_{i,(1)} =t}dt$.
\end{claim}
\begin{proof}[Proof of the claim]

	Given $X_{i,(1)} = t$, we note that $1-\frac{A_{i,j}}{2}$ for $j=2,\ldots,k$ is nonincreasing function over  $(X_{i,(2)},\ldots, X_{i,(k)})$.
	Now, we use the following correlation inequality called as \emph{Fortuin-Kasteleyn-Ginibre} (FKG) inequality.
	\begin{lemma}[\cite{alon2016probabilistic}]\label{lm:FKG}
		Let $\nu_1,\ldots, \nu_n$ be a probability measure on $\R$.
		If functions $f,g:\R^n \mapsto [0,1]$ are monotone increasing functions, then $\Exu{\nu}{f\cdot g} \ge \Exu{\nu}{f} \cdot \Exu{\nu}{g}$, where $\nu$ is a product measure over $\nu_1,\ldots, \nu_n$.
		The same inequality holds when both $f,g$ are monotone decreasing functions.
	\end{lemma}
        For given $i$, we have $f_j = 1-\frac{A_{i,j}}{2}$ for $j=2,3,\ldots, k$, each of which is a nonincreasing function over $X_{i,(j)}$.\footnote{Note in fact that we cannot directly plug in $f_j = 1-A_{i,j}/2$ in the inequality due to the correlation between the order statistics. That said, by rewriting the expression to be over the original random variables instead of order statistics, we can easily see that the same result follows, since the expectation remains for either of expressions.}
        Hence, by applying FKG inequality, we can further expand it as follows.
	\begin{align}
		\Ex{\prod_{j=2}^k \parans{1-\frac{A_{i,j}}{2}}}
		&=
		\int_{t=0}^1\Ex{\prod_{j=2}^k \paranl{1-\frac{1}{2}\parans{\frac{X_{i,(j)}}{t}}^{k(n-1)}} \given X_{i,(1)} = t}\Pr{X_{i,(1)} = t}dt
		\nonumber
		\\
		&\ge
		\int_{t=0}^1\prod_{j=2}^k\parans{\Ex{1-\frac{1}{2}\parans{\frac{X_{i,(j)}}{t}}^{k(n-1)} \given X_{i,(1)}= t}}\Pr{X_{i,(1)} =t}dt,
		\label{ineq:01022357}
	\end{align}
        and it completes the proof.
\end{proof}

Now we show that we can derive a constant lower bound that depends only on $k,n$ and $j$ on the quantity inside the expectation as follows.
\begin{claim}\label{cl:6}
    For $j =2,3,\ldots, k$, define $g_j(t)$ as follows.
    \begin{align*}
        g_j(t) = 	\Ex{1-\frac{1}{2}\parans{\frac{X_{i,(j)}}{t}}^{k(n-1)}\given X_{i,(1)} = t}.
    \end{align*}
    Then, we have the following lower bound on $g_j(t)$.
    \begin{align*}
        g_j(t) \ge 1 - \frac{1}{2}\parans{\frac{k-1}{kn-1}}^{j-1}.
    \end{align*}
\end{claim}
\begin{proof}[Proof of the claim]
We remark that $g_j(t)$ is a lower bound on the probability that $E_i(\omb_i)$ happens conditioned on $X_{i,(1)}=t$.
	Rewrite $g_j(t)$ as
	\begin{align*}
		\int_{x=0}^t \parans{1-\frac{1}{2}\parans{\frac{x}{t}}^{k(n-1)}}f_{X_{i,(j)} | X_{i,(1)}=t}(x)dx
	\end{align*}
	Now, the joint density function of $X_{i,(j)},X_{i,(1)}$ can be computed as
	\begin{align*}
		f_{X_{i,(j)},X_{i,(1)}}(x,t)
		&=
		f_{U_{(k-j+1)}, U_{(k)}}(x,t)
		\\
		&=
		k!\frac{x^{k-j}}{(k-j)!}\frac{(t-x)^{j-2}}{(j-2)!},
	\end{align*}
	where $U_{(j)}$ denotes $j$-th order statistics from $k$ i.i.d. random variables with $U[0,1]$.
	
	The density function of $X_{i,(1)}$ is p.d.f. of $Beta(k,1)$, which is $t^{k-1}/B(k,1)$ where $B(\cdot,\cdot)$ is the beta function.
	This implies that
	\begin{align*}
		f_{X_{i,(j)}|X_{i,(1)}=t}(x) = 	k!\frac{x^{k-j}}{(k-j)!}\frac{(t-x)^{j-2}}{(j-2)!}\frac{B(k,1)}{t^{k-1}}
	\end{align*}

	Hence, we have
	\begin{align*}
		g_j(t)
		=
		\int_{x=0}^t \parans{1-\frac{1}{2}\parans{\frac{x}{t}}^{k(n-1)}}\frac{k!}{(k-j)!(j-2)!}\frac{B(k,1)}{t^{k-1}}x^{k-j}(t-x)^{j-2}dx.
	\end{align*}
	By change of variables with $x = ty$, and we have
	\begin{align}
		g_j(t)
		&=
		\int_{y=0}^1 \parans{1-\frac{y^{k(n-1)}}{2}}
		\frac{k!}{(k-j)!(j-2)!}\frac{B(k,1)}{t^{k-1}}(ty)^{k-j}(t-ty)^{j-2}tdy
		\nonumber
            \\
		&=
		\frac{k!B(k,1)}{(k-j)!(j-2)!}
		\int_{y=0}^1
		\parans{1-\frac{y^{k(n-1)}}{2}}y^{k-j}(1-y)^{j-2}dy,
		\nonumber
	\end{align}
    where we use the fact that $B(x,y) = \frac{(x-1)!(y-1)!}{(x+y-1)!}$ for positive integers $x,y$,
    Rewriting the RHS using the definition of beta function, we can further expand
    \begin{align}
            g_j(t)=
		\frac{(k-1)!}{(k-j)!(j-2)!}\parans{B(k-j+1,j-1) - \frac{B(kn-j+1,j-1)}{2}}
		&=
		1 - \frac{1}{2} \frac{(kn-j)!(k-1)!}{(kn-1)!(k-j)!}
		\nonumber
            \\
		&=
		1 - \frac{1}{2} \frac{(k-j+1)\ldots (k-1)}{(kn-j+1)\ldots(kn-1)}
		\nonumber
            \\
		&\ge
		1-\frac{1}{2}\parans{\frac{k-1}{kn-1}}^{j-1},\label{ineq:01212139}
    \end{align}
where we use that $B(x,y) = \int_{t=0}^1 t^{x-1}(1-t)^{y-1}dt$ for any real $x,y$.
For the inequality, we use $\frac{k-j+1}{kn-j+1} \le \frac{k-1}{kn-1}$ for any $j =2,3,\ldots, k$.
This completes the proof.
\end{proof}

Due to the claims above (Claim~\ref{cl:4},\ref{cl:5} and \ref{cl:6}) and by \eqref{ineq:01222241}, we obtain the following lower bound on the event $E_i$.
	\begin{align*}
            \Pr{E_i}
            &=
            \Ex{\prod_{j=2}^k \parans{1-\frac{A_{i,j}}{2}}}
            \\
            &\overset{(a)}{\ge}
            \int_{t=0}^1\parans{\prod_{j=2}^k g_j(t)}\Pr{X_{i,(1)} =t}dt
            \\
		&\overset{(b)}{\ge}
		\int_{t=0}^1  \parans{\prod_{j=2}^k 1-\frac{1}{2}\parans{\frac{k-1}{kn-1}}^{j-1}}\Pr{X_{i,(1)}=t}dt
		\\
		&=
		 \prod_{j=2}^k \parans{1-\frac{1}{2}\parans{\frac{k-1}{kn-1}}^{j-1}}
		\\
		&\ge
		 \prod_{j=2}^k \parans{1-\frac{1}{2}\parans{\frac{1}{n}}^{j-1}},
	\end{align*}
        where $(a)$ holds by Claim~\ref{cl:5}  and $(b)$ follows from Claim~\ref{cl:6}.
	
 Using $1-x \ge e^{-x/(1-x)}$ for $x \in (0,1)$, we finally have 
	\begin{align}
		\Pr{E(\omb_i)}
		&\ge
		\exp\parans{-\sum_{j=2}^k \frac{\parans{\frac{1}{n}}^{j-1}}{2-\parans{\frac{1}{n}}^{j-1}}}.\label{ineq:05031951}
	\end{align}
        Now we proceed with the following series of algebraic manipulation.
        \begin{align*}
            \eqref{ineq:05031951}
            &\ge
		\exp \parans{-\sum_{j=2}^k \frac{1}{2n^{j-1}-1}}
		\\
		&\ge
		\exp \parans{-\sum_{j=2}^k \frac{1}{2}\frac{1}{n^{j-1}-1}}
		\\
		&=
		\exp \parans{-\frac{1}{2(n-1)}\sum_{j=2}^k \frac{1}{n^{j-2}+n^{j-3}+\ldots +1}}
		\\
		&\ge
		\exp \parans{-\frac{1}{2(n-1)}\sum_{j=2}^k \frac{1}{n^{j-2}}}
		\\
		&=
		\exp \parans{-\frac{1}{2(n-1)} \frac{1-\parans{\frac{1}{n}}^{k-1}}{1-\frac{1}{n}}}
		\\
		&\ge
		\exp\parans{-\frac{n}{2(n-1)^2}}.
        \end{align*}
	Multiplying it over $i \in [n]$, we have
	\begin{align}
		\Pr{E} \ge \exp\parans{-\frac{n^2}{2(n-1)^2}},
		\label{ineq:01062049}
	\end{align}
	where it converges to $1/\sqrt{e} \approx 0.6065...$ as $n$ increases.
	Hence, we conclude that playing $\sigma_i^x$ is $(1 - e^{-n^2/2(n-1)^2})$-approximate BNE, and as the principal's utility is $X_{\max}$ for any observation of solutions under this equilibrium, we complete the proof. 
\end{proof}

\subsection{Proof of Theorem \ref{theorem:multi-agent-limited-budget}}
\label{pf_theorem:multi-agent-limited-budget}
\begin{proof}
Due to our definition of RSPM, $B$ of proposed candidates will be randomly chosen by the principal and the best candidate among them will be selected as the winner.
Let $E$ be the set of agents who have at least one eligible solution to propose, and $E'$ be the set of other agents. 
Denote the size of sets $E$ and $E'$ by $e$ and $e'$ respectively so that $e + e' = n$.
Let $\omb = \cup_{i \in [n]}\omb_i$ be the entire set of solutions observed by the agents and $\sigma = (\sigma_1,\ldots, \sigma_n)$ be arbitrary equilibrium strategy.
For $i \in E'$,  since agent $i$ does not have eligible solutions,
even his best-observed solution which gives the utility of $X_{i, (1)}$ to the principal is not eligible.
This implies that $X_{i, (1)} < \tau$. 
On the other hand, for $i \in E$, agent $i$ has at least one eligible observed solution to propose which
implies that his best-observed solution is also eligible as well, \ie $X_{i, (1)} \geq \tau$.
These arguments further imply that $X_{(e)} \geq \tau > X_{(e + 1)}$.
We consider two cases (i) when $e \le B$ and (ii) when $e > B$, and prove the theorem separately. 

{Case (i): $e \leq B$.} 
In this case, the number of agents who have eligible solutions is at most $B$, which implies that $X_{(B+1)} < \tau$. 
In this case, there exist at most $B$ proposed candidates thus the principal can examine all of them and take the best one. 
As a result, this case is equivalent to the unlimited budget we had in Theorem \ref{theorem:multi-agent-unlimited-budget},
and hence from the same argument, the principal's expected utility is at least
\begin{align*}
    \Ex{X_{(2)}\IND\left[X_{(2)} \geq \tau > X_{(B+1)}\right]} + 
    \Ex{\tau\IND\left[X_{(1)} \geq \tau > X_{(2)}\right]}
    .
\end{align*}


{Case (ii): $e > B$.} 
We follow the proof steps of Theorem~\ref{theorem:multi-agent-unlimited-budget}. Assume that agents play mixed strategies, \ie they construct a probability distribution over the observed solutions.
Formally, given the solutions $\omb$, denote by $p_{i, j}$, the probability that agent $i$ submits $\omega_{i, j}$ as his candidate.
These probability distributions shape a mixed Nash equilibrium.

For each $i \in [n]$ we have $\sum_{j=1}^{k_i} p_{i, j} = 1$.
Define $S_i = \{\omega_{i, j} : p_{i, j} > 0\}$ as the multi-set of solutions that agent $i$ proposes with positive probability.
For each agent $i$, let $b_i$ be the candidate which gives the worst utility to the principal, \ie $b_i = \argmin_{\omega \in S_i} x(\omega$).
If there are multiple candidates with the worst utility, we arbitrarily choose one.
We use $A$ to denote the set of agents who have non-zero utility.

We proceed with the following claim.
\begin{claim}
    $|A| \le e-B+1$.
\end{claim}
\begin{proof}[Proof of the claim]
Suppose not, then $|A| \ge e-B+2$.
Let $t \in A$ be an agent such that for each agent $i \in A \setminus \{t\}$, either
\begin{itemize}
    \item $x(b_t) < x(b_i)$
    \item $x(b_t) = x(b_i)$, and mechanism prefers agent $i$ over agent $t$ using tie-breaking rule $\rho$,
\end{itemize}
\ie agent $t$ is the agent with non-zero utility having worst $b_t$ and least preferred using $\rho$.
As he has non-zero utility and $b_t \in S_i$, proposing $b_t$ gives him non-zero utility as well.

According to the definition of $b_t$,
when $t$ proposes $b_t$, no matter what agents in $A \setminus \{t\}$ propose,
their candidates are dominant, \ie
the mechanism prefers candidate submitted by any agent in $A \setminus \{t\}$ over $b_t$. 

There are $e$ candidates in total and as $|A \setminus \{t\}| \ge e-B+1$, any random choice of $B$ candidates includes at least one candidate of agents in $A \setminus \{t\}$.
This implies that $b_t$ is never selected as the winner,
which is a contradiction with the fact that proposing $b_t$ gives non-zero utility to agent $t$.
So, we conclude that $|A| \le e - B + 1$.
\end{proof}

Now, we focus on agents in $E \setminus A$.
Let $t' \in E \setminus A$ be an agent such that for any agent $j \in E \setminus (A \cup \{t'\})$, either
\begin{itemize}
    \item $X_{t', (1)} > X_{j, (1)}$
    \item $X_{t', (1)} = X_{j, (1)}$, and mechanism prefers agent $t'$ over agent $j$ using tie-breaking rule $\rho$,
\end{itemize}
\ie agent $t'$ has the first-best solution among agents in $E \setminus A$ and
mostly preferred by tie-breaking rule $\rho$.
Note that we proved $|A| \le e - B + 1$ so $|E \setminus A| \ge B - 1$.
Now we claim $X_{t', (1)} \le x(b_t)$ as follows.
\begin{claim}
    $X_{t', (1)} \le x(b_t)$.
\end{claim}
\begin{proof}[Proof of the claim]
Suppose not and let $x(b_t) < X_{t', (1)}$.
We know that under Nash equilibrium, agent $t'$ gets zero utility thus given the strategies of others, no matter what his strategy is, he gets zero utility. 
As a result, even if he proposes his first-best solution that gives the utility of $X_{t', (1)}$ to the principal, his solution is never selected.
However, according to the definition of $t'$, this solution is always preferred to
all candidates proposed by agents in $E \setminus \left(A \cup \{t'\}\right)$ which are at least $B - 2$ candidates. It is also dominant over $b_t$.
Hence due to the randomness of the mechanism,
there exists a non-zero probability that the principal takes these $B - 1$ candidates
and the candidate proposed by $t'$ so that $t'$ becomes the winner and gets non-zero utility.
This is a contradiction to the fact that $t'$ has zero utility and thus $X_{t', (1)} \le x(b_t)$.
\end{proof}

Let $\omega^*$ be the winner. $\omega^*$ belongs to agents in $A$ and according to definition of $b_t$, $x(b_t) \le x(\omega^*)$.
We also proved that $X_{t', (1)} \le x(b_t)$ so principal's utility is at least $X_{t', (1)}$. 
We recall that
\begin{itemize}
    \item 
    According to definitions of $E$ and $E'$,
    $X_{t', (1)} \ge \tau > X_{i, (1)}$ for all $i \in E'$. 
    \item
    According to the definition of $t'$, we have
    $X_{t', (1)} \ge X_{j, (1)}$ for each $j \in E \setminus \left(A \cup \{t'\}\right)$.
\end{itemize}
 
Hence $X_{t', (1)}$ is at least as large as the first-best solutions of $|E'| + |E \setminus \left(A \cup \{t'\}\right)| \geq n - e + B - 2$ agents.
Therefore, we can conclude that $X_{t', (1)} \geq X_{(e - B + 2)}$.
So, in this case where $X_{(e)} \geq \tau > X_{(e+1)}$, we obtain a lower bound on principal's utility as
\begin{align*}
\Ex{X_{(e-B+2)}\IND\left[X_{(e)} \geq \tau > X_{(e+1)}\right]}.
\end{align*}

Now by iterating over all values for $e > B$, we state that the expected utility is at least
\begin{align*}
\sum_{e=B+1}^{n} \Ex{X_{(e-B+2)}\IND\left[X_{(e)} \geq \tau > X_{(e+1)}\right]},
\end{align*}
where we write $X_{(n + 1)} := -\infty$. 
By combining the result of case (i), and case (ii), we can infer that the principal's expected utility is at least
\begin{align*}
\Ex{X_{(2)}\IND\left[X_{(2)} \geq \tau > X_{(B+1)}\right]} +  \Ex{\tau\IND\left[X_{(1)} \geq \tau > X_{(2)}\right]}
+
\sum_{e=B+1}^{n} \Ex{X_{(e-B+2)}\IND\left[X_{(e)} \geq \tau > X_{(e+1)}\right]},
\end{align*}
which shows that the additive approximation factor is at most
\begin{align*}
\Ex{X_{(1)}}
-
\Ex{X_{(2)}\IND\left[X_{(2)} \geq \tau > X_{(B+1)}\right]} -  \Ex{\tau\IND\left[X_{(1)} \geq \tau > X_{(2)}\right]} \\
-
\sum_{e=B+1}^{n} \Ex{X_{(e-B+2)}\IND\left[X_{(e)} \geq \tau > X_{(e+1)}\right]}
\end{align*}
which finishes the proof.

\end{proof}

\subsection{Proof of Theorem \ref{thm:bce}}\label{pf_thm:bce}
\begin{proof}
        We recommend the readers skim the proof of Theorem~\ref{thm:bayes_oblivious} in Appendix~\ref{pf_thm:bayes_oblivious} before reading this proof.
        Similar to the proof of Theorem~\ref{thm:bayes_oblivious}, we proceed with three main steps.
        First, we construct a Bayes correlated equilibrium with a certain choice of signaling scheme in which we signal the agents whenever some good event $E$ happens.
        Next, we show that the principal's utility given the signal can be decomposed as $\Ex{X_{\max}} \cdot \Pr{E}$.
        Finally, we obtain a lower bound on the probability that $E$ happens.

        \paragraph{Constructing Bayes correlated equilibrium}
	Consider an MSPM with no eligible sets and a deterministic tie-breaking rule $\rho$.
	Denote by $\omb = \cup_{i \in [n]}\omb_i$ the solutions observed by the agents.
	Similar to the proof of Theorem~\ref{thm:bayes_oblivious}, define an event $E_i(\cdot)$ as follows.
	\begin{align*}
		E_i 
		=
		\paranm{x(\omb_i), y(\omb_i) : \argmax_{\om \in \omb_i} x(\om)^{k(n-1)}y(\om) = \argmax_{\om \in \omb_i} x(\om)}.
	\end{align*}
	Let $x(A) = \{x(\om)\}_{\om \in A}$ and similarly for $y(A)$.
	Define $E = \cap_{i \in [n]} E_i$.
	Let $\sigma_i^x$ be a strategy to propose a solution that maximizes $x(\cdot)$ for agent $i$.
	Denote by $u_i(\sigma_i,\sigma_{-i} | A)$ agent $i$'s expected utility conditioned on the event $A$, \ie
	\begin{align*}
		u_i(\sigma_i,\sigma_{-i} | A)
		=
		\Ex{\int_{\omb \in A}y_i(f_{M,(\sigma_i, \sigma_{-i})})f(\omb | A)d\omb},
	\end{align*}
	where we write $f(\omb | A)$ to denote the probability density function of $\omb$ conditioned on $A$.
	Conditioned on $E_i$, agent $i$'s expected payoff of proposing a solution $\om \in \omb_i$ given that others play $\sigma_{-i}^x$ is
	\begin{align*}
		u_i(\om, \sigma_{-i}^x|E_i)
		&=
		\Pr{\om \text{ is winner} \given E}\cdot y(\om)
		\\
		&=
		\Pr{x(\om) \ge x(\om'), \forall\om' \in \omb_{-i}\given E}\cdot y(\om)
		\\
		&= x(\om)^{k(n-1)}y(\om).
	\end{align*}
	This implies that conditioned on $E_i$, agent $i$'s best response given that others play $\sigma_{-i}^x$ is also $\sigma_i^x$.
	Hence, conditioned on $E(\omb)$, playing $\sigma^x = (\sigma_1^x, \ldots, \sigma_n^x)$ is best-response for each agent given the others' strategies.
	
	Observe that due to the Nash's existence theorem, there exists at least one BNE.
	Denote this equilibrium by $\sigma^o= (\sigma^o_1, \ldots, \sigma^o_n)$.
	Now consider a mediator who recommends $\psi_i = \sigma_i^x$ to agent $i \in [n]$ whenever $E(\omb)$ happens, and recommends $\sigma'_i$ otherwise.
	Due to our construction of $E$, we have
	\begin{align*}
		u_i(\psi_i, \sigma_{-i}^x | E)
		\ge 
		u_i(\sigma_i, \sigma_{-i}^x | E),	
	\end{align*}
	for any other strategy $\sigma'_i$ of agent $i$.
	From our definition of $\sigma^o$, we also have
	\begin{align*}
		u_i(\psi_i, \sigma^o_{-i} | \neg E)
		\ge
		u_i(\sigma'_i, \sigma^o_{-i} | \neg E),
	\end{align*}
	for any other strategy $\sigma'_i$ of agent $i$.
	Combining two inequalities above, we can verify that obeying $\psi = (\psi_1, \ldots, \psi_n)$ constructs a BNE as follows
	\begin{align*}
		u_i(\psi_i, \psi_{-i})
		&=
		u_i(\psi_i, \sigma_{-i}^x | E) \Pr{E} + u_i(\psi_i, \sigma_{-i}^o | \neg E) \Pr{\neg E}
		\\
		&=
		u_i(\psi_i, \sigma_{-i}^x | E) \Pr{E} + u_i(\psi_i, \sigma_{-i}^o | \neg E) \Pr{\neg E}
		\\
		&\ge
		u_i(\sigma'_i, \sigma_{-i}^x | E) \Pr{E} + u_i(\sigma'_i, \sigma_{-i}^o | \neg E) \Pr{\neg E}.
	\end{align*}

        \paragraph{Decomposing the principal's utility}
	For notational simplicity, denote by $x(\sigma)$ the principal's expected utility given the strategies $\sigma$, and $x(\sigma | \omb)$ the principal's utility given $\omb$, \ie given all the random variables $x(\omb)$ and $y(\omb)$.
	Denote by $f_x(x(\omb))$ and $f_y(y(\omb))$ the p.d.f. of $x(\omb)$ and $y(\omb)$, respectively.
	Now, under the BCE with $\psi$, suppose that the agents obey the recommended strategies, \ie $\sigma = \psi$.
	The principal's expected utility $x(\sigma)$ can be lower bounded as
	\begin{align}
		&\int_{x(\omb)}\left( \int_{y(\omb)} \Ex{x(\sigma  | \omb) \given x(\omb), y(\omb)} f_y(y(\omb))dy(\omb) \right) f_x(x(\omb))dx(\omb)
		\nonumber
		\\
		&\overset{(a)}{=} 
		\int_{x(\omb)}\left( \int_{y(\omb)} \Ex{x(\sigma  | \omb) \given x(\omb), y(\omb), E}\Pr{E \given x(\omb), y(\omb)}f_y(y(\omb))dy(\omb) \right) 
		\nonumber
		\\
		&+ \left( \int_{y(\omb)} \Ex{x(\sigma  | \omb) \given x(\omb), y(\omb), \neg E} \Pr{\neg E \given x(\omb), y(\omb)}f_y(y(\omb))dy(\omb) \right)  f_x(x(\omb))dx(\omb)
		\nonumber
		\\
		&\overset{(b)}{\ge}
		\int_{x(\omb)}\left( \int_{y(\omb)} \Ex{x(\sigma  | \omb) \given x(\omb), y(\omb), E} \Pr{E \given x(\omb), y(\omb)}f_y(y(\omb))dy(\omb) \right) f_x(x(\omb))dx(\omb),
		\nonumber
	\end{align}
	where $(a)$ holds by the law of total expectation, and $(b)$ follows from ignoring the latter integral, \ie the one conditioned on $\neg E$.
	First, we observe that conditioned on the event $E$, the principal's utility $x(\sigma | \omb)$ is simply $X_{\max}$.
	This is because conditioned on $E$, each agent $i$ submits $\sigma_i^x$, and the principal finally chooses the one that maximizes her own utility.
	By the independence of $X_{\max}$ over the random variables $y(\cdot)$, we can further obtain
	\begin{align}
		x(\sigma )
		&\ge
		\int_{x(\omb)}\left( \int_{y(\omb)}\Pr{E \given x(\omb), y(\omb)}f_y(y(\omb))dy(\omb) \right)   X_{\max} f_x(x(\omb))dx(\omb).
		\\
		&=
		\int_{x(\omb)}\prod_{i=1}^n\left(\int_{y_i(\omb_i)}\Pr{E \given x(\omb), y_i(\omb)}f_{y_i}(y_i(\omb))dy_i(\omb) \right)   X_{\max} f_x(x(\omb))dx(\omb),
		\label{ineq:12251801}
	\end{align}
	where the last equality follows from the independence of the solutions between the agents.
	Defining $A_{i,j} = \parans{\frac{X_{i,(j)}}{X_{i,(1)}}}^{k(n-1)}$, due to \eqref{ineq:01061949}, we can rewrite \eqref{ineq:12251801} as
	\begin{align}
		\int_{x(\omb)}\parans{\prod_{i=1}^n\prod_{j=2}^k\parans{1 - \frac{A_{i,j}}{2}}}  X_{\max} f_x(x(\omb))dx(\omb)
		&=
		\Ex{X_{\max} \prod_{i=1}^n \prod_{j=2}^k \parans{1 - \frac{A_{i,j}}{2}}}
		\nonumber
		\\
		&=
		\Ex{ \Ex{t\prod_{i=1}^n \prod_{j=2}^k \parans{1 - \frac{A_{i,j}}{2}} \given X_{\max} = t} }.
		\\
		&=
		\int_{t=0}^{1}\Ex{\prod_{i=1}^n \prod_{j=2}^k \parans{1 - \frac{A_{i,j}}{2}} \given X_{\max} = t}t \Pr{X_{\max} = t}dt
		\label{ineq:01062038}
	\end{align}
	Now we prove the following claim.
	\begin{claim}
		$g(t) = \Ex{\prod_{i=1}^n \prod_{j=2}^k \parans{1 - \frac{A_{i,j}}{2}} \given X_{\max} = t}$ is an increasing function on $t$.
	\end{claim}
	\begin{proof}[Proof of the claim]
		We can write $g(t)$ as
		\begin{align*}
			\sum_{s=1}^n\Ex{\prod_{i=1}^n \prod_{j=2}^k \parans{1 - \frac{A_{i,j}}{2}} \given X_{s,(1)} = t, \argmax_{i \in [n]} X_{i,(1)} = s}\Pr{\argmax_{i \in [n]} X_{i,(1)} = s \given X_{\max} = t}.
		\end{align*}
		Due to symmetry, we know $\Pr{\argmax_{i \in [n]} X_{i,(1)} = s \given X_{\max} = t}$ is exactly equal for $s \in [n]$, and it does not change with respect to $t$.
		In addition, given $X_{s,(1)} = t $ and $\argmax_{i \in [n]} X_{i,(1)} = s$, in $\prod_{i=1}^n \prod_{j=2}^k \parans{1 - \frac{A_{i,j}}{2}}$, only the terms $\prod_{j=2}^k\parans{1 - \frac{A_{s,j}}{2}}$ changes with respect to $t$.
		Given $X_{s,(1)} = t $, we can rewrite it as
		\begin{align*}
			\prod_{j=2}^k\parans{1 - \frac{A_{s,j}}{2}} = \prod_{j=2}^k \parans{1- \parans{\frac{X_{s,(j)}}{t}}^{k(n-1)}},
		\end{align*}
		and it is obviously an increasing function with respect to $t$.
		Combining the results, we conclude that $g(t)$ is an increasing function on $t$.
	\end{proof}
        Rewriting the expectation based on $X_{i,j}$ and applying FKG inequality similar to the proof of Theorem~\ref{thm:bayes_oblivious}, we obtain
	\begin{align*}
		\eqref{ineq:01062038}
		&\ge
		\parans{\int_{t=0}^1 \Ex{\prod_{i=1}^n \prod_{j=2}^k \parans{1 - \frac{A_{i,j}}{2}} \given X_{\max} = t} \Pr{X_{\max} = t} t} \cdot \parans{ \int_{t=0}^1 t \Pr{X_{\max} = t} dt}
		\\
		&=
		\Ex{\prod_{i=1}^n \prod_{j=2}^k \parans{1 - \frac{A_{i,j}}{2}}} \Ex{X_{\max}}
		\\
		&=
		\Pr{E}\cdot \Ex{X_{\max}}.
	\end{align*}

        \paragraph{Obtaining a lower bound on $\Pr{E}$}
        Now it remains to lower bound the probability $\Pr{E}$, which is already given in \eqref{ineq:01062049}.
	Hence, we conclude that the principal's expected utility under this BCE is at least
	\begin{align*}
		\exp\parans{- \frac{n^2}{2(n-1)^2}}\Ex{X_{\max}},
	\end{align*}
	and it completes the proof.
\end{proof}

\subsection{Proof of Theorem \ref{theorem:multi-agent-correspondence}}\label{pf_theorem:multi-agent-correspondence}
\begin{proof}
We start by presenting the following lemma from~\citet{kleinberg2018delegated}.
\begin{lemma}[\citet{kleinberg2018delegated}]\label{thm:singleagent_revel}
In the single-agent setting, for any mechanism $M$ and corresponding dominant strategies $\sigma$, there exists a MSPM such that for any equilibrium strategies $\sigma'$, we have $x(f_{M,\sigma}(\omb)) = x(f_{SPM,\sigma'}(\omb))$ for any $\omb \in \Omega^*$.
\end{lemma}
By the lemma, it suffices to consider the case when $M$ is a SPM.
Let the realized solutions in $P$ be $\omb = \{\om_1, \om_2,\ldots, \ldots \om_k\}$. Consider a partition of $\omb$ as follows
\begin{align*}
	\omb_{1} &= \{\om_1,\ldots, \om_{k_1}\}
	\\
	\omb_{2} &= \{\om_{k_1+1},\ldots, \om_{k_1+k_2}\}
	\\
	&\vdots
	\\
	\omb_{n} &= \{\om_{k_1+\ldots+k_{n-1}+1},\ldots, \om_{k_1+\ldots + k_n=k}\}.
\end{align*}
Suppose that agent $i$'s realization of solutions in $P'$ is $\omb_{i}$ so that the outcomes of $P$ and $P'$ are coupled. 
This coupling is indeed possible because the probability that agent $i$'s solutions are realized to be $\omb_i = \{\om_{i_1},\ldots, \om_{i_{k_i}}\}$ is exactly equal to the marginal probability that $\om$'s $i$-th partition is realized to be $\omb_i = \{\om_{i_1},\ldots, \om_{i_{k_i}}\}$.

Consider a SPM with an eligible set $R$ under $P$, and corresponding equilibrium strategies $\sigma$.
Denote by $x$ and $y$ the principal and the agent's utility function under $P$ respectively.
Given the solutions realized as above, SPM will select the solution which maximizes $y(\om_j)$ for $j \in [k]$ such that $\om_j \in R$, \ie $\om_1^* = \argmax_{j \in [n], \om_j \in R} y(\om_j)$.
Let this solution be $\om_1^*$.
Now consider a MSPM with arbitrary tie-breaking rule and the same eligible set $R$ under $P'$, and suppose that the agents play equilibrium strategies $\sigma'$.
Let $\om_2^*$ be the corresponding winner in this case.

We want to show that $x(\om_2^*) \ge x(\om_1^*)$  for any realization of solutions $\omb$.
Suppose not.
It is obvious that both $x(\om_1^*)$ and $x(\om_2^*)$ belong to the eligible set $R$.
If $\om_1^*$ and $\om_2^*$ belongs to the same partition $\omb_i$, it means that agent $i$ has both the solutions $\om_1^*$ and $\om_2^*$.
By our definition on $\om_1^*$ such that $\om_1^* = \argmax_{j \in [n], \om_j \in R} y(\om_j)$, submitting $\om_1^*$ is gives the better or the same utility for agent $i$.
If $y(\om_1^*) > y(\om_2^*)$, then simply submitting $\om_1^*$ gives a strictly better utility for agent $i$, and it contradicts that $\sigma'$ is equilibrium strategy.
If $y(\om_1^*) = y(\om_2^*)$, then $x(\om_1^*) \le x(\om_2^*)$ by Assumption~\ref{as:pareto}.
Hence, $x(\om_1^*) \le x(\om_2^*)$ in both cases.

Now suppose that $\om_1^*$ and $\om_2^*$ belongs to the different partitions $\omb_i$ and $\omb_j$.
Suppose that agent $i$'s equilibrium strategy satisfies $\sigma'_i \neq \om_1^*$.
If playing $\omb_i$ is strictly better for agent $i$, it contradicts that $\sigma'_i$ is an equilibrium strategy.
Hence $y(\om_1^*) = y(\sigma_i')$. Again by Assumption~\ref{as:pareto}, since agent $i$ plays $\sigma_i'$, we have $x(\sigma'_i) \ge x(\om_1^*)$.
Note that the principal should observe both $\sigma_i'$ and $\om_2^*$, and commits to $\om_2^*$.
This implies that $x(\om_2^*) \ge x(\sigma_i') \ge x(\om_1^*)$, and we finish the proof.
\end{proof}

\end{document}